\documentclass[twocolumn]{article}
\usepackage[english]{babel}
\usepackage[dvipsnames]{xcolor}
\usepackage{amsmath,amssymb,enumerate,latexsym,theorem}
\usepackage[a4paper, total={7in, 9in}]{geometry}
\usepackage{tikz}
\usepackage{comment}
\usepackage{subcaption}
\usepackage{graphicx}
\usepackage[noadjust]{cite}
\usepackage{wrapfig}
\usepackage{float}
\usepackage[toc]{appendix}
\usepackage{titling}
\usepackage{enumitem}
\usepackage{lipsum}  
\usepackage{hyperref}
\usepackage{abstract}
\usepackage{tcolorbox}
\tcbuselibrary{listings, breakable, skins}
\usepackage{glossaries}
\usepackage{ulem}

\newcommand{\tmstrong}[1]{\textbf{#1}}

\newcommand{\assign}{:=}

\newcommand{\backassign}{=:}

\newcommand{\RNum}[1]{\uppercase\expandafter{\romannumeral #1\relax}}
\usepackage{tcolorbox}
\tcbuselibrary{listings, breakable, skins}
\usepackage{xcolor}

\newcommand{\ra}{\rightarrow}

\newcommand{\T}[1]{\text{#1}} % Shortcut for plain text in math mode

\renewcommand{\i}{\infty}

\definecolor{lasallegreen}{rgb}{0.03, 0.47, 0.19}

\definecolor{myboxbackground}{HTML}{fdf7c9}
\definecolor{myboxborder}{HTML}{fdcf50}
\definecolor{myboxtitle}{HTML}{000000}
\newtcolorbox{myglossarybox}{
  enhanced,
  colback=myboxbackground, % Background color
  colframe=myboxborder, % Border color
  boxrule=0.5mm, % Border width
  arc=0mm, % Round corners
  title={Glossary}, % Box title
  fonttitle=\bfseries\color{myboxtitle}\normalsize,
  fontupper=\footnotesize,
}

\newtheorem{theorem}{Theorem}

\newtheorem{definition}[theorem]{Definition}
\newtheorem{proposition}[theorem]{Proposition}
\newtheorem{lemma}[theorem]{Lemma}
\newtheorem{corollary}[theorem]{Corollary}
\newtheorem{remark}[theorem]{Remark}

\numberwithin{theorem}{section}

\newenvironment{proof}{\noindent\textbf{Proof\ }}{\hspace*{\fill}$\Box$\medskip}
\newenvironment{itemizedot}{\begin{itemize} }{\end{itemize}}

\begin{document}
\title{Geometry and stability of species complexes}
\author{\tmstrong{Amaury Lambert$^1$, Emmanuel Schertzer$^2$, Yannic Wenzel$^2$}}
\date{%
    {\small $^1$Stochastic Models for the Inference of Life Evolution (SMILE),
Institute of Biology of ENS (IBENS), CNRS, INSERM, Université PSL, Ecole Normale Supérieure, 46 rue d’Ulm, 75005 Paris, France \\
\& \\
Center for Interdisciplinary Research in Biology (CIRB), CNRS, INSERM, Université PSL, Collège de France, 11 place Marcelin Berthelot, 75005 Paris, France\\%
    $^2$Faculty of Mathematics, University of Vienna, Oskar-Morgenstern-Platz 1, 1090 Wien, Austria\\[2ex]}%
    \today
}

\maketitle

\begin{abstract}
  Species complexes are groups of closely related populations exchanging genes through dispersal.
  We study the dynamics of the
   structure of species complexes in a class of metapopulation models where demes can exchange genetic material through migration and diverge through the accumulation of new mutations. Importantly, we model the ecological feedback of differentiation on gene flow by assuming that the success of migrations decreases with genetic distance, through a specific function $h$.
   We investigate the effects of  metapopulation size on the coherence of species structures, depending on some mathematical characteristics of the feedback function $h$. Our results suggest that with larger metapopulation sizes, species form increasingly coherent, transitive, and uniform entities.  We conclude that the initiation of speciation events in large species requires the existence of idiosyncratic geographic or selective restrictions on gene flow.
  \vspace{0.5cm}
\end{abstract}

% \twocolumn
\section{Introduction}

Speciation is the process by which diverging populations become reproductively isolated from each other, preventing them from interbreeding or ensuring that hybrid offspring are inviable or sterile. The development of reproductive isolation (RI) relies on the accumulation of reproductive isolating barriers, i.e., the biological features that impede gene exchange between populations (see \cite{coyne2004speciation}, p.29).  If this accumulation leads to complete reproductive isolation, we speak of different species (see \cite{coyne2004speciation}, p.26ff).

In general, we distinguish modes of speciation by the extent to which geographic conditions impede gene flow. In perfect geographic segregation and zero gene flow (allopatry), the divergent accumulation of different mutations leads to failure of outcrossing at a secondary contact. Under geographic conditions allowing for limited gene flow (parapatry), a combination of forces including natural and sexual selection can lead to the evolution of reproductive barriers between migrating individuals (see \cite{coyne2004speciation}, Sections 3 and 4).

\begin{figure}[t]
\begin{myglossarybox}
\textbf{\hypertarget{species_complex}{Species complex:}} a set of populations connected through direct or indirect (i.e., through intermediary populations) gene exchange.

\vspace{0.2cm}
\textbf{\hypertarget{feedback}{Divergence feedback:}} the negative relationship between genetic distance and effective migration rate whereby lowering genetic similarity between two populations reduces gene flow between them, further reducing their genetic similarity.

\vspace{0.2cm}
\textbf{\hypertarget{genetic_incomp}{Genetic incompatibilities:}} post-zygotic reproductive barriers that lead to inviability, sterility or other types of fitness reduction in hybrids.

\vspace{0.2cm}
\textbf{\hypertarget{feedback_func}{Feedback function:}} a nondecreasing, continuous function $h$, which encodes how effective migration rate varies with genetic similarity.

\vspace{0.2cm}
\textbf{\hypertarget{trans}{Transitivity:}} an ideal property of some species complexes ensuring that for any three populations $i,j,k$ such that $i$ can interbreed with $j$, and $j$ can interbreed with $k$, then $i$ can interbreed with $k$.

\vspace{0.2cm}
\textbf{\hypertarget{clust}{Subspecies clustering:}} a property of some species complexes that occurs when populations can be partitioned into clusters of genetically similar populations (subspecies) showing reduced genetic exchange between them (partial reproductive isolation).

\vspace{0.2cm}
\textbf{\hypertarget{reversibility}{Irreversibility:}} a natural property of genomes ensuring that interfertility cannot be re-established after complete reproductive isolation has been built up.

\vspace{0.2cm}
\textbf{\hypertarget{neutral}{Neutrality:}}
Neutrality here refers to the assumption that no selection is acting on genes other than that resulting in reduced effective migration between populations that are genetically distant (e.g. hybrid depression). 
\end{myglossarybox}
\end{figure}

Although it has been suggested that they may be quite common in nature (see \cite{coyne2004speciation}, p.111ff, \cite{nosil2008speciation}), processes of speciation with gene flow seem to have received relatively little attention in evolutionary modeling compared to allopatric speciation (see \cite{gavrilets2014models}, p.748). %\yw{According to Gavrilets and Coyne and Orr, only parapatric speciation has received less attention, not sympatric speciation}. 
Recently, a new class of general speciation models started gaining popularity: a population- or individual-based framework, in which the degree of divergence between spatially dispersed groups of organisms is measured by their genetic distance (see \cite{gavrilets2014models}, p.745ff for a review). Within this class of models, diversity between populations arises from mutations (increasing genetic distance), while homogeneity arises from migrations between populations (decreasing genetic distance). The increase in genetic distance following mutation events is based on the infinite-allele assumption that each mutation at a locus results in an allele of a novel type; the decrease of genetic distance following migration events is due to the fixation of part of the migrant genome in a resident population.

In most of these models (see for instance \cite{higgs1991stochastic,manzo1994geographic,gavrilets1998rapid}), individuals migrate between populations at a constant rate, independent of genetic distance (exceptions including for instance \cite{gavrilets2000waiting}, for parapatric speciation between two populations). Once sufficient divergence has taken place, the classification as a new species is usually defined by the crossing of a predefined critical threshold of genetic distance between populations. By exceeding this threshold, the degree of reproductive isolation between the affected populations is typically assumed to jump from no isolation to complete isolation.

In this paper, we present a simple stochastic ``genetic distance" model in which the emergence of complete reproductive isolation occurs gradually, as a natural consequence of the interaction between gene flow and genetic distance between populations exposed to migration. In fact, through the coupling of migration rates to genetic distance, speciation results from an initial perturbation resulting in an increase in genetic distance, causing effective migration to decrease, which tends to increase genetic distance further, and so on. One can think of this dynamic as a positive feedback loop, which causes divergent populations to naturally snowball into complete reproductive isolation. We establish a general framework for the study of species complexes that is suitable to describe the emergence and stability of interbreeding structures ranging from ideal, transitive complexes to ring species or \hyperlink{clust}{subspecies clusterings}. 

The integration of this \hyperlink{feedback}{feedback effect} into the model through the function $h$, which encodes the translation of genetic distances into effective migration rates, raises some intriguing questions: Can we link characteristics of species complexes, such as \hyperlink{trans}{transitivity}, \hyperlink{clust}{clustering}, or stability, to analytical properties of the function $h$? Between geographic migration restrictions and the shape of the feedback function, which force has a stronger influence on the shape of large species complexes? How does the shape of the species complex depend on the number of populations? And finally, can we infer information about quantities related to speciation, such as the distribution of time to first speciation, or the average number of new species upon speciation from the structure of a species complex?

\section{Model description} \label{chap:model}

In this section, we present the idea of the model, the underlying biological assumptions and its mathematical implementation.

\paragraph{Evolutionary divergence feedback.}

\noindent The central idea of the model is to understand speciation as a consequence of a self-sustaining interaction between effective migration rates and the genetic differences between populations connected by migration. Here, we use the term ``effective migration rate'' to refer to the rate at which an individual migrates from one population to another, and fixes part of its genetic material in the arrival population. As alluded to above, the coupling of effective migration rates to genetic proximity can cause speciation by an initial decrease in genetic proximity (due to mutation) causing effective migration rates to decrease, which tends to decrease genetic proximity further, and so on. We will refer to this dynamic as \textbf{\hyperlink{feedback}{divergence feedback}}.

The term ``genetic differences between populations'' is intentionally kept broad, in order to encompass different theoretical views of speciation genomics. For instance, these differences could refer to different alleles at ``speciation genes'' (see \cite{nosil2011genes} for a precise definition and review of this term). The number of these ``speciation genes'' can be as little as two, or reach into the hundreds, depending on the species one considers (see \cite{coyne2004speciation}, p.302).

Another interpretation of the genetic differences between populations is the net synonymous divergence, i.e. the number of single nucleotide substitutions at synonymous sites (i.e., contained in noncoding regions or leaving unchanged the amino acid sequence produced). Data from different animal populations/species (see \cite{roux2016shedding} and, for instance, Fig. 3 therein) indicate that the net synonymous divergence between populations is a good predictor of the degree of reproductive isolation between populations. This fact makes this interpretation especially appealing from an application point of view, because synonymous substitutions are much easier to quantitatively determine than different alleles in speciation genes (see for instance \cite{nosil2011genes}). 

Two populations undergoing differentiation are said to be in the gray zone of speciation if they are not sufficiently differentiated to form distinct species but already too different to be identified as one single species. 
Analysis of empirical data in \cite{roux2016shedding} shows that this region of fuzzy species boundaries is relatively narrow in the logarithmic scale of genetic distances. To quantify this isolation gradient, we introduce a function $h$ that takes as input the genetic similarity between any given pair of populations, and returns the acceptance probability of migrants between them. We will denote this function by $h$, and refer to it as the \textbf{feedback function}.

We emphasize that measures of effective migration exist, and can serve as a good predictor for the shape of the feedback function $h$. As alluded to above, the authors of \cite{roux2016shedding} estimate the probability of ongoing gene flow between pairs of populations as a function of their divergence at synonymous sites, from observed genomic data (see, for instance, \cite{roux2016shedding} Fig. 3). The results indicate that across various animal species and populations, the probability of ongoing gene flow shows a consistent pattern of increase from values below 0.2 to values above 0.8 as genomic distance decreases (and genomic similarity increases) in the same critical region of distances around $0.01$. The feedback function $h$ can be thought of as encapsulating the shape and speed of this increase. 

By coupling the effective migration rate to the genetic proximity of two populations, we can understand speciation as the diverse process it is understood to be. Speciation is neither always a sudden, nor always a gradual process. 
Examples from nature can be found at either end of the spectrum, see \cite{nosil2017tipping, coyne2004speciation}. 
However, most speciation models (see for instance \cite{gavrilets1998rapid, pina2019does}) focusing on the genetic distance between populations rely implicitly on the assumption that the function $h$ is a Heaviside step function, equal to zero below some predefined threshold (complete reproductive isolation) and equal to one above the threshold (free interbreeding). In this framework, we stress that there is no feedback between differentiation and reproductive isolation: as long as genetic proximities are above the threshold, the effective migration rates stay unchanged. Once the genetic proximity between two populations falls below this level, reproductive isolation is complete and the frequency of migration events can go to zero in one fell swoop. As mentioned above, effective migration rates are known to exhibit different behaviors (see for example \cite{nosil2017tipping, roux2016shedding}), which motivates the incorporation of a feedback function that allows expressing different strengths of reduction in effective migration rates associated to genetic divergence.

\paragraph{Technical assumptions.}

Our aim is to build a model that keeps track of genetic diversity at $L$ loci of interest, across a metapopulation made of $N$ island-like populations.

Our first assumption is the absence of intra-population polymorphism, at the genes under consideration. To ensure that this property holds after mutation or migration events, we assume that the time between the appearance and loss/fixation of an allele is significantly shorter than the waiting time between two events. Thus, one conventionally ignores the short phases during which the population is polymorphic due to multiple segregating alleles at a given locus (see \cite{mccandlish2014modeling, patwa2008fixation} for reviews).

Genetic diversity across populations emerges from the interplay of mutation and migration: mutation events increase genetic diversity, while migration events tend to homogenize gene pools. 
Let us be more specific about the assumptions we make here. Since populations are thought of as monomorphic at all times, we need only consider the mutation and migration events that are followed by fixation of the novel/alien variant. 

First, we define a mutation event as a substitution, i.e., a mutation followed by the fixation (assumed ``instantaneous'') of the novel allele. In the realm of neutral theory, the substitution rate at a neutral locus equals the mutation rate per individual at this locus (see \cite{kimura1962probability}). Second, we understand a migration event as the migration of an individual followed by the fixation (assumed ``instantaneous'') of a fraction of its genome into the target genome.

We will further make the simplifying assumption that migration events always result in fixation at only one locus, i.e., replacement of the allele present at a single locus in the target genome by the allele present at the homologous locus of the migrating genome. In order to justify this assumption, we first note that if recombination rates are high enough, this will cause substantial fragmentation of the mutant genome and break genetic correlations. Then, after a few generations, linkage disequilibrium becomes very small, and we can expect alleles to fix independently. 
Under a neutrality assumption for the $L$ loci, the number of migrant alleles fixing in a population of size $n$ is thus given by a Binomial random variable $B_n$ with parameters $L$ and $\frac{1}{n}$. Hence, if $n\gg L$, then the probability $\mathbb{P} (B_n=1\,|\,B_n\not=0)$ that only one allele fixes conditional on fixation of at least one allele, goes to 1.
Finally, note that our assumption of fixation at a single locus is mainly made out of mathematical convenience and that our model could be easily adapted to multi-locus fixations, but at the cost of analytical tractability.

The last and most important assumption is divergence feedback, namely, as seen previously, that the likelihood of a successful migration between two populations decreases with their genetic distance. Let us now introduce the model.

\begin{figure}
    
\centering

\scalebox{0.8}{

\tikzset{every picture/.style={line width=0.75pt}} %set default line width to 0.75pt        

\begin{tikzpicture}[x=0.75pt,y=0.75pt,yscale=-1,xscale=1]
%uncomment if require: \path (0,548); %set diagram left start at 0, and has height of 548

%Shape: Circle [id:dp3545314846320323] 
\draw  [fill={rgb, 255:red, 184; green, 233; blue, 134 }  ,fill opacity=1 ] (76.75,428) .. controls (76.75,423.31) and (80.56,419.5) .. (85.25,419.5) .. controls (89.94,419.5) and (93.75,423.31) .. (93.75,428) .. controls (93.75,432.69) and (89.94,436.5) .. (85.25,436.5) .. controls (80.56,436.5) and (76.75,432.69) .. (76.75,428) -- cycle ;
%Shape: Circle [id:dp17885812737061935] 
\draw  [fill={rgb, 255:red, 208; green, 2; blue, 27 }  ,fill opacity=1 ] (76.75,394) .. controls (76.75,389.31) and (80.56,385.5) .. (85.25,385.5) .. controls (89.94,385.5) and (93.75,389.31) .. (93.75,394) .. controls (93.75,398.69) and (89.94,402.5) .. (85.25,402.5) .. controls (80.56,402.5) and (76.75,398.69) .. (76.75,394) -- cycle ;
%Shape: Circle [id:dp6543106702052773] 
\draw  [fill={rgb, 255:red, 245; green, 166; blue, 35 }  ,fill opacity=1 ] (76.75,411) .. controls (76.75,406.31) and (80.56,402.5) .. (85.25,402.5) .. controls (89.94,402.5) and (93.75,406.31) .. (93.75,411) .. controls (93.75,415.69) and (89.94,419.5) .. (85.25,419.5) .. controls (80.56,419.5) and (76.75,415.69) .. (76.75,411) -- cycle ;
%Shape: Circle [id:dp0920747826253635] 
\draw  [fill={rgb, 255:red, 248; green, 231; blue, 28 }  ,fill opacity=1 ] (120.75,349) .. controls (120.75,344.31) and (124.56,340.5) .. (129.25,340.5) .. controls (133.94,340.5) and (137.75,344.31) .. (137.75,349) .. controls (137.75,353.69) and (133.94,357.5) .. (129.25,357.5) .. controls (124.56,357.5) and (120.75,353.69) .. (120.75,349) -- cycle ;
%Shape: Circle [id:dp6204776953339409] 
\draw  [fill={rgb, 255:red, 208; green, 2; blue, 27 }  ,fill opacity=1 ] (120.75,315) .. controls (120.75,310.31) and (124.56,306.5) .. (129.25,306.5) .. controls (133.94,306.5) and (137.75,310.31) .. (137.75,315) .. controls (137.75,319.69) and (133.94,323.5) .. (129.25,323.5) .. controls (124.56,323.5) and (120.75,319.69) .. (120.75,315) -- cycle ;
%Shape: Circle [id:dp11921796624765313] 
\draw   (96,332) .. controls (96,313.64) and (110.89,298.75) .. (129.25,298.75) .. controls (147.61,298.75) and (162.5,313.64) .. (162.5,332) .. controls (162.5,350.36) and (147.61,365.25) .. (129.25,365.25) .. controls (110.89,365.25) and (96,350.36) .. (96,332) -- cycle ;
%Shape: Circle [id:dp9017005698763623] 
\draw  [fill={rgb, 255:red, 80; green, 227; blue, 194 }  ,fill opacity=1 ] (120.75,332) .. controls (120.75,327.31) and (124.56,323.5) .. (129.25,323.5) .. controls (133.94,323.5) and (137.75,327.31) .. (137.75,332) .. controls (137.75,336.69) and (133.94,340.5) .. (129.25,340.5) .. controls (124.56,340.5) and (120.75,336.69) .. (120.75,332) -- cycle ;
%Right Arrow [id:dp785074918954023] 
\draw   (209.67,361.12) -- (235.37,361.12) -- (235.37,353.82) -- (252.5,368.41) -- (235.37,383) -- (235.37,375.71) -- (209.67,375.71) -- cycle ;
%Shape: Circle [id:dp8529727847204907] 
\draw   (114.25,349) .. controls (114.25,340.72) and (120.97,334) .. (129.25,334) .. controls (137.53,334) and (144.25,340.72) .. (144.25,349) .. controls (144.25,357.28) and (137.53,364) .. (129.25,364) .. controls (120.97,364) and (114.25,357.28) .. (114.25,349) -- cycle ;
%Curve Lines [id:da7841148598995448] 
\draw    (101.5,431) .. controls (157.65,437.9) and (171.1,404.04) .. (145.69,369.57) ;
\draw [shift={(144.5,368)}, rotate = 52.35] [color={rgb, 255:red, 0; green, 0; blue, 0 }  ][line width=0.75]    (10.93,-3.29) .. controls (6.95,-1.4) and (3.31,-0.3) .. (0,0) .. controls (3.31,0.3) and (6.95,1.4) .. (10.93,3.29)   ;
%Shape: Circle [id:dp6065493410857921] 
\draw   (52,411) .. controls (52,392.64) and (66.89,377.75) .. (85.25,377.75) .. controls (103.61,377.75) and (118.5,392.64) .. (118.5,411) .. controls (118.5,429.36) and (103.61,444.25) .. (85.25,444.25) .. controls (66.89,444.25) and (52,429.36) .. (52,411) -- cycle ;
%Shape: Circle [id:dp5857385405339562] 
\draw  [fill={rgb, 255:red, 184; green, 233; blue, 134 }  ,fill opacity=1 ] (317.75,429) .. controls (317.75,424.31) and (321.56,420.5) .. (326.25,420.5) .. controls (330.94,420.5) and (334.75,424.31) .. (334.75,429) .. controls (334.75,433.69) and (330.94,437.5) .. (326.25,437.5) .. controls (321.56,437.5) and (317.75,433.69) .. (317.75,429) -- cycle ;
%Shape: Circle [id:dp5685972370786367] 
\draw  [fill={rgb, 255:red, 208; green, 2; blue, 27 }  ,fill opacity=1 ] (317.75,395) .. controls (317.75,390.31) and (321.56,386.5) .. (326.25,386.5) .. controls (330.94,386.5) and (334.75,390.31) .. (334.75,395) .. controls (334.75,399.69) and (330.94,403.5) .. (326.25,403.5) .. controls (321.56,403.5) and (317.75,399.69) .. (317.75,395) -- cycle ;
%Shape: Circle [id:dp4935630998779651] 
\draw  [fill={rgb, 255:red, 245; green, 166; blue, 35 }  ,fill opacity=1 ] (317.75,412) .. controls (317.75,407.31) and (321.56,403.5) .. (326.25,403.5) .. controls (330.94,403.5) and (334.75,407.31) .. (334.75,412) .. controls (334.75,416.69) and (330.94,420.5) .. (326.25,420.5) .. controls (321.56,420.5) and (317.75,416.69) .. (317.75,412) -- cycle ;
%Shape: Circle [id:dp9153014589218107] 
\draw  [fill={rgb, 255:red, 184; green, 233; blue, 134 }  ,fill opacity=1 ] (357.75,349) .. controls (357.75,344.31) and (361.56,340.5) .. (366.25,340.5) .. controls (370.94,340.5) and (374.75,344.31) .. (374.75,349) .. controls (374.75,353.69) and (370.94,357.5) .. (366.25,357.5) .. controls (361.56,357.5) and (357.75,353.69) .. (357.75,349) -- cycle ;
%Shape: Circle [id:dp8004735583833645] 
\draw  [fill={rgb, 255:red, 208; green, 2; blue, 27 }  ,fill opacity=1 ] (357.75,315) .. controls (357.75,310.31) and (361.56,306.5) .. (366.25,306.5) .. controls (370.94,306.5) and (374.75,310.31) .. (374.75,315) .. controls (374.75,319.69) and (370.94,323.5) .. (366.25,323.5) .. controls (361.56,323.5) and (357.75,319.69) .. (357.75,315) -- cycle ;
%Shape: Circle [id:dp6333591566446721] 
\draw   (333,332) .. controls (333,313.64) and (347.89,298.75) .. (366.25,298.75) .. controls (384.61,298.75) and (399.5,313.64) .. (399.5,332) .. controls (399.5,350.36) and (384.61,365.25) .. (366.25,365.25) .. controls (347.89,365.25) and (333,350.36) .. (333,332) -- cycle ;
%Shape: Circle [id:dp8192790496666432] 
\draw  [fill={rgb, 255:red, 80; green, 227; blue, 194 }  ,fill opacity=1 ] (357.75,332) .. controls (357.75,327.31) and (361.56,323.5) .. (366.25,323.5) .. controls (370.94,323.5) and (374.75,327.31) .. (374.75,332) .. controls (374.75,336.69) and (370.94,340.5) .. (366.25,340.5) .. controls (361.56,340.5) and (357.75,336.69) .. (357.75,332) -- cycle ;
%Shape: Circle [id:dp33824905165138397] 
\draw   (293,412) .. controls (293,393.64) and (307.89,378.75) .. (326.25,378.75) .. controls (344.61,378.75) and (359.5,393.64) .. (359.5,412) .. controls (359.5,430.36) and (344.61,445.25) .. (326.25,445.25) .. controls (307.89,445.25) and (293,430.36) .. (293,412) -- cycle ;

% Text Node
\draw (32,388) node [anchor=north west][inner sep=0.75pt]   [align=left] {1};
% Text Node
\draw (159,292) node [anchor=north west][inner sep=0.75pt]   [align=left] {2};
% Text Node
\draw (273,389) node [anchor=north west][inner sep=0.75pt]   [align=left] {1};
% Text Node
\draw (396,292) node [anchor=north west][inner sep=0.75pt]   [align=left] {2};

\end{tikzpicture}
}
\captionsetup{font=small}
\caption{Toy realisation of the model and a migration event. Here, $N = 2, L = 3$, and the migration event occurs from population 1 to 2, affecting locus 3. The genetic proximity between 1 and 2 changes from $P_{12}=1/3$ to $P_{12} = 2/3$.}
\label{fig:Model_trans}
\end{figure}

\paragraph{The model.}

\noindent Consider a metapopulation comprising $N$ populations. Each population is monomorphic and thus can be considered as harboring one single genome with $L$ loci. In the following, lower case letters represent the populations and upper case letters the loci. We will represent the state of the metapopulation at time $t$ by a matrix of allelic types $A(t) := (A_{i, K}(t))_{1\le i\le N, 1\le K\le L}$, where $A_{i,K}$ represents the allelic type in population $i$ at locus $K$.
The dynamics between the $N$ populations will depend on a coupling factor between the loci. This coupling is enforced through the \tmstrong{genetic proximities}, defined between any populations $i$ and $j$ by
\begin{align}
    P^L_{i j}(t) \assign \frac{1}{L} \sum_{K = 1}^L \mathbf{1}_{\{ A_{i,K}(t) = A_{j,K}(t)\}}\, . \label{eq:def_genprox}
\end{align}
Here, the notation $\mathbf{1}_{\{ A_{i,K}(t) = A_{j,K}(t)\}}$ is defined through 
\begin{align*}
\mathbf{1}_{\{ A_{i,K}(t) = A_{j,K}(t)\}} = 
    \begin{cases} 
      1 & \text{if }A_{i,K}(t) = A_{j,K}(t) \\
      0 & \text{otherwise}
   \end{cases}.
\end{align*}
In words, the genetic proximity between $i$ and $j$ is the fraction of loci at which populations $i$ and $j$ currently carry the same allele. 

\noindent The model depends on the following parameters:

\begin{itemize}
    \item the mutation rate $\mu > 0$,
    \item the migration matrix $(M_{i j})$, where $M_{ii}=0$ and $M_{i j}\geq 0$ are the natural migration rates, reflecting the topology and ecology of the metapopulation in the absence of divergence feedback,
    \item and the feedback function $h$, that is a nondecreasing, continuous function on $[0,1]$, verifying $h(0)=0$ and  $h(1)=1$
\end{itemize}

In each population $i$ and at each locus $K$, \textbf{mutation} events occur at rate $\mu$. Any lineage $(i,K)$ experiencing a mutation event takes on a new type (infinite-allele model). At any time $t\geq 0$, between each pair of populations $i$ and $j$, and at each locus $K$, \textbf{migration} events from $i$ to $j$ occur at rate 
\begin{align}
    M_{i j}h(P^L_{i j}(t)) \, ,\label{eq:mig_rate}
\end{align}
called \textbf{effective migration rates}.
In the type matrix, this amounts to replacing the allele of $(j,K)$ by the allele of $(i,K)$, see Fig. \ref{fig:Model_trans}.

\begin{figure*}[t]
    \centering
    \includegraphics[width=\textwidth]{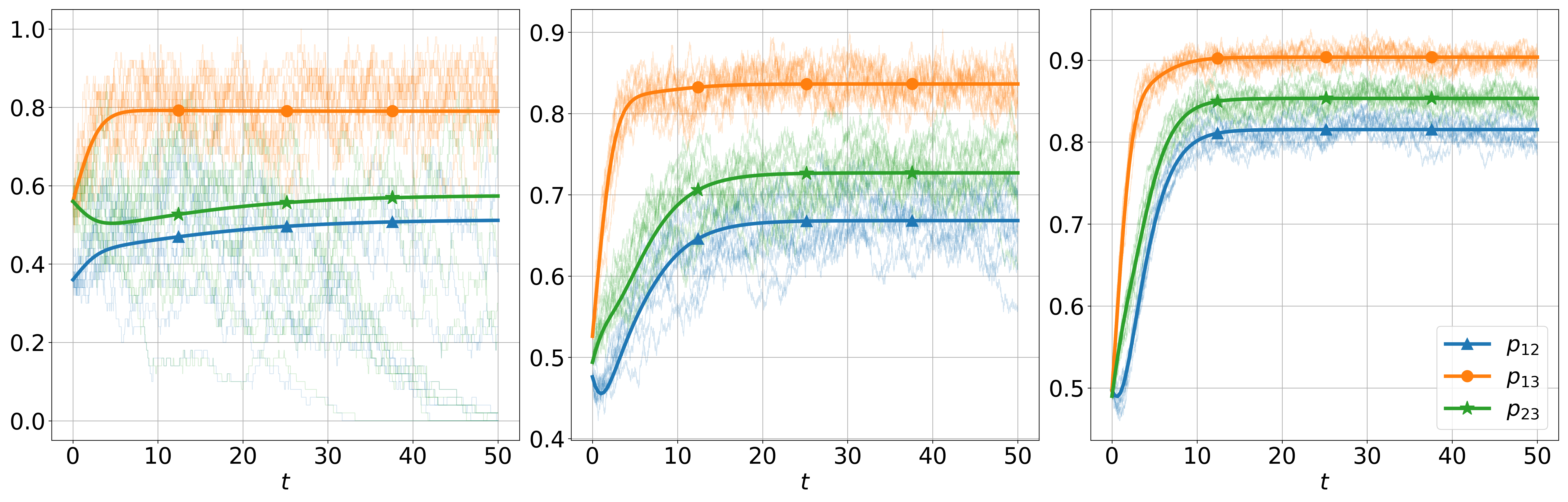}
    \captionsetup{font=small}
    \caption{Convergence of the stochastic genetic proximities to the solution of the ODE for 3 populations as the number $L$ of loci gets large. The strong, solid lines are numerically simulated solutions to the ODE (\ref{eq:ODE}). The transparent lines are simulations of the stochastic model for different numbers of loci, namely $L = 50, 500, 1000$ from left to right. Additionally, we varied mutation rates, namely $\mu = 0.1, 0.08, 0.05$ from left to right, while keeping the migration matrix constant: $M=((0, 0.1, 0.8), (0.1, 0, 0.5), (0.8, 0.5, 0))$.}
    \label{fig:conv_to_ODE}
\end{figure*}

We note that after a mutation event, the genetic proximity between the affected population $i$ (at some locus $K$) and every other population $j$ decreases by $1/L$, if $i$ did not already carry a different allele than $j$ at locus $K$ prior to the mutation event. Furthermore, after a migration event from $i$ to $j$ (at some locus $K$), the genetic proximity between $i$ and $j$ increases by $1/L$ if $i$ and $j$  carried different alleles prior to the migration event.

\section{Mathematical analysis}
\label{chap:ODE_Dual}
Here, we introduce two mathematical tools that considerably simplify the analysis when the number of loci gets large: 
\begin{itemize}
    \item 
First we will show that instead of keeping track of the allele identity at each locus in each population, the dynamics can be described by a system of ODE's following the fraction of shared alleles (genetic proximity) between each pair of populations;
\item Second, a time reversal (duality) approach allows us to express as a fixed-point problem the genetic proximity $P_{ij}$ between populations $i$ and $j$, as the probability that two random walks started at $i$ and $j$ respectively and moving according to (dual) effective migration rates, which themselves depend on all genetic proximities, coalesce before a mutation occurs.
\end{itemize}

\paragraph{ODE approximation.}
We describe how our stochastic model can be approximated by the solution to an ordinary differential equation (ODE), when the number of loci is sufficiently large. This result will allow us to examine the evolution of the genetic proximities over time in a deterministic context, and thus analytically study the evolution of reproductive isolation in our model.

More specifically, we will illustrate that the genetic proximities $(P^L_{i j}(t))_{1\le i,j\le N}$ in our stochastic model can be approximated, as $L$ gets large, by a continuous, deterministic function $P(t):=(P_{i j}(t))_{1\le i,j\le N}$, solution to the non-linear differential equation
\begin{eqnarray*}
\dot{P}_{i j} & = & \sum_{k = 1}^N (M_{k i} h (P_{k i}) P_{k j}+ M_{k j} h (P_{k j}) P_{k i})\\
    &  - & P_{i j} \left( \sum_{k = 1}^N (M_{k i} h (P_{k i}) + M_{k j} h (P_{j
    k})) + 2 \mu \right) \, ,
  \end{eqnarray*}
for all $i\neq j$. This will be written shortly as 
\begin{align}
    \dot{P}(t) = \vec F(P(t))\, . \label{eq:ODE}
\end{align}
Note that for $N=2$ and$M_{12}=M_{21} = m$, this ODE becomes
\begin{align}
    \dot p = 2mh(p)(1-p)-2\mu p \, . \label{eq:ODE_dim2}
\end{align}
In Fig. \ref{fig:conv_to_ODE}, we illustrate the convergence for large $L$ of the stochastic genetic proximities to the solution of the ODE with simulations. 

We now give a brief heuristics for the system of equations (\ref{eq:ODE}) and refer to the SI \ref{chap:master_eq} for a rigorous derivation. 

Recall from the previous section that $A_{i,K}(t)$ is the allelic type at locus $K$, in population $i$ at time $t$. To gain some intuition, we start by assuming that $h\equiv 1$, so that the effective migration rates are not impacted by genetic distances (absence of feedback). In this setting, the allelic composition 
at each locus $K$ 
\begin{equation}
\label{eqn:def_A_K}
    A_K(t) := (A_{1,K}(t),\dots,A_{N,K}(t))
    \end{equation}
evolves independently, according to a Moran model on a weighted graph. That is, each population is thought of as an individual; new mutations arise at rate $\mu$ and ``individual" $j$ takes on the type of ``individual" $i$ at rate $M_{i j}$. In particular, when $M_{i j} = m$ for all $i\neq j$, this process corresponds to the standard Moran process, see \cite{etheridge2011some}.

How does changing $h$ to a non-trivial feedback function influence the model? If $h$ is not constant, the previous representation remains valid under an important adaptation: the reproduction rate $M_{i j}$ in the case $h\equiv 1$ needs to be replaced by $M_{i j} h(P^L_{i j})$. The resulting allelic processes $A_1,\dots,A_L$ defined by \eqref{eqn:def_A_K} are now coupled through the genetic proximities $P^L_{i j}$ given by \eqref{eq:def_genprox}.

For small values of $L$, this induces a strong interaction between loci. However, for a large number of loci, the interactions between any pair of loci should become negligible. Thus, under the premise that the loci are asymptotically uncorrelated, we can apply the law of large numbers to obtain the convergence of $P^L_{i j}(t)$ to a deterministic quantity. 

This limit, which we will denote by $P_{i j}(t)$, describes the coupling between the allelic processes $A_1,\dots,A_L$, when the number of loci is large. Furthermore, all the limiting allelic processes should be identically distributed, since the property holds true for finite $L$. Let 
\begin{equation}
\label{eqn:def_cal_A}
{\cal A}(t) := ({\cal A}_{1}(t),\dots,{\cal A}_{N}(t))
\end{equation}
be the limiting allelic process. Intuitively, we think of ${\cal A}_i(t)$ as the allelic type at a ``typical" locus, in population $i$ at time $t$.

The representation of $P^L_{i j}$ in equation \eqref{eq:def_genprox} gives an interpretation of the limiting $P_{i j}$ in terms of the allelic process ${\cal A}$, i.e.,  
\begin{eqnarray}
P_{ij}(t) & = &\lim_{L\to\infty} \frac{1}{L} \sum_{K=1}^L \mathbf{1}_{\{A_{i,K}(t) = A_{j,K}(t)\}}\nonumber\\
& = & \mathbb{P}({\cal A}_i(t) = {\cal A}_j(t)),
\label{eq:p_ij_prob}
\end{eqnarray}
where in the last line, we used the law of large numbers.
In other words, $P_{ij}(t)$ is the probability that $i$ and $j$ have the same type at time $t$ in the Moran model ${\cal A}$ describing the dynamics at a ``typical" locus.

Can we provide a description of the dynamics of the limiting allelic process $\mathcal{A}$? To deduce the reproduction rates, we recall that for finite $L$, the rate at which $j$ takes on the type of $i$ is $M_{ij}h(P^L_{ij}(t))$, 
which by continuity of $h$ converges to $M_{ij} h(P_{ij}(t))$.

We thus obtain a single-locus Moran representation of our stochastic model via the process ${\cal A}$, whose dynamics are given as follows.
 For each ``individual'' $i$, mutations occur at rate $\mu$ and give rise to a novel allele. At any time $t$, reproduction events from $i$ to $j$ occur, that is, the individual $j$ takes on the type of $i$ at rate $$M_{i j} h\left(\mathbb{P}(\mathcal{A}_i(t) = \mathcal{A}_j(t))\right).$$

This process ${\cal A}$ is an example of a $nonlinear$ Markov process, characterized by the dependence of the transition probabilities not only on the state, but also on the law of the process itself (here, only the probabilities that ``individuals'' $i$ and $j$ carry the same allele). The term $nonlinear$ represents the non-linearity in the Chapman-Kolmogorov equation, that the transition probabilities of the Markov process satisfy. We will call ${\cal A}$ a $nonlinear$ Moran process. 

Crucially, the nonlinear Moran process ${\cal A}$ allows us to express the deterministic genetic proximities $P_{i j}$ as the solution to a system of ODEs. This property can be seen by the ``backward'' representation of the Moran process thanks to a duality approach.  

\paragraph{Duality approach.}

To gain some intuition, consider the process $\cal A$ at equilibrium,  i.e., when the quantities $P_{i j}(t)$ have attained their equilibrium state $P^\text{eq}_{i j}$.  In this case, the process ${\cal A}$ corresponds to a Moran process on a weighted graph. We consider its graphical representation on $\{1,\dots, N\}\times \mathbb{R}_+$ (see \cite{etheridge2011some}):
\begin{itemize}
  \item For a reproductive event from vertex $i$ to vertex $j$ at time $t$, draw an arrow with origin at $(i,t)$ and tip at $(j,t)$ 
  \item For a mutation event at vertex $k$ at time $t$, draw a $\star$ at  $(k,t)$. 
 \end{itemize}

Let us now consider the population at a reference time $T$. Via this graphical representation (see Fig. \ref{fig:MoranDual}), we can associate to every vertex an ancestral lineage carrying its allele using the arrow-star configuration. Then, the system of ancestral lineages is distributed like random walks on a graph: they evolve independently until they coalesce, jumping from site $i$ to $j$ at rate $M_{j i}  h(P^\text{eq}_{j i})$. Each lineage is killed upon encountering a mutation ($\star$). This is because once an ancestral lineage encounters a mutation, the allele it carries has no further ascent.

By (\ref{eq:p_ij_prob}), the quantity $P^\text{eq}_{i j}$ can be computed as the probability that $i$ and $j$ are of the same type. This occurs if and only if the ancestral lineages starting from $i$ and $j$ coalesce before being killed. Since the transition rates themselves depend on the genetic proximities, we obtain that $P^\text{eq}_{i j}$ can be computed by solving a fixed point problem. More formally, define the coalescing time
\[ T_{ij} := \inf\{ u > 0 : S^{i}(u) = S^{j} (u) \}, \] where $S^{i}, S^{j}$ are the ancestral lineages starting from site $(i,T)$ and $(j,T)$.
We note that the law of $T_{i j}$ depends on $P^\text{eq} = (P^\text{eq}_{ij})_{i,j}$ through the jump rates of the ancestral lineages, we will thus write $T_{i j} = T_{i j}(P^\text{eq})$. According to the previous argument, the matrix of genetic proximities $P^\text{eq}$ satisfies the \tmstrong{fixed point problem}
    \begin{equation}
\forall i\neq j%\in [N],   
\qquad   P^\text{eq}_{ij} =\mathbb{E} \left( e^{- 2 \mu T_{i j}(P^\text{eq})} \right) \, ,
      \label{eq:FxPtPb}
    \end{equation}
see Theorem (\ref{thm:FixPtPb}) in SI \ref{sec:eq_stab}.

\begin{figure}
	\centering
\scalebox{0.5}{

\tikzset{every picture/.style={line width=0.75pt}} %set default line width to 0.75pt        

\begin{tikzpicture}[x=0.75pt,y=0.75pt,yscale=-1,xscale=1]
%uncomment if require: \path (0,464); %set diagram left start at 0, and has height of 464

%Straight Lines [id:da08371921444165298] 
\draw    (121,40) -- (121,449.5) ;
%Straight Lines [id:da3658000009785445] 
\draw [color={rgb, 255:red, 245; green, 166; blue, 35 }  ,draw opacity=1 ]   (171,40) -- (171,449.5) ;
%Straight Lines [id:da7490915018422539] 
\draw [color={rgb, 255:red, 184; green, 233; blue, 134 }  ,draw opacity=1 ]   (221,40) -- (221,449.5) ;
%Straight Lines [id:da8488216675713511] 
\draw [color={rgb, 255:red, 208; green, 2; blue, 27 }  ,draw opacity=1 ] (270,40) -- (270,450.5) ;
%Straight Lines [id:da4592281668589975] 
\draw [color={rgb, 255:red, 144; green, 19; blue, 254 }  ,draw opacity=1 ]   (320,40) -- (320,450.5) ;
%Straight Lines [id:da66950457352696] 
\draw    (71,40) -- (71,447.5) ;
\draw [shift={(71,449.5)}, rotate = 270] [color={rgb, 255:red, 0; green, 0; blue, 0 }  ]   (10.93,-3.29) .. controls (6.95,-1.4) and (3.31,-0.3) .. (0,0) .. controls (3.31,0.3) and (6.95,1.4) .. (10.93,3.29)   ;
%Straight Lines [id:da3687904260837205] 
\draw    (403,449.5) -- (403,42) ;
\draw [shift={(403,40)}, rotate = 90] [color={rgb, 255:red, 0; green, 0; blue, 0 }  ] (10.93,-3.29) .. controls (6.95,-1.4) and (3.31,-0.3) .. (0,0) .. controls (3.31,0.3) and (6.95,1.4) .. (10.93,3.29)   ;
%Straight Lines [id:da2562756187227786] 
\draw [color={rgb, 255:red, 80; green, 227; blue, 194 }  ,draw opacity=1 ][fill={rgb, 255:red, 80; green, 227; blue, 194 }  ,fill opacity=1 ]   (171,40) -- (171,227.5) ;
%Straight Lines [id:da8541937793456047] 
\draw [color={rgb, 255:red, 155; green, 155; blue, 155 }  ,draw opacity=1 ]   (121,117) -- (263,117) ;
\draw [shift={(265,117)}, rotate = 180] [color={rgb, 255:red, 155; green, 155; blue, 155 }  ,draw opacity=1 ]   (10.93,-3.29) .. controls (6.95,-1.4) and (3.31,-0.3) .. (0,0) .. controls (3.31,0.3) and (6.95,1.4) .. (10.93,3.29)   ;
%Straight Lines [id:da9005854715282999] 
\draw [color={rgb, 255:red, 208; green, 2; blue, 27 }  ,draw opacity=1 ]   (121,40) -- (121,326.5) ;
%Straight Lines [id:da8365392582632334] 
%\draw [color={rgb, 255:red, 208; green, 2; blue, 27 }  ,draw opacity=1 ][line width=1]      (270,115.5) -- (270,450.5) ;
%Straight Lines [id:da07640399185132929] 
\draw [color={rgb, 255:red, 155; green, 155; blue, 155 }  ,draw opacity=1 ]   (171,326) -- (128,326) ;
\draw [shift={(125,326)}] [color={rgb, 255:red, 155; green, 155; blue, 155 }  ,draw opacity=1 ]   (10.93,-3.29) .. controls (6.95,-1.4) and (3.31,-0.3) .. (0,0) .. controls (3.31,0.3) and (6.95,1.4) .. (10.93,3.29)   ;
%Straight Lines [id:da5825927632631069] 
\draw [color={rgb, 255:red, 245; green, 166; blue, 35 }  ,draw opacity=1 ]   (121,326.5) -- (121,449.5) ;
%Straight Lines [id:da0414347109341302] 
\draw [color={rgb, 255:red, 80; green, 227; blue, 194 }  ,draw opacity=1 ][fill={rgb, 255:red, 80; green, 227; blue, 194 }  ,fill opacity=1 ]   (221,40) -- (221,292.5) ;
%Straight Lines [id:da9325405885557423]
\draw [color={rgb, 255:red, 144; green, 19; blue, 254 }   ,draw opacity=1 ][fill={rgb, 255:red, 80; green, 227; blue, 194 }  ,fill opacity=1 ] (270,40) -- (270,117) ;
% \draw [color={rgb, 255:red, 144; green, 19; blue, 254 }  ,draw opacity=1 ]   (270,40) -- (270,115.5) ;
%Straight Lines [id:da5039647096458217] 
\draw [color={rgb, 255:red, 155; green, 155; blue, 155 }  ,draw opacity=1 ]   (221,233) -- (313,233) ;
\draw [shift={(316,233)}, rotate = 180] [color={rgb, 255:red, 155; green, 155; blue, 155 }  ,draw opacity=1 ]    (10.93,-3.29) .. controls (6.95,-1.4) and (3.31,-0.3) .. (0,0) .. controls (3.31,0.3) and (6.95,1.4) .. (10.93,3.29)   ;
%Straight Lines [id:da5927574597443064] 
\draw [color={rgb, 255:red, 80; green, 227; blue, 194 }  ,draw opacity=1 ]   (320,232.5) -- (320,451.5) ;
%Straight Lines [id:da6899206219887832] 
\draw    (450,39) -- (450,448.5) ;
%Straight Lines [id:da058158059493857994] 
\draw [color={rgb, 255:red, 245; green, 166; blue, 35 }  ,draw opacity=1 ][line width=1.5]    (500,325.5) -- (500,448.5) ;
%Straight Lines [id:da9475454450291119] 
\draw [color={rgb, 255:red, 80; green, 227; blue, 194 }  ,draw opacity=1 ]   (550,39) -- (550,448.5) ;
%Straight Lines [id:da7991377342933721] 
%\draw   [line width=1.5]    (599,40) -- (599,449.5) ;
%Straight Lines [id:da4387731916894162] 
\draw [color={rgb, 255:red, 155; green, 155; blue, 155 }  ,draw opacity=1 ]   (649,40) -- (649,449.5) ;
%Straight Lines [id:da3169752275028752] 
\draw [color={rgb, 255:red, 155; green, 155; blue, 155 }  ,draw opacity=1 ][fill={rgb, 255:red, 245; green, 166; blue, 35 }  ,fill opacity=1 ]   (500,39) -- (500,226.5) ;
%Straight Lines [id:da2831480992214013] 
\draw [color={rgb, 255:red, 155; green, 155; blue, 155 }  ,draw opacity=1 ]   (599,117) -- (457,117) ;
\draw [shift={(454,117)}] [color={rgb, 255:red, 155; green, 155; blue, 155 }  ,draw opacity=1 ]   (10.93,-3.29) .. controls (6.95,-1.4) and (3.31,-0.3) .. (0,0) .. controls (3.31,0.3) and (6.95,1.4) .. (10.93,3.29)   ;
%Straight Lines [id:da33455470809031707] 
\draw [color={rgb, 255:red, 155; green, 155; blue, 155 }  ,draw opacity=1 ]   (450,39) -- (450,325.5) ;
%Straight Lines [id:da027428043476081854] 
\draw [color={rgb, 255:red, 208; green, 2; blue, 27 }  ,draw opacity=1 ][line width=1.5]    (599,117.5) -- (599,449.5) ;
%Straight Lines [id:da653916234979548] 
\draw [color={rgb, 255:red, 155; green, 155; blue, 155 }  ,draw opacity=1 ]   (450,325.5) -- (493,325.5) ;
\draw [shift={(495,325.5)}, rotate = 180] [color={rgb, 255:red, 155; green, 155; blue, 155 }  ,draw opacity=1 ][line width=0.75]    (10.93,-3.29) .. controls (6.95,-1.4) and (3.31,-0.3) .. (0,0) .. controls (3.31,0.3) and (6.95,1.4) .. (10.93,3.29)   ;
%Straight Lines [id:da020080881787851945] 
\draw [color={rgb, 255:red, 184; green, 233; blue, 134 }  ,draw opacity=1 ][line width=1.5]    (450,325.5) -- (450,448.5) ;
%Straight Lines [id:da2805884953558879] 
\draw [color={rgb, 255:red, 155; green, 155; blue, 155 }  ,draw opacity=1 ][fill={rgb, 255:red, 184; green, 233; blue, 134 }  ,fill opacity=1 ]   (550,39) -- (550,291.5) ;
%Straight Lines [id:da5462807161557952] 
\draw [color={rgb, 255:red, 155; green, 155; blue, 155 }  ,draw opacity=1 ]  (599,39) -- (599,117.5) ;
%Straight Lines [id:da6188316728913523] 
\draw [color={rgb, 255:red, 155; green, 155; blue, 155 }  ,draw opacity=1 ]   (649,231.5) -- (557,231.5) ;
\draw [shift={(554,231.5)}] [color={rgb, 255:red, 155; green, 155; blue, 155 }  ,draw opacity=1 ]   (10.93,-3.29) .. controls (6.95,-1.4) and (3.31,-0.3) .. (0,0) .. controls (3.31,0.3) and (6.95,1.4) .. (10.93,3.29)   ;
%Straight Lines [id:da6523065209899575] 
\draw [color={rgb, 255:red, 80; green, 227; blue, 194 }  ,draw opacity=1 ][line width=1.5]    (649,231.5) -- (649,450.5) ;
%Straight Lines [id:da2410705113445577] 
\draw [color={rgb, 255:red, 65; green, 117; blue, 5 }  ,draw opacity=1 ]   (320,122.5) -- (320,232.5) ;
%Straight Lines [id:da4859114769783377] 
\draw [color={rgb, 255:red, 184; green, 233; blue, 134 }  ,draw opacity=1 ][line width=1.5]    (500,226.5) -- (500,325.5) ;
%Straight Lines [id:da8141235791703854] 
\draw [color={rgb, 255:red, 80; green, 227; blue, 194 }  ,draw opacity=1 ][line width=1.5]    (550,39) -- (550,231.5) ;
%Straight Lines [id:da9497870650006801] 
\draw [color={rgb, 255:red, 208; green, 2; blue, 27 }  ,draw opacity=1 ][line width=1.5]    (450,39) -- (450,115.5) ;
%Straight Lines [id:da22842658739934185] 
\draw [color={rgb, 255:red, 144; green, 19; blue, 254 }  ,draw opacity=1 ][line width=1.5]    (550,291.5) -- (550,448.5) ;

% Text Node
\draw (41.58,265.55) node [anchor=north west][inner sep=0.75pt]  [rotate=-269.68] [align=left] {time};
% Text Node
\draw (373.58,265.55) node [anchor=north west][inner sep=0.75pt]  [rotate=-269.68] [align=left] {time};
% Text Node
\draw (116,19) node [anchor=north west][inner sep=0.75pt]   [align=left] {1};
% Text Node
\draw (168,19) node [anchor=north west][inner sep=0.75pt]   [align=left] {2};
% Text Node
\draw (217,19) node [anchor=north west][inner sep=0.75pt]   [align=left] {3};
% Text Node
\draw (265,19) node [anchor=north west][inner sep=0.75pt]   [align=left] {4};
% Text Node
\draw (314,19) node [anchor=north west][inner sep=0.75pt]   [align=left] {5};
% Text Node
\draw (30,44) node [anchor=north west][inner sep=0.75pt]   [align=left] {0};
% Text Node
\draw (35,427.5) node [anchor=north west][inner sep=0.75pt]   [align=left] {$\displaystyle t$};
% Text Node
\draw (168.05,216.35) node [anchor=north west][inner sep=0.75pt]  [rotate=-33.19] [align=left] {$\displaystyle \bigstar $};
% Text Node
\draw (218.05,279.35) node [anchor=north west][inner sep=0.75pt]  [rotate=-33.19] [align=left] {$\displaystyle \bigstar $};
% Text Node
\draw (317.05,107.35) node [anchor=north west][inner sep=0.75pt]  [rotate=-33.19] [align=left] {$\displaystyle \bigstar $};
% Text Node
\draw (445,18) node [anchor=north west][inner sep=0.75pt]   [align=left] {1};
% Text Node
\draw (497,18) node [anchor=north west][inner sep=0.75pt]   [align=left] {2};
% Text Node
\draw (546,18) node [anchor=north west][inner sep=0.75pt]   [align=left] {3};
% Text Node
\draw (594,18) node [anchor=north west][inner sep=0.75pt]   [align=left] {4};
% Text Node
\draw (643,18) node [anchor=north west][inner sep=0.75pt]   [align=left] {5};
% Text Node
\draw (497.05,215.35) node [anchor=north west][inner sep=0.75pt]  [rotate=-33.19] [align=left] {$\displaystyle \bigstar $};
% Text Node
\draw (547.05,278.35) node [anchor=north west][inner sep=0.75pt]  [rotate=-33.19] [align=left] {$\displaystyle \bigstar $};
% Text Node
\draw (646.05,106.35) node [anchor=north west][inner sep=0.75pt]  [rotate=-33.19] [align=left] {$\displaystyle \bigstar $};

\end{tikzpicture}

}
\captionsetup{font=small}
\caption{Realisation of the genetic partitions induced by the single-locus Moran model, and its dual for $N = 5$. On the left, colours represent genetic types, whereas on the right, colours represent ancestral lineages.}
\label{fig:MoranDual}
\end{figure}

If the metapopulation is not at equilibrium so that the $P_{i j}(t)$ now depend on time, the same argument applies, with the difference that the jump rates of the random walks become inhomogeneous in time. Using the same genealogical approach, we can compute the probability that two sites $i$ and $j$ have the same type at some instant $t\geq 0$ by tracing their ancestral lineages back in time, starting from $t$. This allows us to deduce that $P_{i j}(t)$ are solution to the differential equation (\ref{eq:ODE}). We refer to Proposition \ref{prop:DynDual} and Corollary \ref{cor:ODE} for details.

\bigskip

\section{A special case: two populations}
\label{sec:2pop}
To get some intuition about how the fixed-point equation (\ref{eq:FxPtPb}) relates 
to the ODE (\ref{eq:ODE}) we first consider the simplest possible case
$N=2$, with symmetric migration $m = M_{12} = M_{21}$.

Denote the one-dimensional, associated equilibrium $P^\text{eq}_{12}$ by $p^\text{eq}$. In this case, the distribution of the random variable $T_{1 2}$  is given by the minimum of two exponential random variables with parameter $ m h(p^\text{eq})$ since coalescence occurs at the first jump of one of the two random walks.
This minimum is an exponential law of parameter $2 m h(p^\text{eq})$ and the fixed point equation (\ref{eq:FxPtPb}) writes 
\begin{align}\label{eq:FxPtPb_dim2}
p^\text{eq} = \frac{mh(p^\text{eq})}{\mu + mh(p^\text{eq})} \backassign f(p^\text{eq}).
\end{align}
which coincides 
with the equilibrium condition for 
the ODE (\ref{eq:ODE_dim2}).

\medskip

Let us now turn to the stability analysis of the ODE.
We remark that $p^\text{eq} = 0$, corresponding to speciation between populations 1 and 2, is always an equilibrium. According to (\ref{eq:ODE_dim2}), $0$ is an unstable equilibrium if and only if 
\begin{align}\label{eq:Stab_dim2}
\frac{d f}{d p}\Bigr\rvert_{p = 0} = \frac{mh'(0)}{\mu} > 1\, .
\end{align}
In words, if  migration between the two populations has ceased for sufficiently long that they  achieve total reproductive isolation ($p=0$) and if a small quantity of genetic material is then artificially introgressed from one population into the other ($p=\varepsilon>0$), they would resume gene flow upon a secondary contact if (\ref{eq:Stab_dim2}) is verified.

If $h'(0) > 0$, this implies that if migration rates are sufficiently large or mutation rates sufficiently small, reproductively isolated populations could fuse again following even modest introgression. The occurrence of such fusions would be problematic and contradict the general observation that complete reproductive isolation is \hyperlink{reversibility}{irreversible} (see \cite{coyne2004speciation}, p. 37f, and \cite{orr1995population}). Therefore, we must and will suppose throughout the rest of the article
\begin{align}
    h'(0) = 0 \, .\label{eq:irrev}
\end{align}
Remarkably, we show that even when $N>2$,  the simple condition \eqref{eq:irrev} guarantees that any configuration made of several species complexes  mutually isolated from each other is also stable under any small perturbation by introgression of previously unshared alleles. See Remark \ref{rem:glob_loc} and Proposition \ref{prop:glob_loc} in the SI for a precise statement and a proof. 

Note that the simplest choices of a nondecreasing function $h$ satisfying $h(0)=h'(0)=0$ and $h(1)=1$ include $h(x)=x^a$ for $a>1$. Returning to the case $N=2$ and assuming $h(x)=x^a$ with $a\ge 2$, the ODE \eqref{eq:ODE_dim2}
then becomes
$$
\dot p = 2mp^a(1-p)-2\mu p,
$$
which has three equilibria $0< p_0<p_1$ (provided $\mu/m<(a-1)^{a-1}/a^a$), where $0$ and $p_1$ are stable, while $p_0$ is unstable. %the latter two being solutions to $p^{a-1}(1-p)=\frac\mu m$.
%derivee p^{a-2}((a-1)-ap) 
%argmax (a-1)/a
% deux solutions asa (a-1)^{a-1}/a^a > mu/m
This gives an inspiring picture of speciation as a bistable process:
\begin{itemize}
\item Under continuous migration, the two populations sit at the
stable migration-mutation equilibrium characterized by genetic proximity $p_1$;
\item If migration were to cease, mutations would accumulate and genetic proximities drop to some value $p$ at the end of the allopatric phase;
\item if $p>p_0$, the populations fuse again at secondary contact and genetic proximities recover to equilibrium value $p_1$, while if $p<p_0$, the divergence feedback takes them into a snowball process where proximities decrease to 0 (total reproductive isolation). See Sections \ref{sec:subsp_clust} and \ref{sec:fluct} for generalizations.
\end{itemize}

\medskip

Before closing this section, let us emphasize that if the ODE approach seems much more direct in the case $N=2$, it is far from obvious how to assess its general behavior in large species complexes.
This already hints at an observation we will address in later sections: the two approaches presented are complementary in the sense that the ODE approach is well suited to describe small metapopulations, while the fixed-point problem is well suited to describe large metapopulations.

\section{Intransitive species}
\label{chap:with_threshold}

Patterns such as ring species or hybrid zones show how diverse the shapes of species complexes can be (see \cite{irwin2001ring,barton1989adaptation}), raising the question: How does divergence feedback determine the shape of a species complex? 

We begin by defining the notion of species complexes in our framework. Let $P^\text{eq} = (P^\text{eq}_{i j})_{1\le i,j\le N}$ be an equilibrium for the system of genetic proximities (\ref{eq:LinSysGenProx}). We say that a group of populations $S\subseteq \{1,\dots,N\}$ forms a \textbf{species} if any two populations $i$ and $j$ therein can exchange genes, either directly (i.e., $h(P^\text{eq}_{i j}) > 0$), or through a chain of intermediary populations (i.e., there is $i = k_0, k_1,\dots,k_n = j$ such that $h(P^\text{eq}_{k_{l-1}k_{l}}) > 0$ for all $1\le l\le n$).

We first claim that if $i$ and $j$ belong to the same species, then we actually must have $P^\text{eq}_{ij}>0$. Mathematically, this can be seen from the right-hand side of the fixed point problem (\ref{eq:FxPtPb}). Indeed, if $i,j$ belong to the same species, then $T_{ij}(P^\text{eq})$ is finite with positive probability (since there is a chain of intermediary populations between $i$ and $j$ that can interbreed), so that $P_{ij}^\text{eq}=\mathbb{E}\left[e^{-2\mu T_{ij}(P^\text{eq}) }\right]>0$ (see Remark \ref{rem:transitivity} in SI).

If we assume that $h>0$ on $(0,1]$, then $P_{ij}^\text{eq}>0$ implies $h(P_{ij}^\text{eq})>0$ and populations within the same species will always be able to interbreed. The situation is more complex if we assume that populations cannot interbreed below a genetic threshold $c$, that is, when there exists $c$ such that $h(x)=0$ for $x<c$. In this case, we observe the emergence of \hyperlink{trans}{intransitive} interbreeding networks, in the sense that, even if $i$ interbreeds with $j$, and $j$ interbreeds with $k$, $i$ and $k$ cannot necessarily interbreed. We provide two examples.

\paragraph{Friendship graph.} First, we consider a complete migration graph of odd size $N$ and constant $M_{ij}=m$. By performing simulations (see \ref{fig:friendship_sim}), we show that we can choose a feedback function $h$, such that the species graph (see Fig. \ref{fig:MigGraphs}, (a2)) is stable so that individuals can only interbreed if they belong to the same triangle. 
This example illustrates that despite the uniformity of the underlying migration structure, non-transitive interbreeding structures can emerge, an interesting case of symmetry breaking. Our simulations also reveal that the friendship graph can only exist for small enough $N$, see Fig. \ref{fig:friendship_sim}. We also demonstrate this property analytically (see Proposition \ref{prop:sym_break_2} and Remark \ref{rem:sym_break} in SI).

\medskip

\paragraph{Ring species.} We now consider $N$ populations in a ring  migration structure (see Fig. \ref{fig:MigGraphs}, (b1)) with reduced migration between the two terminal populations. For the sake of illustration, we will assume that  
 the migration rates are constant equal to $m$ except at the end point where $M_{1N} = M_{N1} = \frac{m}{2}$. This setting models the existence of a geographic obstacle hampering migrations and corresponding to an area of unsuitable habitat, see for instance \cite{kuchta2009closing} for the celebrated example of the salamander \textit{Ensatina eschscholtzii} species complex  surrounding the Californian Central Valley, or \cite{cacho2012caribbean} for the species complex of slipper spurge \textit{Euphorbia tithymaloides} surrounding the Caribbean sea.

In Fig. \ref{fig:ring_species}, we investigate the existence of a ring species where the two end populations $1$ and $N$ are reproductively isolated from each other, despite ongoing gene flow through intermediary populations. The simulations reveal that while requiring very specific conditions (small migration/high mutation, low enough threshold), ring species can exist stably in a static environment. 
The range of parameters allowing for a stable equilibrium reflects the fact that extreme values of $\mu/m$ and $c$ tend to either produce several species ($\mu/m$ resp. $c$ too large) or to close the ring ($\mu/m$ resp. $c$ too small).

\begin{figure} 
\captionsetup{font=small}
\centering 
 
\begin{subfigure}[b]{0.4\linewidth}
\centering
\scalebox{0.65}{

\tikzset{every picture/.style={line width=0.75pt}} %set default line width to 0.75pt        

\begin{tikzpicture}[x=0.75pt,y=0.75pt,yscale=-1,xscale=1]
%uncomment if require: \path (0,600); %set diagram left start at 0, and has height of 600

%Shape: Circle [id:dp9629702481885856] 
\draw   (606.92,171.92) .. controls (606.92,167.08) and (610.83,163.17) .. (615.67,163.17) .. controls (620.5,163.17) and (624.42,167.08) .. (624.42,171.92) .. controls (624.42,176.75) and (620.5,180.67) .. (615.67,180.67) .. controls (610.83,180.67) and (606.92,176.75) .. (606.92,171.92) -- cycle ;
%Shape: Circle [id:dp09467087086796411] 
\draw  [color={rgb, 255:red, 0; green, 0; blue, 0 }  ,draw opacity=1 ][fill={rgb, 255:red, 255; green, 255; blue, 255 }  ,fill opacity=1 ] (606.92,101.67) .. controls (606.92,96.83) and (610.83,92.92) .. (615.67,92.92) .. controls (620.5,92.92) and (624.42,96.83) .. (624.42,101.67) .. controls (624.42,106.5) and (620.5,110.42) .. (615.67,110.42) .. controls (610.83,110.42) and (606.92,106.5) .. (606.92,101.67) -- cycle ;
%Shape: Circle [id:dp49496097120000127] 
\draw   (651.91,118.1) .. controls (651.91,113.27) and (655.83,109.35) .. (660.66,109.35) .. controls (665.49,109.35) and (669.41,113.27) .. (669.41,118.1) .. controls (669.41,122.93) and (665.49,126.85) .. (660.66,126.85) .. controls (655.83,126.85) and (651.91,122.93) .. (651.91,118.1) -- cycle ;
%Shape: Circle [id:dp4496724706042837] 
\draw   (675.85,159.72) .. controls (675.85,154.89) and (679.77,150.97) .. (684.6,150.97) .. controls (689.44,150.97) and (693.35,154.89) .. (693.35,159.72) .. controls (693.35,164.55) and (689.44,168.47) .. (684.6,168.47) .. controls (679.77,168.47) and (675.85,164.55) .. (675.85,159.72) -- cycle ;
%Shape: Circle [id:dp269005800499165] 
\draw   (667.54,207.04) .. controls (667.54,202.21) and (671.46,198.29) .. (676.29,198.29) .. controls (681.12,198.29) and (685.04,202.21) .. (685.04,207.04) .. controls (685.04,211.87) and (681.12,215.79) .. (676.29,215.79) .. controls (671.46,215.79) and (667.54,211.87) .. (667.54,207.04) -- cycle ;
%Shape: Circle [id:dp33752563414087444] 
\draw   (630.86,237.93) .. controls (630.86,233.1) and (634.78,229.18) .. (639.61,229.18) .. controls (644.44,229.18) and (648.36,233.1) .. (648.36,237.93) .. controls (648.36,242.76) and (644.44,246.68) .. (639.61,246.68) .. controls (634.78,246.68) and (630.86,242.76) .. (630.86,237.93) -- cycle ;
%Shape: Circle [id:dp065600634783048] 
\draw   (582.98,237.93) .. controls (582.98,233.1) and (586.89,229.18) .. (591.73,229.18) .. controls (596.56,229.18) and (600.48,233.1) .. (600.48,237.93) .. controls (600.48,242.76) and (596.56,246.68) .. (591.73,246.68) .. controls (586.89,246.68) and (582.98,242.76) .. (582.98,237.93) -- cycle ;
%Shape: Circle [id:dp0732570087493194] 
\draw   (546.29,207.04) .. controls (546.29,202.21) and (550.21,198.29) .. (555.04,198.29) .. controls (559.88,198.29) and (563.79,202.21) .. (563.79,207.04) .. controls (563.79,211.87) and (559.88,215.79) .. (555.04,215.79) .. controls (550.21,215.79) and (546.29,211.87) .. (546.29,207.04) -- cycle ;
%Shape: Circle [id:dp62259264169909] 
\draw   (561.92,118.1) .. controls (561.92,113.27) and (565.84,109.35) .. (570.67,109.35) .. controls (575.5,109.35) and (579.42,113.27) .. (579.42,118.1) .. controls (579.42,122.93) and (575.5,126.85) .. (570.67,126.85) .. controls (565.84,126.85) and (561.92,122.93) .. (561.92,118.1) -- cycle ;
%Straight Lines [id:da5881093867007174] 
\draw    (660.66,118.1) -- (615.67,171.92) ;
%Straight Lines [id:da5159521764018732] 
\draw    (684.6,159.72) -- (615.67,171.92) ;
%Straight Lines [id:da8109414414556537] 
\draw    (615.67,171.92) -- (676.29,207.04) ;
%Straight Lines [id:da30930400110448886] 
\draw    (615.67,171.92) -- (639.61,237.93) ;
%Straight Lines [id:da5566834277521614] 
\draw    (615.67,171.92) -- (591.73,237.93) ;
%Straight Lines [id:da12047348892425702] 
\draw    (615.67,171.92) -- (555.04,207.04) ;
%Straight Lines [id:da4070275986581072] 
\draw    (570.67,118.1) -- (615.67,171.92) ;
%Shape: Circle [id:dp9603243485076272] 
\draw  [fill={rgb, 255:red, 255; green, 255; blue, 255 }  ,fill opacity=1 ] (606.92,171.92) .. controls (606.92,167.08) and (610.83,163.17) .. (615.67,163.17) .. controls (620.5,163.17) and (624.42,167.08) .. (624.42,171.92) .. controls (624.42,176.75) and (620.5,180.67) .. (615.67,180.67) .. controls (610.83,180.67) and (606.92,176.75) .. (606.92,171.92) -- cycle ;
%Shape: Circle [id:dp5770446128960326] 
\draw  [fill={rgb, 255:red, 255; green, 255; blue, 255 }  ,fill opacity=1 ] (630.86,237.93) .. controls (630.86,233.1) and (634.78,229.18) .. (639.61,229.18) .. controls (644.44,229.18) and (648.36,233.1) .. (648.36,237.93) .. controls (648.36,242.76) and (644.44,246.68) .. (639.61,246.68) .. controls (634.78,246.68) and (630.86,242.76) .. (630.86,237.93) -- cycle ;
%Shape: Circle [id:dp41751349556757633] 
\draw  [fill={rgb, 255:red, 255; green, 255; blue, 255 }  ,fill opacity=1 ] (667.54,207.04) .. controls (667.54,202.21) and (671.46,198.29) .. (676.29,198.29) .. controls (681.12,198.29) and (685.04,202.21) .. (685.04,207.04) .. controls (685.04,211.87) and (681.12,215.79) .. (676.29,215.79) .. controls (671.46,215.79) and (667.54,211.87) .. (667.54,207.04) -- cycle ;
%Shape: Circle [id:dp3052587171515303] 
\draw  [fill={rgb, 255:red, 255; green, 255; blue, 255 }  ,fill opacity=1 ] (675.85,159.72) .. controls (675.85,154.89) and (679.77,150.97) .. (684.6,150.97) .. controls (689.44,150.97) and (693.35,154.89) .. (693.35,159.72) .. controls (693.35,164.55) and (689.44,168.47) .. (684.6,168.47) .. controls (679.77,168.47) and (675.85,164.55) .. (675.85,159.72) -- cycle ;
%Shape: Circle [id:dp30917804496333057] 
\draw  [fill={rgb, 255:red, 255; green, 255; blue, 255 }  ,fill opacity=1 ] (651.91,118.1) .. controls (651.91,113.27) and (655.83,109.35) .. (660.66,109.35) .. controls (665.49,109.35) and (669.41,113.27) .. (669.41,118.1) .. controls (669.41,122.93) and (665.49,126.85) .. (660.66,126.85) .. controls (655.83,126.85) and (651.91,122.93) .. (651.91,118.1) -- cycle ;
%Shape: Circle [id:dp3252211000118844] 
\draw  [fill={rgb, 255:red, 255; green, 255; blue, 255 }  ,fill opacity=1 ] (606.92,101.67) .. controls (606.92,96.83) and (610.83,92.92) .. (615.67,92.92) .. controls (620.5,92.92) and (624.42,96.83) .. (624.42,101.67) .. controls (624.42,106.5) and (620.5,110.42) .. (615.67,110.42) .. controls (610.83,110.42) and (606.92,106.5) .. (606.92,101.67) -- cycle ;
%Shape: Circle [id:dp6181849847319776] 
\draw  [fill={rgb, 255:red, 255; green, 255; blue, 255 }  ,fill opacity=1 ] (561.92,118.1) .. controls (561.92,113.27) and (565.84,109.35) .. (570.67,109.35) .. controls (575.5,109.35) and (579.42,113.27) .. (579.42,118.1) .. controls (579.42,122.93) and (575.5,126.85) .. (570.67,126.85) .. controls (565.84,126.85) and (561.92,122.93) .. (561.92,118.1) -- cycle ;
%Shape: Circle [id:dp19822871673310494] 
\draw  [fill={rgb, 255:red, 255; green, 255; blue, 255 }  ,fill opacity=1 ] (546.29,207.04) .. controls (546.29,202.21) and (550.21,198.29) .. (555.04,198.29) .. controls (559.88,198.29) and (563.79,202.21) .. (563.79,207.04) .. controls (563.79,211.87) and (559.88,215.79) .. (555.04,215.79) .. controls (550.21,215.79) and (546.29,211.87) .. (546.29,207.04) -- cycle ;
%Shape: Circle [id:dp9437929826717535] 
\draw  [fill={rgb, 255:red, 255; green, 255; blue, 255 }  ,fill opacity=1 ] (582.98,237.93) .. controls (582.98,233.1) and (586.89,229.18) .. (591.73,229.18) .. controls (596.56,229.18) and (600.48,233.1) .. (600.48,237.93) .. controls (600.48,242.76) and (596.56,246.68) .. (591.73,246.68) .. controls (586.89,246.68) and (582.98,242.76) .. (582.98,237.93) -- cycle ;
%Straight Lines [id:da4503714667673473] 
\draw    (555.04,207.04) -- (591.73,237.93) ;
%Straight Lines [id:da1157856005933624] 
\draw    (615.67,101.67) -- (570.67,118.1) ;
%Straight Lines [id:da7793072202784858] 
\draw    (684.6,159.72) -- (676.29,207.04) ;
%Straight Lines [id:da3317186829885709] 
\draw    (639.61,237.93) -- (591.73,237.93) ;
%Straight Lines [id:da5863104812885737] 
\draw    (660.66,118.1) -- (684.6,159.72) ;
%Straight Lines [id:da8018990075596559] 
\draw    (676.29,207.04) -- (639.61,237.93) ;
%Straight Lines [id:da8330742651402465] 
\draw    (615.67,101.67) -- (615.67,171.92) ;
%Straight Lines [id:da698115191262698] 
\draw    (570.67,118.1) -- (562.77,163.05) -- (555.04,207.04) ;
%Straight Lines [id:da7207243413124796] 
\draw    (660.66,118.1) -- (615.67,101.67) ;
%Straight Lines [id:da8100684306015857] 
\draw [color={rgb, 255:red, 0; green, 0; blue, 0 }  ,draw opacity=0.31 ]   (615.67,101.67) -- (555.04,207.04) ;
%Straight Lines [id:da11227094130016102] 
\draw [color={rgb, 255:red, 0; green, 0; blue, 0 }  ,draw opacity=0.31 ]   (660.66,118.1) -- (555.04,207.04) ;
%Straight Lines [id:da16878447600595203] 
\draw [color={rgb, 255:red, 0; green, 0; blue, 0 }  ,draw opacity=0.31 ]   (684.6,159.72) -- (555.04,207.04) ;
%Straight Lines [id:da6878684836774915] 
\draw [color={rgb, 255:red, 0; green, 0; blue, 0 }  ,draw opacity=0.31 ]   (639.61,237.93) -- (555.04,207.04) ;
%Straight Lines [id:da9097000781476708] 
\draw [color={rgb, 255:red, 0; green, 0; blue, 0 }  ,draw opacity=0.31 ]   (684.6,159.72) -- (615.67,101.67) ;
%Straight Lines [id:da6156786551413175] 
\draw [color={rgb, 255:red, 0; green, 0; blue, 0 }  ,draw opacity=0.31 ]   (676.29,207.04) -- (615.67,101.67) ;
%Straight Lines [id:da23021725544954186] 
\draw [color={rgb, 255:red, 0; green, 0; blue, 0 }  ,draw opacity=0.31 ]   (639.61,237.93) -- (615.67,101.67) ;
%Straight Lines [id:da7481014844695353] 
\draw [color={rgb, 255:red, 0; green, 0; blue, 0 }  ,draw opacity=0.31 ]   (591.73,237.93) -- (615.67,101.67) ;
%Straight Lines [id:da5553670087100985] 
\draw [color={rgb, 255:red, 0; green, 0; blue, 0 }  ,draw opacity=0.31 ]   (676.29,207.04) -- (555.04,207.04) ;
%Straight Lines [id:da7228327301349431] 
\draw [color={rgb, 255:red, 0; green, 0; blue, 0 }  ,draw opacity=0.31 ]   (639.61,237.93) -- (570.67,118.1) ;
%Straight Lines [id:da7496457394974545] 
\draw [color={rgb, 255:red, 0; green, 0; blue, 0 }  ,draw opacity=0.31 ]   (676.29,207.04) -- (570.67,118.1) ;
%Straight Lines [id:da05104613860390772] 
\draw [color={rgb, 255:red, 0; green, 0; blue, 0 }  ,draw opacity=0.31 ]   (684.6,159.72) -- (570.67,118.1) ;
%Straight Lines [id:da9430387334781583] 
\draw [color={rgb, 255:red, 0; green, 0; blue, 0 }  ,draw opacity=0.31 ]   (660.66,118.1) -- (570.67,118.1) ;
%Straight Lines [id:da6126711722289455] 
\draw [color={rgb, 255:red, 0; green, 0; blue, 0 }  ,draw opacity=0.31 ]   (591.73,237.93) -- (570.67,118.1) ;
%Straight Lines [id:da2362797475334344] 
\draw [color={rgb, 255:red, 0; green, 0; blue, 0 }  ,draw opacity=0.31 ]   (660.66,118.1) -- (591.73,237.93) ;
%Straight Lines [id:da641121176010157] 
\draw [color={rgb, 255:red, 0; green, 0; blue, 0 }  ,draw opacity=0.31 ]   (684.6,159.72) -- (591.73,237.93) ;
%Straight Lines [id:da7944149101587542] 
\draw [color={rgb, 255:red, 0; green, 0; blue, 0 }  ,draw opacity=0.31 ]   (676.29,207.04) -- (591.73,237.93) ;
%Straight Lines [id:da368001904560038] 
\draw [color={rgb, 255:red, 0; green, 0; blue, 0 }  ,draw opacity=0.31 ]   (639.61,237.93) -- (660.66,118.1) ;
%Straight Lines [id:da5669111716496954] 
\draw [color={rgb, 255:red, 0; green, 0; blue, 0 }  ,draw opacity=0.31 ]   (639.61,237.93) -- (684.6,159.72) ;
%Shape: Circle [id:dp7820445967082957] 
\draw  [fill={rgb, 255:red, 255; green, 255; blue, 255 }  ,fill opacity=1 ] (561.92,118.1) .. controls (561.92,113.27) and (565.84,109.35) .. (570.67,109.35) .. controls (575.5,109.35) and (579.42,113.27) .. (579.42,118.1) .. controls (579.42,122.93) and (575.5,126.85) .. (570.67,126.85) .. controls (565.84,126.85) and (561.92,122.93) .. (561.92,118.1) -- cycle ;
%Shape: Circle [id:dp9859083012189594] 
\draw  [fill={rgb, 255:red, 255; green, 255; blue, 255 }  ,fill opacity=1 ] (546.29,207.04) .. controls (546.29,202.21) and (550.21,198.29) .. (555.04,198.29) .. controls (559.88,198.29) and (563.79,202.21) .. (563.79,207.04) .. controls (563.79,211.87) and (559.88,215.79) .. (555.04,215.79) .. controls (550.21,215.79) and (546.29,211.87) .. (546.29,207.04) -- cycle ;
%Shape: Circle [id:dp9260845027189283] 
\draw  [fill={rgb, 255:red, 255; green, 255; blue, 255 }  ,fill opacity=1 ] (606.92,101.67) .. controls (606.92,96.83) and (610.83,92.92) .. (615.67,92.92) .. controls (620.5,92.92) and (624.42,96.83) .. (624.42,101.67) .. controls (624.42,106.5) and (620.5,110.42) .. (615.67,110.42) .. controls (610.83,110.42) and (606.92,106.5) .. (606.92,101.67) -- cycle ;
%Shape: Circle [id:dp5139678001547064] 
\draw  [fill={rgb, 255:red, 255; green, 255; blue, 255 }  ,fill opacity=1 ] (651.91,118.1) .. controls (651.91,113.27) and (655.83,109.35) .. (660.66,109.35) .. controls (665.49,109.35) and (669.41,113.27) .. (669.41,118.1) .. controls (669.41,122.93) and (665.49,126.85) .. (660.66,126.85) .. controls (655.83,126.85) and (651.91,122.93) .. (651.91,118.1) -- cycle ;
%Shape: Circle [id:dp4475141086766762] 
\draw  [fill={rgb, 255:red, 255; green, 255; blue, 255 }  ,fill opacity=1 ] (606.92,171.92) .. controls (606.92,167.08) and (610.83,163.17) .. (615.67,163.17) .. controls (620.5,163.17) and (624.42,167.08) .. (624.42,171.92) .. controls (624.42,176.75) and (620.5,180.67) .. (615.67,180.67) .. controls (610.83,180.67) and (606.92,176.75) .. (606.92,171.92) -- cycle ;
%Shape: Circle [id:dp5842994503748883] 
\draw  [fill={rgb, 255:red, 255; green, 255; blue, 255 }  ,fill opacity=1 ] (582.98,237.93) .. controls (582.98,233.1) and (586.89,229.18) .. (591.73,229.18) .. controls (596.56,229.18) and (600.48,233.1) .. (600.48,237.93) .. controls (600.48,242.76) and (596.56,246.68) .. (591.73,246.68) .. controls (586.89,246.68) and (582.98,242.76) .. (582.98,237.93) -- cycle ;
%Shape: Circle [id:dp9423424460235557] 
\draw  [fill={rgb, 255:red, 255; green, 255; blue, 255 }  ,fill opacity=1 ] (630.86,237.93) .. controls (630.86,233.1) and (634.78,229.18) .. (639.61,229.18) .. controls (644.44,229.18) and (648.36,233.1) .. (648.36,237.93) .. controls (648.36,242.76) and (644.44,246.68) .. (639.61,246.68) .. controls (634.78,246.68) and (630.86,242.76) .. (630.86,237.93) -- cycle ;
%Shape: Circle [id:dp7147854759930389] 
\draw  [fill={rgb, 255:red, 255; green, 255; blue, 255 }  ,fill opacity=1 ] (667.54,207.04) .. controls (667.54,202.21) and (671.46,198.29) .. (676.29,198.29) .. controls (681.12,198.29) and (685.04,202.21) .. (685.04,207.04) .. controls (685.04,211.87) and (681.12,215.79) .. (676.29,215.79) .. controls (671.46,215.79) and (667.54,211.87) .. (667.54,207.04) -- cycle ;
%Shape: Circle [id:dp7366223722739519] 
\draw  [fill={rgb, 255:red, 255; green, 255; blue, 255 }  ,fill opacity=1 ] (675.85,159.72) .. controls (675.85,154.89) and (679.77,150.97) .. (684.6,150.97) .. controls (689.44,150.97) and (693.35,154.89) .. (693.35,159.72) .. controls (693.35,164.55) and (689.44,168.47) .. (684.6,168.47) .. controls (679.77,168.47) and (675.85,164.55) .. (675.85,159.72) -- cycle ;

\end{tikzpicture}

}
\vspace*{0.15mm}
\caption*{(a1) Complete migration graph} \label{fig:FriendshipMG}  
  \end{subfigure}
\qquad
\begin{subfigure}[b]{0.4\linewidth}
\centering
\scalebox{0.65}{

\tikzset{every picture/.style={line width=0.75pt}} %set default line width to 0.75pt        
\begin{tikzpicture}[x=0.75pt,y=0.75pt,yscale=-1,xscale=1]
%uncomment if require: \path (0,333); %set diagram left start at 0, and has height of 333

%Shape: Circle [id:dp8840691018266539] 
\draw   (319.25,163.3) .. controls (319.25,158.47) and (323.17,154.55) .. (328,154.55) .. controls (332.83,154.55) and (336.75,158.47) .. (336.75,163.3) .. controls (336.75,168.13) and (332.83,172.05) .. (328,172.05) .. controls (323.17,172.05) and (319.25,168.13) .. (319.25,163.3) -- cycle ;
%Shape: Circle [id:dp4920788766510933] 
\draw  [color={rgb, 255:red, 0; green, 0; blue, 0 }  ,draw opacity=1 ][fill={rgb, 255:red, 255; green, 255; blue, 255 }  ,fill opacity=1 ] (319.25,93.05) .. controls (319.25,88.22) and (323.17,84.3) .. (328,84.3) .. controls (332.83,84.3) and (336.75,88.22) .. (336.75,93.05) .. controls (336.75,97.88) and (332.83,101.8) .. (328,101.8) .. controls (323.17,101.8) and (319.25,97.88) .. (319.25,93.05) -- cycle ;
%Shape: Circle [id:dp8234062232598044] 
\draw   (364.25,109.49) .. controls (364.25,104.65) and (368.16,100.74) .. (373,100.74) .. controls (377.83,100.74) and (381.75,104.65) .. (381.75,109.49) .. controls (381.75,114.32) and (377.83,118.24) .. (373,118.24) .. controls (368.16,118.24) and (364.25,114.32) .. (364.25,109.49) -- cycle ;
%Shape: Circle [id:dp598763263780157] 
\draw   (388.19,151.1) .. controls (388.19,146.27) and (392.1,142.35) .. (396.94,142.35) .. controls (401.77,142.35) and (405.69,146.27) .. (405.69,151.1) .. controls (405.69,155.93) and (401.77,159.85) .. (396.94,159.85) .. controls (392.1,159.85) and (388.19,155.93) .. (388.19,151.1) -- cycle ;
%Shape: Circle [id:dp7385583464906915] 
\draw   (379.87,198.43) .. controls (379.87,193.59) and (383.79,189.68) .. (388.62,189.68) .. controls (393.45,189.68) and (397.37,193.59) .. (397.37,198.43) .. controls (397.37,203.26) and (393.45,207.18) .. (388.62,207.18) .. controls (383.79,207.18) and (379.87,203.26) .. (379.87,198.43) -- cycle ;
%Shape: Circle [id:dp9560557042545129] 
\draw   (343.19,229.31) .. controls (343.19,224.48) and (347.11,220.56) .. (351.94,220.56) .. controls (356.77,220.56) and (360.69,224.48) .. (360.69,229.31) .. controls (360.69,234.15) and (356.77,238.06) .. (351.94,238.06) .. controls (347.11,238.06) and (343.19,234.15) .. (343.19,229.31) -- cycle ;
%Shape: Circle [id:dp6291624879679977] 
\draw   (295.31,229.31) .. controls (295.31,224.48) and (299.23,220.56) .. (304.06,220.56) .. controls (308.89,220.56) and (312.81,224.48) .. (312.81,229.31) .. controls (312.81,234.15) and (308.89,238.06) .. (304.06,238.06) .. controls (299.23,238.06) and (295.31,234.15) .. (295.31,229.31) -- cycle ;
%Shape: Circle [id:dp8488227593961265] 
\draw   (258.63,198.43) .. controls (258.63,193.59) and (262.55,189.68) .. (267.38,189.68) .. controls (272.21,189.68) and (276.13,193.59) .. (276.13,198.43) .. controls (276.13,203.26) and (272.21,207.18) .. (267.38,207.18) .. controls (262.55,207.18) and (258.63,203.26) .. (258.63,198.43) -- cycle ;
%Shape: Circle [id:dp003404283001073871] 
\draw   (274.25,109.49) .. controls (274.25,104.65) and (278.17,100.74) .. (283,100.74) .. controls (287.84,100.74) and (291.75,104.65) .. (291.75,109.49) .. controls (291.75,114.32) and (287.84,118.24) .. (283,118.24) .. controls (278.17,118.24) and (274.25,114.32) .. (274.25,109.49) -- cycle ;
%Straight Lines [id:da7611644974801981] 
\draw    (373,109.49) -- (328,163.3) ;
%Straight Lines [id:da19027816834305789] 
\draw    (396.94,151.1) -- (328,163.3) ;
%Straight Lines [id:da3840772411222835] 
\draw    (328,163.3) -- (388.62,198.43) ;
%Straight Lines [id:da25918163320099474] 
\draw    (328,163.3) -- (351.94,229.31) ;
%Straight Lines [id:da9463976013826302] 
\draw    (328,163.3) -- (304.06,229.31) ;
%Straight Lines [id:da5566724738436036] 
\draw    (328,163.3) -- (267.38,198.43) ;
%Straight Lines [id:da10192895914517486] 
\draw    (283,109.49) -- (328,163.3) ;
%Shape: Circle [id:dp5736411472985462] 
\draw  [fill={rgb, 255:red, 255; green, 255; blue, 255 }  ,fill opacity=1 ] (319.25,163.3) .. controls (319.25,158.47) and (323.17,154.55) .. (328,154.55) .. controls (332.83,154.55) and (336.75,158.47) .. (336.75,163.3) .. controls (336.75,168.13) and (332.83,172.05) .. (328,172.05) .. controls (323.17,172.05) and (319.25,168.13) .. (319.25,163.3) -- cycle ;
%Shape: Circle [id:dp5440162406890321] 
\draw  [fill={rgb, 255:red, 255; green, 255; blue, 255 }  ,fill opacity=1 ] (343.19,229.31) .. controls (343.19,224.48) and (347.11,220.56) .. (351.94,220.56) .. controls (356.77,220.56) and (360.69,224.48) .. (360.69,229.31) .. controls (360.69,234.15) and (356.77,238.06) .. (351.94,238.06) .. controls (347.11,238.06) and (343.19,234.15) .. (343.19,229.31) -- cycle ;
%Shape: Circle [id:dp7528691341773774] 
\draw  [fill={rgb, 255:red, 255; green, 255; blue, 255 }  ,fill opacity=1 ] (379.87,198.43) .. controls (379.87,193.59) and (383.79,189.68) .. (388.62,189.68) .. controls (393.45,189.68) and (397.37,193.59) .. (397.37,198.43) .. controls (397.37,203.26) and (393.45,207.18) .. (388.62,207.18) .. controls (383.79,207.18) and (379.87,203.26) .. (379.87,198.43) -- cycle ;
%Shape: Circle [id:dp28731299401271826] 
\draw  [fill={rgb, 255:red, 255; green, 255; blue, 255 }  ,fill opacity=1 ] (388.19,151.1) .. controls (388.19,146.27) and (392.1,142.35) .. (396.94,142.35) .. controls (401.77,142.35) and (405.69,146.27) .. (405.69,151.1) .. controls (405.69,155.93) and (401.77,159.85) .. (396.94,159.85) .. controls (392.1,159.85) and (388.19,155.93) .. (388.19,151.1) -- cycle ;
%Shape: Circle [id:dp08246981966481082] 
\draw  [fill={rgb, 255:red, 255; green, 255; blue, 255 }  ,fill opacity=1 ] (364.25,109.49) .. controls (364.25,104.65) and (368.16,100.74) .. (373,100.74) .. controls (377.83,100.74) and (381.75,104.65) .. (381.75,109.49) .. controls (381.75,114.32) and (377.83,118.24) .. (373,118.24) .. controls (368.16,118.24) and (364.25,114.32) .. (364.25,109.49) -- cycle ;
%Shape: Circle [id:dp7121051956213246] 
\draw  [fill={rgb, 255:red, 255; green, 255; blue, 255 }  ,fill opacity=1 ] (319.25,93.05) .. controls (319.25,88.22) and (323.17,84.3) .. (328,84.3) .. controls (332.83,84.3) and (336.75,88.22) .. (336.75,93.05) .. controls (336.75,97.88) and (332.83,101.8) .. (328,101.8) .. controls (323.17,101.8) and (319.25,97.88) .. (319.25,93.05) -- cycle ;
%Shape: Circle [id:dp6190328271337764] 
\draw  [fill={rgb, 255:red, 255; green, 255; blue, 255 }  ,fill opacity=1 ] (274.25,109.49) .. controls (274.25,104.65) and (278.17,100.74) .. (283,100.74) .. controls (287.84,100.74) and (291.75,104.65) .. (291.75,109.49) .. controls (291.75,114.32) and (287.84,118.24) .. (283,118.24) .. controls (278.17,118.24) and (274.25,114.32) .. (274.25,109.49) -- cycle ;
%Shape: Circle [id:dp9984115710797348] 
\draw  [fill={rgb, 255:red, 255; green, 255; blue, 255 }  ,fill opacity=1 ] (258.63,198.43) .. controls (258.63,193.59) and (262.55,189.68) .. (267.38,189.68) .. controls (272.21,189.68) and (276.13,193.59) .. (276.13,198.43) .. controls (276.13,203.26) and (272.21,207.18) .. (267.38,207.18) .. controls (262.55,207.18) and (258.63,203.26) .. (258.63,198.43) -- cycle ;
%Shape: Circle [id:dp10281025089769236] 
\draw  [fill={rgb, 255:red, 255; green, 255; blue, 255 }  ,fill opacity=1 ] (295.31,229.31) .. controls (295.31,224.48) and (299.23,220.56) .. (304.06,220.56) .. controls (308.89,220.56) and (312.81,224.48) .. (312.81,229.31) .. controls (312.81,234.15) and (308.89,238.06) .. (304.06,238.06) .. controls (299.23,238.06) and (295.31,234.15) .. (295.31,229.31) -- cycle ;
%Straight Lines [id:da8137361009909511] 
\draw    (328,101.8) -- (328,154.55) ;
%Straight Lines [id:da6325061791114853] 
\draw [fill={rgb, 255:red, 245; green, 166; blue, 35 }  ,fill opacity=1 ]   (283,109.49) -- (328,93.05) ;
%Straight Lines [id:da44976092279712565] 
\draw [fill={rgb, 255:red, 74; green, 144; blue, 226 }  ,fill opacity=1 ]   (328,163.3) -- (267.38,198.43) ;
%Straight Lines [id:da09518673904177688] 
%Shape: Circle [id:dp5191181402023368] 
\draw  [fill={rgb, 255:red, 255; green, 255; blue, 255 }  ,fill opacity=1 ] (319.25,93.05) .. controls (319.25,88.22) and (323.17,84.3) .. (328,84.3) .. controls (332.83,84.3) and (336.75,88.22) .. (336.75,93.05) .. controls (336.75,97.88) and (332.83,101.8) .. (328,101.8) .. controls (323.17,101.8) and (319.25,97.88) .. (319.25,93.05) -- cycle ;
%Shape: Circle [id:dp788103556129673] 
\draw  [fill={rgb, 255:red, 255; green, 255; blue, 255 }  ,fill opacity=1 ] (274.25,109.49) .. controls (274.25,104.65) and (278.17,100.74) .. (283,100.74) .. controls (287.84,100.74) and (291.75,104.65) .. (291.75,109.49) .. controls (291.75,114.32) and (287.84,118.24) .. (283,118.24) .. controls (278.17,118.24) and (274.25,114.32) .. (274.25,109.49) -- cycle ;
%Shape: Circle [id:dp09838917655511392] 
\draw  [fill={rgb, 255:red, 255; green, 255; blue, 255 }  ,fill opacity=1 ] (258.63,198.43) .. controls (258.63,193.59) and (262.55,189.68) .. (267.38,189.68) .. controls (272.21,189.68) and (276.13,193.59) .. (276.13,198.43) .. controls (276.13,203.26) and (272.21,207.18) .. (267.38,207.18) .. controls (262.55,207.18) and (258.63,203.26) .. (258.63,198.43) -- cycle ;
%Shape: Circle [id:dp5803140903834672] 
\draw  [fill={rgb, 255:red, 255; green, 255; blue, 255 }  ,fill opacity=1 ] (319.25,163.3) .. controls (319.25,158.47) and (323.17,154.55) .. (328,154.55) .. controls (332.83,154.55) and (336.75,158.47) .. (336.75,163.3) .. controls (336.75,168.13) and (332.83,172.05) .. (328,172.05) .. controls (323.17,172.05) and (319.25,168.13) .. (319.25,163.3) -- cycle ;
%Shape: Circle [id:dp8619925644948444] 
\draw  [fill={rgb, 255:red, 255; green, 255; blue, 255 }  ,fill opacity=1 ] (388.19,151.1) .. controls (388.19,146.27) and (392.1,142.35) .. (396.94,142.35) .. controls (401.77,142.35) and (405.69,146.27) .. (405.69,151.1) .. controls (405.69,155.93) and (401.77,159.85) .. (396.94,159.85) .. controls (392.1,159.85) and (388.19,155.93) .. (388.19,151.1) -- cycle ;
%Shape: Circle [id:dp04198977066412557] 
\draw  [fill={rgb, 255:red, 255; green, 255; blue, 255 }  ,fill opacity=1 ] (379.87,198.43) .. controls (379.87,193.59) and (383.79,189.68) .. (388.62,189.68) .. controls (393.45,189.68) and (397.37,193.59) .. (397.37,198.43) .. controls (397.37,203.26) and (393.45,207.18) .. (388.62,207.18) .. controls (383.79,207.18) and (379.87,203.26) .. (379.87,198.43) -- cycle ;
%Straight Lines [id:da567611257767954] 
\draw    (267.38,198.43) -- (304.06,229.31) ;
%Straight Lines [id:da3742000952678197] 
\draw    (388.62,198.43) -- (351.94,229.31) ;
%Straight Lines [id:da9235330931507829] 
\draw    (373,109.49) -- (396.94,151.1) ;
%Shape: Circle [id:dp6279845666065064] 
\draw  [fill={rgb, 255:red, 255; green, 255; blue, 255 }  ,fill opacity=1 ] (364.25,109.49) .. controls (364.25,104.65) and (368.16,100.74) .. (373,100.74) .. controls (377.83,100.74) and (381.75,104.65) .. (381.75,109.49) .. controls (381.75,114.32) and (377.83,118.24) .. (373,118.24) .. controls (368.16,118.24) and (364.25,114.32) .. (364.25,109.49) -- cycle ;
%Shape: Circle [id:dp1422344236232851] 
\draw  [fill={rgb, 255:red, 255; green, 255; blue, 255 }  ,fill opacity=1 ] (388.19,151.1) .. controls (388.19,146.27) and (392.1,142.35) .. (396.94,142.35) .. controls (401.77,142.35) and (405.69,146.27) .. (405.69,151.1) .. controls (405.69,155.93) and (401.77,159.85) .. (396.94,159.85) .. controls (392.1,159.85) and (388.19,155.93) .. (388.19,151.1) -- cycle ;
%Shape: Circle [id:dp5536968544610096] 
\draw  [fill={rgb, 255:red, 255; green, 255; blue, 255 }  ,fill opacity=1 ] (379.87,198.43) .. controls (379.87,193.59) and (383.79,189.68) .. (388.62,189.68) .. controls (393.45,189.68) and (397.37,193.59) .. (397.37,198.43) .. controls (397.37,203.26) and (393.45,207.18) .. (388.62,207.18) .. controls (383.79,207.18) and (379.87,203.26) .. (379.87,198.43) -- cycle ;
%Shape: Circle [id:dp6477801936947635] 
\draw  [fill={rgb, 255:red, 255; green, 255; blue, 255 }  ,fill opacity=1 ] (343.19,229.31) .. controls (343.19,224.48) and (347.11,220.56) .. (351.94,220.56) .. controls (356.77,220.56) and (360.69,224.48) .. (360.69,229.31) .. controls (360.69,234.15) and (356.77,238.06) .. (351.94,238.06) .. controls (347.11,238.06) and (343.19,234.15) .. (343.19,229.31) -- cycle ;
%Shape: Circle [id:dp7373818581479141] 
\draw  [fill={rgb, 255:red, 255; green, 255; blue, 255 }  ,fill opacity=1 ] (295.31,229.31) .. controls (295.31,224.48) and (299.23,220.56) .. (304.06,220.56) .. controls (308.89,220.56) and (312.81,224.48) .. (312.81,229.31) .. controls (312.81,234.15) and (308.89,238.06) .. (304.06,238.06) .. controls (299.23,238.06) and (295.31,234.15) .. (295.31,229.31) -- cycle ;
%Shape: Circle [id:dp6536979168306176] 
\draw  [fill={rgb, 255:red, 255; green, 255; blue, 255 }  ,fill opacity=1 ] (258.63,198.43) .. controls (258.63,193.59) and (262.55,189.68) .. (267.38,189.68) .. controls (272.21,189.68) and (276.13,193.59) .. (276.13,198.43) .. controls (276.13,203.26) and (272.21,207.18) .. (267.38,207.18) .. controls (262.55,207.18) and (258.63,203.26) .. (258.63,198.43) -- cycle ;

\end{tikzpicture} }
\vspace*{0.15mm}
\caption*{(a2) Friendship interbreeding graph} \label{fig:FriendshipIG}  
  \end{subfigure}
\qquad
\begin{subfigure}[b]{0.4\linewidth}
\centering
\tikzset{every picture/.style={line width=0.75pt}} %set default line width to 0.75pt        
\scalebox{0.65}{

\tikzset{every picture/.style={line width=0.75pt}} %set default line width to 0.75pt        

\begin{tikzpicture}[x=0.75pt,y=0.75pt,yscale=-1,xscale=1]
%uncomment if require: \path (0,600); %set diagram left start at 0, and has height of 600

%Shape: Circle [id:dp49949437729851254] 
\draw  [color={rgb, 255:red, 0; green, 0; blue, 0 }  ,draw opacity=1 ][fill={rgb, 255:red, 255; green, 255; blue, 255 }  ,fill opacity=1 ] (166.75,366.5) .. controls (166.75,361.67) and (170.67,357.75) .. (175.5,357.75) .. controls (180.33,357.75) and (184.25,361.67) .. (184.25,366.5) .. controls (184.25,371.33) and (180.33,375.25) .. (175.5,375.25) .. controls (170.67,375.25) and (166.75,371.33) .. (166.75,366.5) -- cycle ;
%Shape: Circle [id:dp22482824242909127] 
\draw   (211.75,382.94) .. controls (211.75,378.1) and (215.66,374.19) .. (220.5,374.19) .. controls (225.33,374.19) and (229.25,378.1) .. (229.25,382.94) .. controls (229.25,387.77) and (225.33,391.69) .. (220.5,391.69) .. controls (215.66,391.69) and (211.75,387.77) .. (211.75,382.94) -- cycle ;
%Shape: Circle [id:dp41526713041522745] 
\draw   (235.69,424.55) .. controls (235.69,419.72) and (239.6,415.8) .. (244.44,415.8) .. controls (249.27,415.8) and (253.19,419.72) .. (253.19,424.55) .. controls (253.19,429.38) and (249.27,433.3) .. (244.44,433.3) .. controls (239.6,433.3) and (235.69,429.38) .. (235.69,424.55) -- cycle ;
%Shape: Circle [id:dp750244328665868] 
\draw   (227.37,471.88) .. controls (227.37,467.04) and (231.29,463.13) .. (236.12,463.13) .. controls (240.95,463.13) and (244.87,467.04) .. (244.87,471.88) .. controls (244.87,476.71) and (240.95,480.63) .. (236.12,480.63) .. controls (231.29,480.63) and (227.37,476.71) .. (227.37,471.88) -- cycle ;
%Shape: Circle [id:dp7933297610307806] 
\draw   (190.69,502.76) .. controls (190.69,497.93) and (194.61,494.01) .. (199.44,494.01) .. controls (204.27,494.01) and (208.19,497.93) .. (208.19,502.76) .. controls (208.19,507.6) and (204.27,511.51) .. (199.44,511.51) .. controls (194.61,511.51) and (190.69,507.6) .. (190.69,502.76) -- cycle ;
%Shape: Circle [id:dp7684296144587093] 
\draw   (142.81,502.76) .. controls (142.81,497.93) and (146.73,494.01) .. (151.56,494.01) .. controls (156.39,494.01) and (160.31,497.93) .. (160.31,502.76) .. controls (160.31,507.6) and (156.39,511.51) .. (151.56,511.51) .. controls (146.73,511.51) and (142.81,507.6) .. (142.81,502.76) -- cycle ;
%Shape: Circle [id:dp6304223159153954] 
\draw   (106.13,471.88) .. controls (106.13,467.04) and (110.05,463.13) .. (114.88,463.13) .. controls (119.71,463.13) and (123.63,467.04) .. (123.63,471.88) .. controls (123.63,476.71) and (119.71,480.63) .. (114.88,480.63) .. controls (110.05,480.63) and (106.13,476.71) .. (106.13,471.88) -- cycle ;
%Shape: Circle [id:dp8471997665304902] 
\draw   (121.75,382.94) .. controls (121.75,378.1) and (125.67,374.19) .. (130.5,374.19) .. controls (135.34,374.19) and (139.25,378.1) .. (139.25,382.94) .. controls (139.25,387.77) and (135.34,391.69) .. (130.5,391.69) .. controls (125.67,391.69) and (121.75,387.77) .. (121.75,382.94) -- cycle ;
%Straight Lines [id:da14931945489796605] 
\draw [line width=3.5]    (220.5,382.94) -- (244.44,424.55) ;
%Straight Lines [id:da7808680800906047] 
\draw [line width=0.75]    (244.44,424.55) -- (236.12,471.88) ;
%Straight Lines [id:da3586158426796501] 
\draw [line width=3.5]    (199.44,502.76) -- (236.12,471.88) ;
%Straight Lines [id:da682882928744017] 
\draw [line width=3.5]    (151.56,502.76) -- (199.44,502.76) ;
%Straight Lines [id:da20642802202081756] 
\draw [line width=3.5]    (114.88,471.88) -- (151.56,502.76) ;
%Straight Lines [id:da7581156902498262] 
\draw [line width=3.5]    (103.44,426.55) -- (114.88,471.88) ;
%Straight Lines [id:da752863313275981] 
\draw [line width=3.5]    (130.5,382.94) -- (175.5,366.5) ;
%Shape: Circle [id:dp03428200000950099] 
\draw  [fill={rgb, 255:red, 255; green, 255; blue, 255 }  ,fill opacity=1 ] (190.69,502.76) .. controls (190.69,497.93) and (194.61,494.01) .. (199.44,494.01) .. controls (204.27,494.01) and (208.19,497.93) .. (208.19,502.76) .. controls (208.19,507.6) and (204.27,511.51) .. (199.44,511.51) .. controls (194.61,511.51) and (190.69,507.6) .. (190.69,502.76) -- cycle ;
%Shape: Circle [id:dp9008603563231945] 
\draw  [fill={rgb, 255:red, 255; green, 255; blue, 255 }  ,fill opacity=1 ] (227.37,471.88) .. controls (227.37,467.04) and (231.29,463.13) .. (236.12,463.13) .. controls (240.95,463.13) and (244.87,467.04) .. (244.87,471.88) .. controls (244.87,476.71) and (240.95,480.63) .. (236.12,480.63) .. controls (231.29,480.63) and (227.37,476.71) .. (227.37,471.88) -- cycle ;
%Shape: Circle [id:dp6797875465681915] 
\draw  [fill={rgb, 255:red, 255; green, 255; blue, 255 }  ,fill opacity=1 ] (94.69,426.55) .. controls (94.69,421.72) and (98.6,417.8) .. (103.44,417.8) .. controls (108.27,417.8) and (112.19,421.72) .. (112.19,426.55) .. controls (112.19,431.38) and (108.27,435.3) .. (103.44,435.3) .. controls (98.6,435.3) and (94.69,431.38) .. (94.69,426.55) -- cycle ;
%Shape: Circle [id:dp6971374726562931] 
\draw  [fill={rgb, 255:red, 255; green, 255; blue, 255 }  ,fill opacity=1 ] (211.75,382.94) .. controls (211.75,378.1) and (215.66,374.19) .. (220.5,374.19) .. controls (225.33,374.19) and (229.25,378.1) .. (229.25,382.94) .. controls (229.25,387.77) and (225.33,391.69) .. (220.5,391.69) .. controls (215.66,391.69) and (211.75,387.77) .. (211.75,382.94) -- cycle ;
%Shape: Circle [id:dp11117625237799711] 
\draw  [fill={rgb, 255:red, 255; green, 255; blue, 255 }  ,fill opacity=1 ] (166.75,366.5) .. controls (166.75,361.67) and (170.67,357.75) .. (175.5,357.75) .. controls (180.33,357.75) and (184.25,361.67) .. (184.25,366.5) .. controls (184.25,371.33) and (180.33,375.25) .. (175.5,375.25) .. controls (170.67,375.25) and (166.75,371.33) .. (166.75,366.5) -- cycle ;
%Shape: Circle [id:dp05934627409547266] 
\draw  [fill={rgb, 255:red, 255; green, 255; blue, 255 }  ,fill opacity=1 ] (121.75,382.94) .. controls (121.75,378.1) and (125.67,374.19) .. (130.5,374.19) .. controls (135.34,374.19) and (139.25,378.1) .. (139.25,382.94) .. controls (139.25,387.77) and (135.34,391.69) .. (130.5,391.69) .. controls (125.67,391.69) and (121.75,387.77) .. (121.75,382.94) -- cycle ;
%Shape: Circle [id:dp37438968411757734] 
\draw  [fill={rgb, 255:red, 255; green, 255; blue, 255 }  ,fill opacity=1 ] (106.13,471.88) .. controls (106.13,467.04) and (110.05,463.13) .. (114.88,463.13) .. controls (119.71,463.13) and (123.63,467.04) .. (123.63,471.88) .. controls (123.63,476.71) and (119.71,480.63) .. (114.88,480.63) .. controls (110.05,480.63) and (106.13,476.71) .. (106.13,471.88) -- cycle ;
%Shape: Circle [id:dp3843927618853077] 
\draw  [fill={rgb, 255:red, 255; green, 255; blue, 255 }  ,fill opacity=1 ] (142.81,502.76) .. controls (142.81,497.93) and (146.73,494.01) .. (151.56,494.01) .. controls (156.39,494.01) and (160.31,497.93) .. (160.31,502.76) .. controls (160.31,507.6) and (156.39,511.51) .. (151.56,511.51) .. controls (146.73,511.51) and (142.81,507.6) .. (142.81,502.76) -- cycle ;
%Straight Lines [id:da6054582182564548] 
\draw [line width=3.5]    (175.5,366.5) -- (220.5,382.94) ;
%Straight Lines [id:da16472816673775736] 
\draw [line width=3.5]    (130.5,382.94) -- (120.57,398.95) -- (103.44,426.55) ;
%Shape: Circle [id:dp5659879449276575] 
\draw  [fill={rgb, 255:red, 255; green, 255; blue, 255 }  ,fill opacity=1 ] (94.69,426.55) .. controls (94.69,421.72) and (98.6,417.8) .. (103.44,417.8) .. controls (108.27,417.8) and (112.19,421.72) .. (112.19,426.55) .. controls (112.19,431.38) and (108.27,435.3) .. (103.44,435.3) .. controls (98.6,435.3) and (94.69,431.38) .. (94.69,426.55) -- cycle ;
%Shape: Circle [id:dp35710600905039946] 
\draw  [fill={rgb, 255:red, 255; green, 255; blue, 255 }  ,fill opacity=1 ] (121.75,382.94) .. controls (121.75,378.1) and (125.67,374.19) .. (130.5,374.19) .. controls (135.34,374.19) and (139.25,378.1) .. (139.25,382.94) .. controls (139.25,387.77) and (135.34,391.69) .. (130.5,391.69) .. controls (125.67,391.69) and (121.75,387.77) .. (121.75,382.94) -- cycle ;
%Shape: Circle [id:dp07586671886595198] 
\draw  [fill={rgb, 255:red, 255; green, 255; blue, 255 }  ,fill opacity=1 ] (166.75,366.5) .. controls (166.75,361.67) and (170.67,357.75) .. (175.5,357.75) .. controls (180.33,357.75) and (184.25,361.67) .. (184.25,366.5) .. controls (184.25,371.33) and (180.33,375.25) .. (175.5,375.25) .. controls (170.67,375.25) and (166.75,371.33) .. (166.75,366.5) -- cycle ;
%Shape: Circle [id:dp8159604540494534] 
\draw  [fill={rgb, 255:red, 255; green, 255; blue, 255 }  ,fill opacity=1 ] (211.75,382.94) .. controls (211.75,378.1) and (215.66,374.19) .. (220.5,374.19) .. controls (225.33,374.19) and (229.25,378.1) .. (229.25,382.94) .. controls (229.25,387.77) and (225.33,391.69) .. (220.5,391.69) .. controls (215.66,391.69) and (211.75,387.77) .. (211.75,382.94) -- cycle ;
%Shape: Circle [id:dp5418357150076692] 
\draw  [fill={rgb, 255:red, 255; green, 255; blue, 255 }  ,fill opacity=1 ] (235.69,424.55) .. controls (235.69,419.72) and (239.6,415.8) .. (244.44,415.8) .. controls (249.27,415.8) and (253.19,419.72) .. (253.19,424.55) .. controls (253.19,429.38) and (249.27,433.3) .. (244.44,433.3) .. controls (239.6,433.3) and (235.69,429.38) .. (235.69,424.55) -- cycle ;

\end{tikzpicture}
}
\vspace*{0.15mm}
\caption*{(b1) Ring species migration graph} \label{fig:RingSpeciesMG}  
  \end{subfigure}
  \qquad
\begin{subfigure}[b]{0.4\linewidth}
\centering
\tikzset{every picture/.style={line width=0.75pt}} %set default line width to 0.75pt        
\scalebox{0.65}{

\tikzset{every picture/.style={line width=0.75pt}} %set default line width to 0.75pt        

\begin{tikzpicture}[x=0.75pt,y=0.75pt,yscale=-1,xscale=1]
%uncomment if require: \path (0,600); %set diagram left start at 0, and has height of 600

%Shape: Circle [id:dp3813294298743828] 
\draw  [color={rgb, 255:red, 0; green, 0; blue, 0 }  ,draw opacity=1 ][fill={rgb, 255:red, 255; green, 255; blue, 255 }  ,fill opacity=1 ] (427.75,370.5) .. controls (427.75,365.67) and (431.67,361.75) .. (436.5,361.75) .. controls (441.33,361.75) and (445.25,365.67) .. (445.25,370.5) .. controls (445.25,375.33) and (441.33,379.25) .. (436.5,379.25) .. controls (431.67,379.25) and (427.75,375.33) .. (427.75,370.5) -- cycle ;
%Shape: Circle [id:dp6130160530243184] 
\draw   (472.75,386.94) .. controls (472.75,382.1) and (476.66,378.19) .. (481.5,378.19) .. controls (486.33,378.19) and (490.25,382.1) .. (490.25,386.94) .. controls (490.25,391.77) and (486.33,395.69) .. (481.5,395.69) .. controls (476.66,395.69) and (472.75,391.77) .. (472.75,386.94) -- cycle ;
%Shape: Circle [id:dp24870273603816584] 
\draw   (496.69,428.55) .. controls (496.69,423.72) and (500.6,419.8) .. (505.44,419.8) .. controls (510.27,419.8) and (514.19,423.72) .. (514.19,428.55) .. controls (514.19,433.38) and (510.27,437.3) .. (505.44,437.3) .. controls (500.6,437.3) and (496.69,433.38) .. (496.69,428.55) -- cycle ;
%Shape: Circle [id:dp46156034871254703] 
\draw   (488.37,475.88) .. controls (488.37,471.04) and (492.29,467.13) .. (497.12,467.13) .. controls (501.95,467.13) and (505.87,471.04) .. (505.87,475.88) .. controls (505.87,480.71) and (501.95,484.63) .. (497.12,484.63) .. controls (492.29,484.63) and (488.37,480.71) .. (488.37,475.88) -- cycle ;
%Shape: Circle [id:dp007085357011123783] 
\draw   (451.69,506.76) .. controls (451.69,501.93) and (455.61,498.01) .. (460.44,498.01) .. controls (465.27,498.01) and (469.19,501.93) .. (469.19,506.76) .. controls (469.19,511.6) and (465.27,515.51) .. (460.44,515.51) .. controls (455.61,515.51) and (451.69,511.6) .. (451.69,506.76) -- cycle ;
%Shape: Circle [id:dp9799954356136159] 
\draw   (403.81,506.76) .. controls (403.81,501.93) and (407.73,498.01) .. (412.56,498.01) .. controls (417.39,498.01) and (421.31,501.93) .. (421.31,506.76) .. controls (421.31,511.6) and (417.39,515.51) .. (412.56,515.51) .. controls (407.73,515.51) and (403.81,511.6) .. (403.81,506.76) -- cycle ;
%Shape: Circle [id:dp04397713413854232] 
\draw   (367.13,475.88) .. controls (367.13,471.04) and (371.05,467.13) .. (375.88,467.13) .. controls (380.71,467.13) and (384.63,471.04) .. (384.63,475.88) .. controls (384.63,480.71) and (380.71,484.63) .. (375.88,484.63) .. controls (371.05,484.63) and (367.13,480.71) .. (367.13,475.88) -- cycle ;
%Shape: Circle [id:dp6044023535660797] 
\draw   (382.75,386.94) .. controls (382.75,382.1) and (386.67,378.19) .. (391.5,378.19) .. controls (396.34,378.19) and (400.25,382.1) .. (400.25,386.94) .. controls (400.25,391.77) and (396.34,395.69) .. (391.5,395.69) .. controls (386.67,395.69) and (382.75,391.77) .. (382.75,386.94) -- cycle ;
%Straight Lines [id:da7188258366599721] 
\draw    (481.5,386.94) -- (505.44,428.55) ;
%Straight Lines [id:da44956679680167944] 
\draw    (460.44,506.76) -- (497.12,475.88) ;
%Straight Lines [id:da001949083372178806] 
\draw    (412.56,506.76) -- (460.44,506.76) ;
%Straight Lines [id:da6761883341021606] 
\draw    (375.88,475.88) -- (412.56,506.76) ;
%Straight Lines [id:da000052175204788063034] 
\draw    (364.44,430.55) -- (375.88,475.88) ;
%Straight Lines [id:da9167888498254576] 
\draw    (391.5,386.94) -- (436.5,370.5) ;
%Shape: Circle [id:dp8198143520547259] 
\draw  [fill={rgb, 255:red, 255; green, 255; blue, 255 }  ,fill opacity=1 ] (451.69,506.76) .. controls (451.69,501.93) and (455.61,498.01) .. (460.44,498.01) .. controls (465.27,498.01) and (469.19,501.93) .. (469.19,506.76) .. controls (469.19,511.6) and (465.27,515.51) .. (460.44,515.51) .. controls (455.61,515.51) and (451.69,511.6) .. (451.69,506.76) -- cycle ;
%Shape: Circle [id:dp32420413641899626] 
\draw  [fill={rgb, 255:red, 255; green, 255; blue, 255 }  ,fill opacity=1 ] (488.37,475.88) .. controls (488.37,471.04) and (492.29,467.13) .. (497.12,467.13) .. controls (501.95,467.13) and (505.87,471.04) .. (505.87,475.88) .. controls (505.87,480.71) and (501.95,484.63) .. (497.12,484.63) .. controls (492.29,484.63) and (488.37,480.71) .. (488.37,475.88) -- cycle ;
%Shape: Circle [id:dp13993747576546678] 
\draw  [fill={rgb, 255:red, 255; green, 255; blue, 255 }  ,fill opacity=1 ] (355.69,430.55) .. controls (355.69,425.72) and (359.6,421.8) .. (364.44,421.8) .. controls (369.27,421.8) and (373.19,425.72) .. (373.19,430.55) .. controls (373.19,435.38) and (369.27,439.3) .. (364.44,439.3) .. controls (359.6,439.3) and (355.69,435.38) .. (355.69,430.55) -- cycle ;
%Shape: Circle [id:dp616170874128534] 
\draw  [fill={rgb, 255:red, 255; green, 255; blue, 255 }  ,fill opacity=1 ] (472.75,386.94) .. controls (472.75,382.1) and (476.66,378.19) .. (481.5,378.19) .. controls (486.33,378.19) and (490.25,382.1) .. (490.25,386.94) .. controls (490.25,391.77) and (486.33,395.69) .. (481.5,395.69) .. controls (476.66,395.69) and (472.75,391.77) .. (472.75,386.94) -- cycle ;
%Shape: Circle [id:dp8401835455910411] 
\draw  [fill={rgb, 255:red, 255; green, 255; blue, 255 }  ,fill opacity=1 ] (427.75,370.5) .. controls (427.75,365.67) and (431.67,361.75) .. (436.5,361.75) .. controls (441.33,361.75) and (445.25,365.67) .. (445.25,370.5) .. controls (445.25,375.33) and (441.33,379.25) .. (436.5,379.25) .. controls (431.67,379.25) and (427.75,375.33) .. (427.75,370.5) -- cycle ;
%Shape: Circle [id:dp519356108444725] 
\draw  [fill={rgb, 255:red, 255; green, 255; blue, 255 }  ,fill opacity=1 ] (382.75,386.94) .. controls (382.75,382.1) and (386.67,378.19) .. (391.5,378.19) .. controls (396.34,378.19) and (400.25,382.1) .. (400.25,386.94) .. controls (400.25,391.77) and (396.34,395.69) .. (391.5,395.69) .. controls (386.67,395.69) and (382.75,391.77) .. (382.75,386.94) -- cycle ;
%Shape: Circle [id:dp8254099491404759] 
\draw  [fill={rgb, 255:red, 255; green, 255; blue, 255 }  ,fill opacity=1 ] (367.13,475.88) .. controls (367.13,471.04) and (371.05,467.13) .. (375.88,467.13) .. controls (380.71,467.13) and (384.63,471.04) .. (384.63,475.88) .. controls (384.63,480.71) and (380.71,484.63) .. (375.88,484.63) .. controls (371.05,484.63) and (367.13,480.71) .. (367.13,475.88) -- cycle ;
%Shape: Circle [id:dp26941884460777965] 
\draw  [fill={rgb, 255:red, 255; green, 255; blue, 255 }  ,fill opacity=1 ] (403.81,506.76) .. controls (403.81,501.93) and (407.73,498.01) .. (412.56,498.01) .. controls (417.39,498.01) and (421.31,501.93) .. (421.31,506.76) .. controls (421.31,511.6) and (417.39,515.51) .. (412.56,515.51) .. controls (407.73,515.51) and (403.81,511.6) .. (403.81,506.76) -- cycle ;
%Straight Lines [id:da08683317971789906] 
\draw    (436.5,370.5) -- (481.5,386.94) ;
%Straight Lines [id:da24030166085561155] 
\draw    (391.5,386.94) -- (364.44,430.55) ;
%Shape: Circle [id:dp8502103352142811] 
\draw  [fill={rgb, 255:red, 255; green, 255; blue, 255 }  ,fill opacity=1 ] (355.69,430.55) .. controls (355.69,425.72) and (359.6,421.8) .. (364.44,421.8) .. controls (369.27,421.8) and (373.19,425.72) .. (373.19,430.55) .. controls (373.19,435.38) and (369.27,439.3) .. (364.44,439.3) .. controls (359.6,439.3) and (355.69,435.38) .. (355.69,430.55) -- cycle ;
%Shape: Circle [id:dp8592554295766601] 
\draw  [fill={rgb, 255:red, 255; green, 255; blue, 255 }  ,fill opacity=1 ] (382.75,386.94) .. controls (382.75,382.1) and (386.67,378.19) .. (391.5,378.19) .. controls (396.34,378.19) and (400.25,382.1) .. (400.25,386.94) .. controls (400.25,391.77) and (396.34,395.69) .. (391.5,395.69) .. controls (386.67,395.69) and (382.75,391.77) .. (382.75,386.94) -- cycle ;
%Shape: Circle [id:dp5901493529296267] 
\draw  [fill={rgb, 255:red, 255; green, 255; blue, 255 }  ,fill opacity=1 ] (427.75,370.5) .. controls (427.75,365.67) and (431.67,361.75) .. (436.5,361.75) .. controls (441.33,361.75) and (445.25,365.67) .. (445.25,370.5) .. controls (445.25,375.33) and (441.33,379.25) .. (436.5,379.25) .. controls (431.67,379.25) and (427.75,375.33) .. (427.75,370.5) -- cycle ;
%Shape: Circle [id:dp3728140683901321] 
\draw  [fill={rgb, 255:red, 255; green, 255; blue, 255 }  ,fill opacity=1 ] (472.75,386.94) .. controls (472.75,382.1) and (476.66,378.19) .. (481.5,378.19) .. controls (486.33,378.19) and (490.25,382.1) .. (490.25,386.94) .. controls (490.25,391.77) and (486.33,395.69) .. (481.5,395.69) .. controls (476.66,395.69) and (472.75,391.77) .. (472.75,386.94) -- cycle ;
%Shape: Circle [id:dp20059942977723366] 
\draw  [fill={rgb, 255:red, 255; green, 255; blue, 255 }  ,fill opacity=1 ] (496.69,428.55) .. controls (496.69,423.72) and (500.6,419.8) .. (505.44,419.8) .. controls (510.27,419.8) and (514.19,423.72) .. (514.19,428.55) .. controls (514.19,433.38) and (510.27,437.3) .. (505.44,437.3) .. controls (500.6,437.3) and (496.69,433.38) .. (496.69,428.55) -- cycle ;

\end{tikzpicture}
}
\vspace*{0.15mm}
\caption*{(b2) Ring species interbreeding graph} \label{fig:RingSpeciesIG}  
  \end{subfigure}
\caption{Migration and interbreeding graphs corresponding to intransitive equilibria. Intransitive interbreeding graphs ((a2): the Friendship graph) can emerge in complete and symmetric migration (a1). A ring of populations connected through migration (b1) can give rise to the ring species interbreeding graph (b2): the terminal forms of the ring species complex are reproductively isolated, despite ongoing gene flow through the chain of intermediary populations.
} \label{fig:MigGraphs}  
\end{figure}  

\begin{figure}
    \centering
    \includegraphics[width=0.4\textwidth]{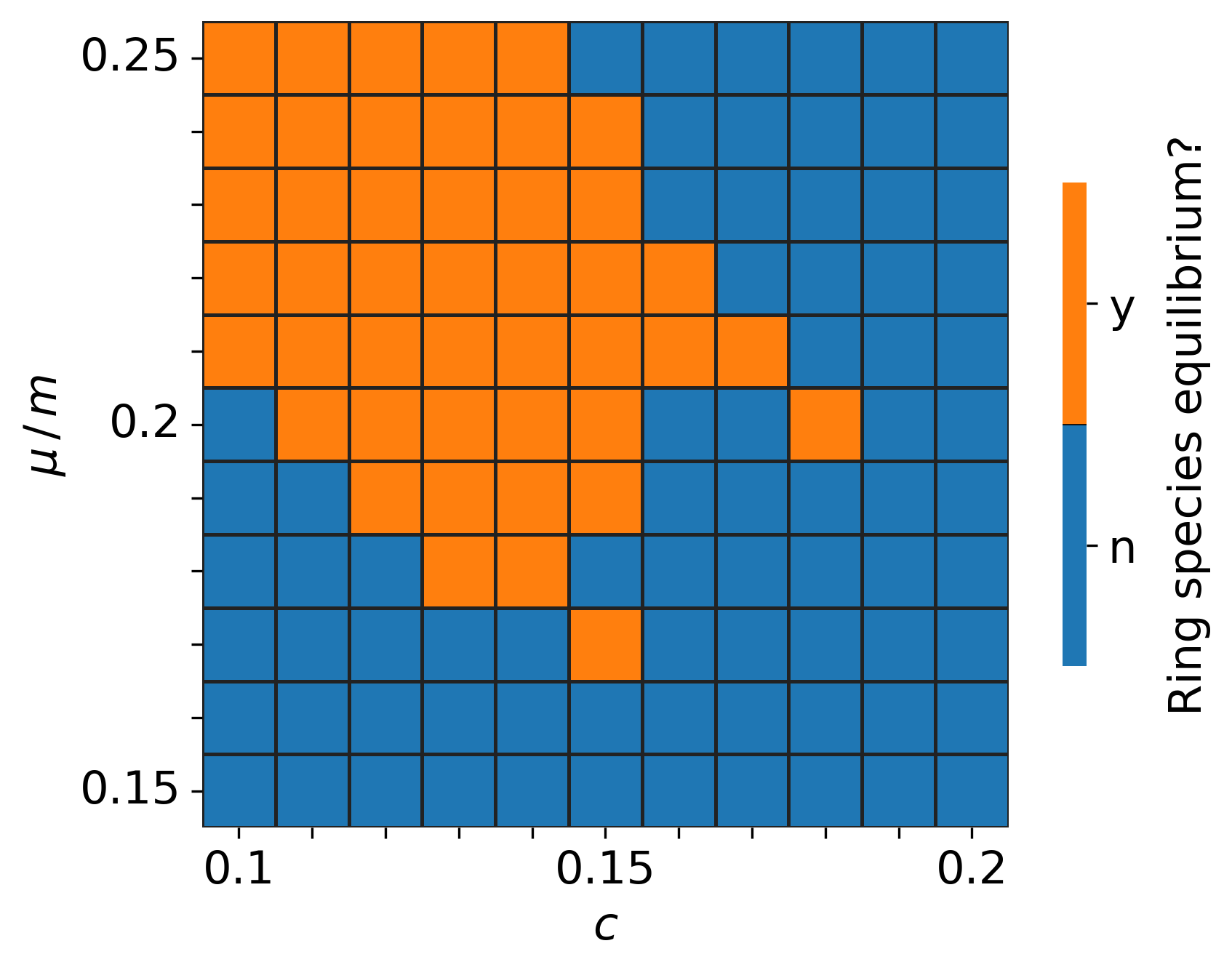}
    \setlength{\belowcaptionskip}{-10pt}
    \captionsetup{font=small}
    \caption{Existence of stable ring species equilibria depending on the mutation / migration ratio and the threshold value. We performed a systematic root search (as described in Fig. \ref{fig:asym_N}), but for ring species equilibria. Here, we set $h(x) = \mathbf{1}_{\{ x \geq c\}} \frac{x - c}{1 - c}, N = 6.$}
    \label{fig:ring_species}
\end{figure}

\noindent
\section{Subspecies clustering}
\label{sec:subsp_clust}
Partial reproductive isolation (PRI) refers to a situation where some \hypertarget{clust}{clusters} of populations, that can be viewed as ``subspecies'', retain some ability to interbreed \textit{within} them but face reproductive barriers that limit gene flow \textit{between} them. Predominantly seen as a stage in the speciation process, the evolution of PRI has recently received attention in light of the hypothesis that it could to the contrary represent in some situations a stable evolutionary endpoint (see \cite{servedio2020evolution}). This interest is sparked in part by empirical findings that suggest that ongoing hybridization between species is taxonomically widespread, begging the question: If a species is composed of population clusters of increased genetic similarity (within them)  that still exchange genes (between them) at low rates, do we generically observe speciation in progress, or can such genetic inhomogeneities within a species persist on an evolutionary time scale?

\begin{figure}
    \centering
    \includegraphics[width=0.5\textwidth]{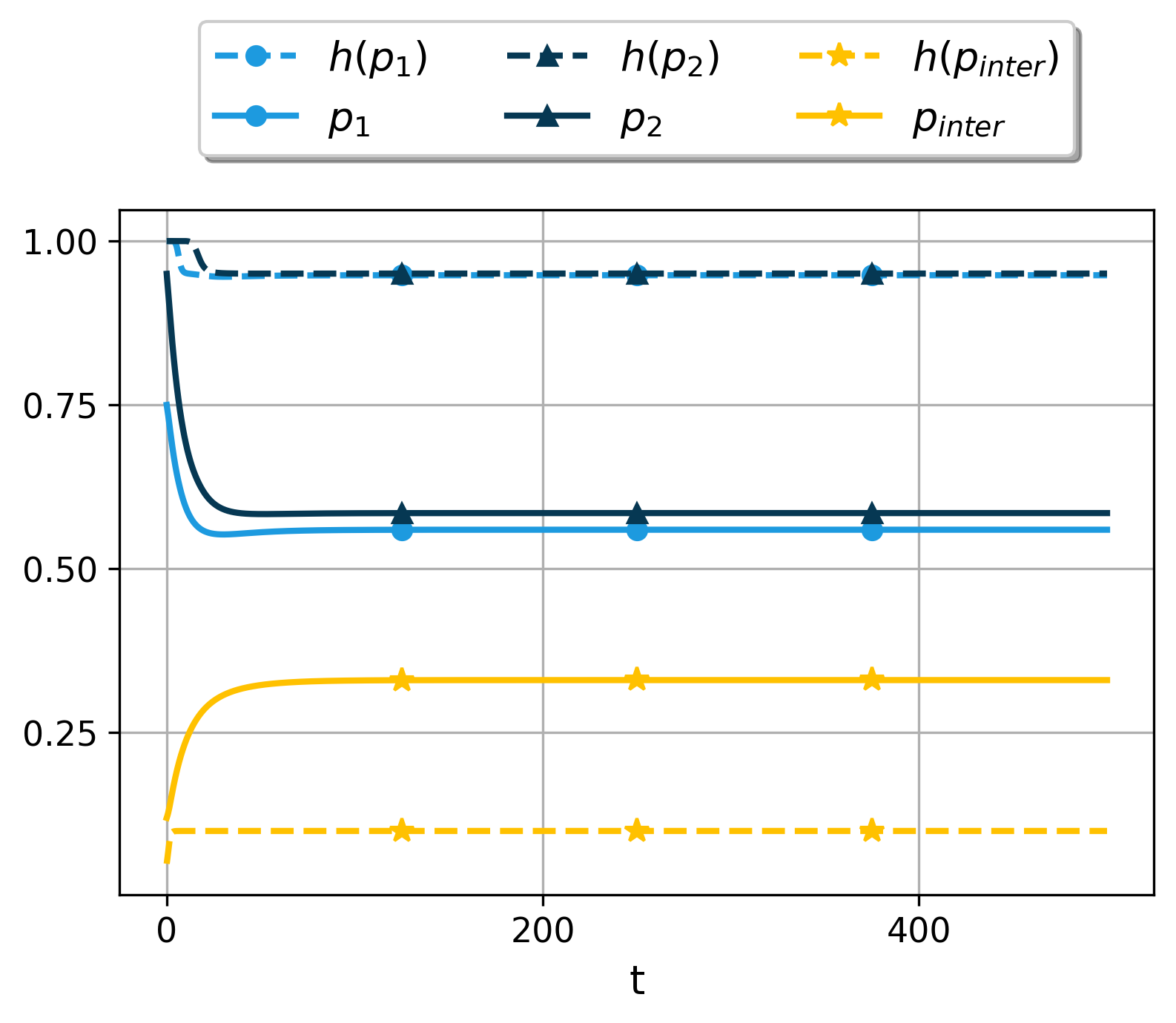}
    \setlength{\abovecaptionskip}{0pt}
    \setlength{\belowcaptionskip}{-5pt}
    \captionsetup{font=small}
    \caption{Subspecies clustering equilibrium in the  uniform case ($M_{i j} = m$ for all $i\neq j$). The plot shows genetic proximities over time. {The clusters $V_1$ and $V_2$ show the same degree of interbreeding ($p_1\approx p_2\approx 0.6$) within them, with a lower degree of interbreeding between them ($p_{inter}\approx 0.3$.} Here, we considered the feedback function $h_2$ displayed on the right of Fig. \ref{fig:asym_N}, and $\mu = 0.01, m = 0.02, |V_1| = 4, |V_2| = 6$. }
    \label{fig:asym_eq}
\end{figure}

To address this question, we consider the simplest migration setting with uniform migration ($M_{i j} = m$ for every $i,j$).
By considering (\ref{eq:ODE}), we first see that a uniform vector $P^\text{eq}$, i.e.,  such that for all $i \neq j$,
\begin{align}\label{eq:sym_eq}
    P^\text{eq}_{i j} = p^\text{eq} \, ,
\end{align}
is a stable equilibrium if and only if the following two conditions are satisfied:
\begin{align}
    h(p^\text{eq}) (1 - p^\text{eq}) = \frac{\mu}{m} p^\text{eq} \, ,\label{eq:ODE_sym}
\end{align}
giving the equilibrium property, and
\begin{align}
    h'(p^\text{eq}) (1 - p^\text{eq}) - h(p^\text{eq}) < \frac{\mu}{m} \, ,\label{eq:ODE_sym_eq}
\end{align}
giving local stability. This result holds for any $N$ and can be deduced by considering the Jacobian of $\vec F$, see Proposition \ref{prop:stabcplt} for details and its proof for a rigorous mathematical derivation. Note that equation \eqref{eq:ODE_sym} is the fixed-point problem (\ref{eq:FxPtPb_dim2}) when $N=2$ and that either condition \eqref{eq:ODE_sym} and \eqref{eq:ODE_sym_eq} is independent of $N$. A natural question is whether there exist transitive equilibria that do not satisfy the property (\ref{eq:sym_eq}), that is,
species complexes with several groups of populations exhibiting higher genetic similarity within groups than between them (partial reproductive isolation).

\begin{figure}[t]
    \centering
    \includegraphics[width=0.5\textwidth]{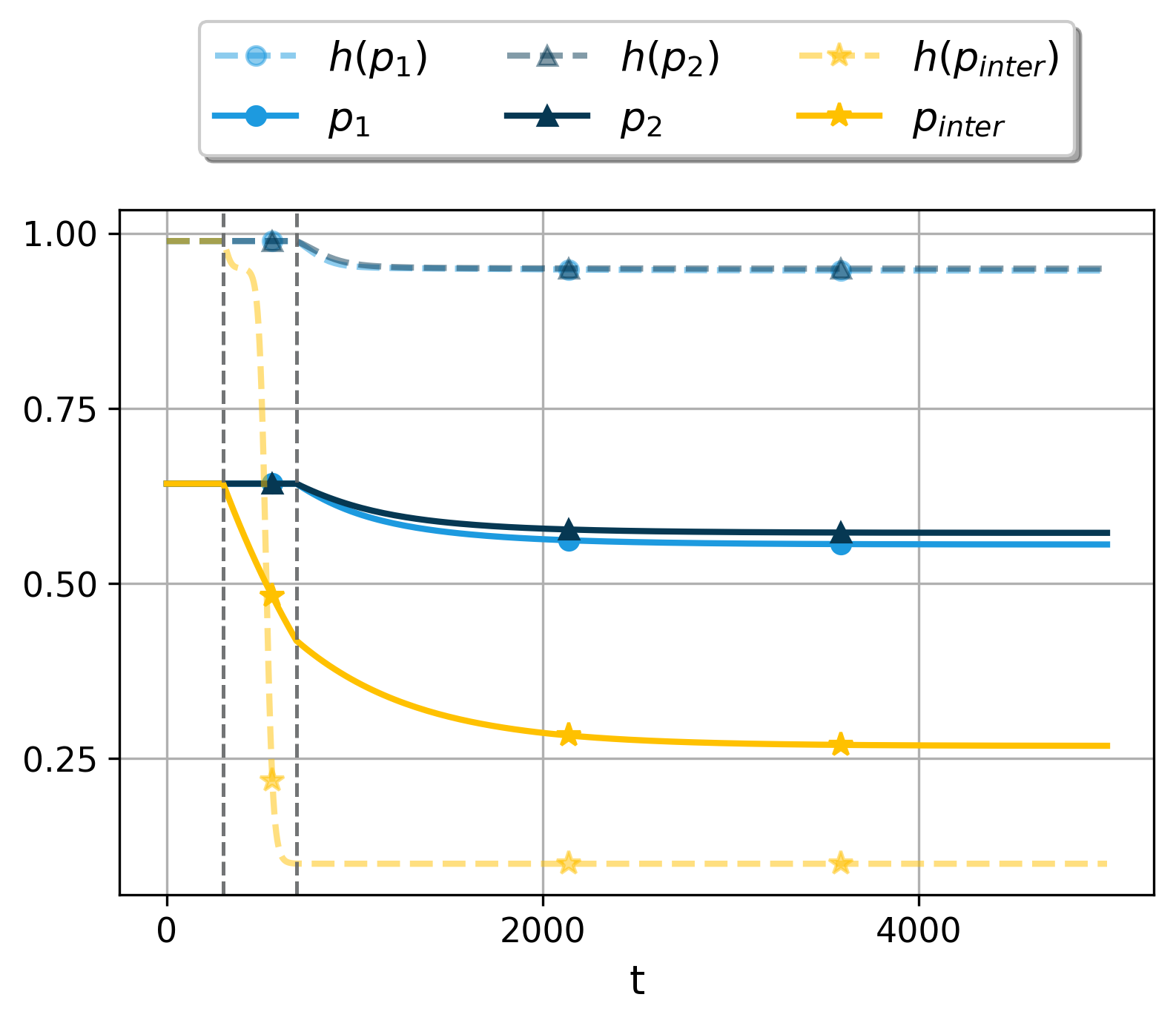}
    \setlength{\abovecaptionskip}{0pt}
    \setlength{\belowcaptionskip}{-5pt}
    \captionsetup{font=small}
    \caption{Multi-stability: Transition from symmetric to clustering equilibrium. In the time interval between the dotted vertical lines, migration rates between the nodes of $V_1$ and $V_2$ are set to zero. The plot shows genetic proximities over time. Here, we considered the feedback function $h_2$, displayed on the right of Fig. \ref{fig:asym_N}, and $\mu = 0.0055, m = 0.01, |V_1| = 3, |V_2| = 4$.}
    \label{fig:plast}
\end{figure}

\begin{figure*}[t]
    \centering
    \includegraphics[width=\textwidth]{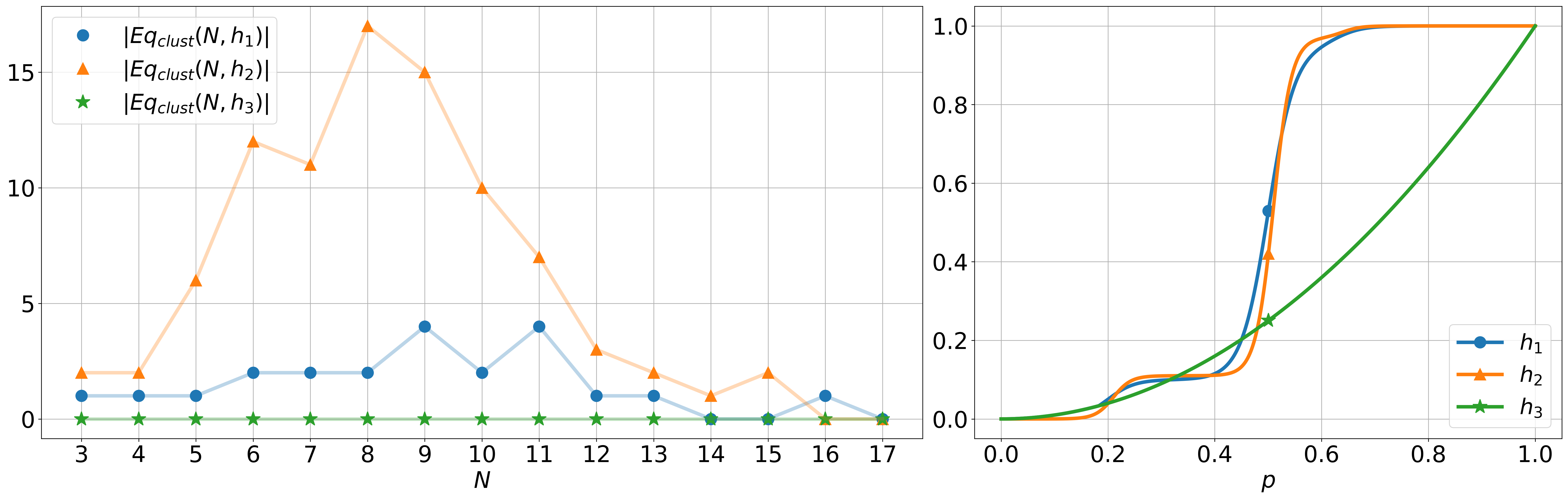}
    \captionsetup{font=small}
    \caption{Disappearing of inhomogeneous equilibria, when $N$ becomes large. On the left, we performed a systematic search for inhomogeneous roots of the function $\vec F$ from (\ref{eq:ODE}), using a L-BFGS-B optimization algorithm. Then, we tested the roots stability by simulating the ODE (\ref{eq:ODE}), using the potential asymmetric root as initial position. The number of roots displayed corresponds to the number of different stable species graphs that were found. If two equilibria correspond to the same species graph (up to a permutation of nodes), they are not counted twice. On the right, we plotted different feedback functions, given by two smoothed versions of a step-function with jumps at $(s_1, s_2, s_3) = (0.12, 0.5, 0.63)$ to the steps $(y_1, y_2, y_3) = (0.1, 0.85, 1)$, and $h_3(x) = x^2$. The functions $h_1$ and $h_2$ differ mainly in their behavior between 0 and $s_1$, with $h_1$ not having a threshold and decaying like $x^2$, and $h_2$ having a threshold. Further, we chose $\mu = 0.1, m = 0.42$.}
    \label{fig:asym_N}
\end{figure*}

In Fig. \ref{fig:asym_eq} we consider the feedback $h_2$ as in Fig. \ref{fig:asym_N}. Intuitively, this function can be thought of as representing incompatibilities that arise in stages, with each plateau being interpreted as a degree of genetic incompatibility.

We now consider a case where $\{1,\dots,N\}$ is split into two sets of vertices $V_1$ and $V_2$. We then consider equilibria $P^\text{eq}$ with three degrees of freedom, namely, the genetic proximity within $V_1$ (denoted by $p_1$), the genetic proximity within $V_2$ (denoted by $p_2$), and the genetic proximity between $V_1$ and $V_2$ (denoted by $p_{inter}$).  In Fig. \ref{fig:asym_eq}, we observe the existence of stable equilibria with $p_{1},p_2 > p_{inter}$, thus showing that partially isolated clusters can coexist within the same species. An analytical treatment of this phenomenon is given in SI, see Proposition \ref{prop:sym_break_1}. 

\medskip

In Fig \ref{fig:plast}, we show how partial reproductive isolation can emerge from temporary geographic isolation. Namely, consider the same splitting of $\{1,\dots, N\}$ into $V_1$ and $V_2$ as above, and genetic proximities at a uniform equilibrium at time $t=0$. At time $T>0$, we impose isolation in a time window of duration $t_{\text{stress}}$ so that we set $M_{i j} = 0$ if $i$ and $j$ belong to different $V_k$, for $k = 1,2$. At time $T + t_{\text{ stress}}$, we reestablish complete migration (i.e.,  $M_{i j} = m$). When carefully choosing the size of the isolation window given by $t_{\text{stress}}$, the genetic proximities converge to an asymmetric equilibrium. In fact, it suffices to choose the time window of isolation such that the genetic proximity between $V_1$ and $V_2$ falls into the basin of attraction of the asymmetric equilibrium at the end of the isolation window (see also end of Section \ref{sec:2pop}). Notice that the genetic proximity inside each group of vertices remains unchanged during the isolation window, because uniform equilibria are independent of the number of populations, see equation (\ref{eq:ODE_sym}). 

\section{Large metapopulations}
\label{chap:large_pops}
The previous two sections have demonstrated that species can exhibit complex interbreeding structures. First, we showed that when a speciation threshold is present, species graphs can be intransitive.  
Second, we identified scenarios in which populations consistently interbreed while forming ``subspecies" clusters that remain in partial isolation. 

But to what degree are such features generic vs. idiosyncratic? We will argue that while such configurations can persist in small metapopulations, 
they actually get rarer as an effect of large metapopulation sizes, where species complexes tend to become increasingly coherent, transitive, and homogeneous.

\begin{figure*}[t]
    \centering
    \includegraphics[width=\textwidth]{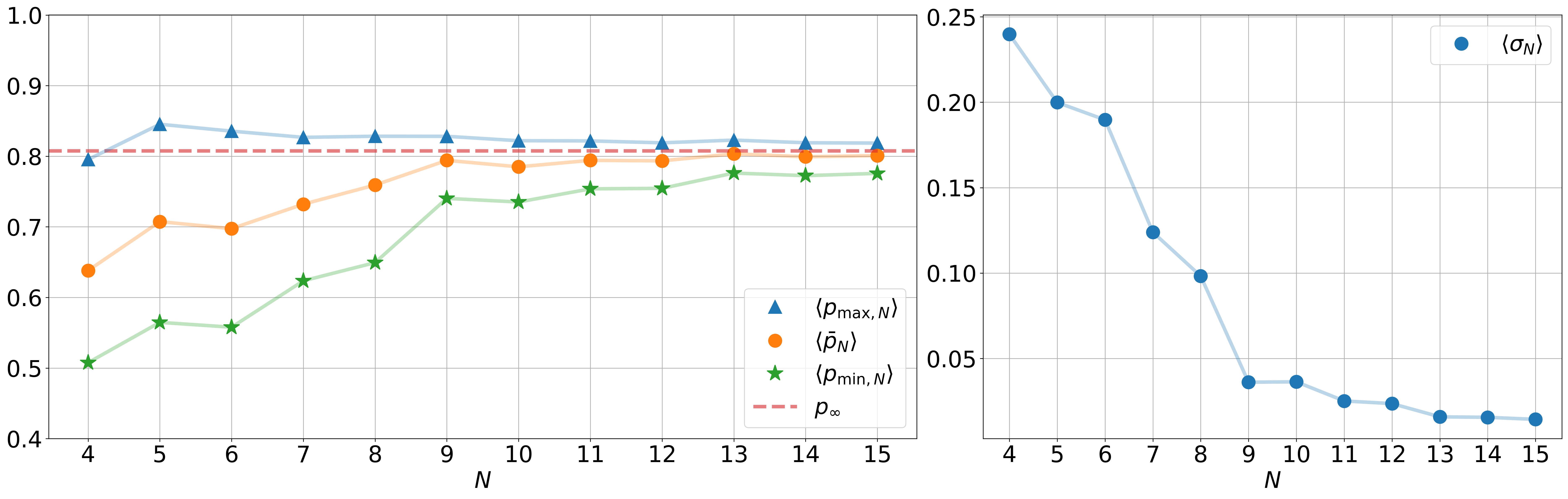}
    \captionsetup{font=small}
    \caption{Convergence to uniform equilibrium. On the left, we plotted means of different measures of the genetic proximities over 50 runs, which differed by migration rates and initial conditions to the ODE (\ref{eq:ODE}),  which were drawn independently from a rescaled beta distribution $m_\T{rate} \cdot \beta(0.5,0.5)$ distribution, with $m_\T{rate} = 0.84$. Here, $p_{\max,N}$, respectively $p_{\min,N}$, denotes the minimum, respectively maximum, of the genetic proximities at equilibrium  $((P^{\text{eq}}_{i j}))_{1\le i,j\le N}$. Further, $\bar p_{N}$ denotes the average genetic proximity and $p_{\infty}$ the genetic proximity of the uniform equilibrium associated to (\ref{eq:average}). On the right, we plotted the mean of the empirical standard deviation (normalised by the corresponding $\bar p_{N}$). The feedback function $h$ was chosen as the function $h_1$ in Fig. (\ref{fig:asym_N}). Additionally, we chose $\mu = 0.1$.}
    
    \label{fig:cvg_eq_N}
\end{figure*}

\medskip

We begin by considering the case of uniform migration. Previously, we showed that a suitable choice of the feedback function enables the existence of exotic equilibria such as intransitive interbreeding structures (friendship graphs) or species with clusters in
partial reproductive isolation (subspecies clustering). However, we show in SI that these specific inhomogeneous equilibria can only exist for small values of $N$ (see Propositions \ref{prop:sym_break_1}, \ref{prop:sym_break_2}).

In Fig. \ref{fig:asym_N}, we perform a systematic search of inhomogeneous equilibria when migration is uniform. As conjectured, we observe the existence of a critical size $N_c$, such that for $N>N_c$, the ODE system (\ref{eq:ODE}) only exhibits uniform stable equilibria, which indicates that the clustering effect previously observed can only hold for small populations (and presumably for a suitable choice of $h$).

In fact, we believe that the absence of clustering is valid not only for uniform migration, but also for a much  broader class of migration rates $(M_{ij})$. 
To test this hypothesis, we consider metapopulations of size $N$, where the migration rates $M_{i j}$, are drawn at random from the same distribution. For the sake of illustration, we assumed a $U$-shaped distribution $\beta(0.5,0.5)$
which puts more mass around the values 0 and 1, generating a strongly heterogeneous migration structure.
In Fig. \ref{fig:cvg_eq_N}, we observe that as $N$ gets large, the system equilibrates at a quasi-uniform state.

Biologically speaking, this result suggests that most large species complexes should form rather simple and coherent structures. In particular, it follows that the specific migration rate between populations $i$ and $j$ does not have a strong influence on their genetic incompatibility. Intuitively, this can be understood from the fact that the main contribution to gene flow between $i$ and $j$ occurs through long and indirect paths. In fact, even if a significant geographical constraint substantially impedes direct gene exchange between the two populations, in
a large migration network the constraint is bypassed through enough indirect paths (i.e., gene exchanges through many intermediary populations) between $i$ and $j$. In this view, the gene flow between $i$ and $j$ should only ``feel" the average migration rate 
\begin{equation}
\label{eq:average}
m = \mathbb{E}[M_{ij}]
\end{equation}
This heuristic is confirmed by Fig. \ref{fig:cvg_eq_N}, where the quasi-uniform equilibrium in a population with heterogeneous migration rates is well approximated by the uniform equilibrium of a uniform migration model with equal rates (\ref{eq:average}).

\medskip

How can we understand 
this homogenization effect in general species complexes (and not only random)?
We now argue that if we make the further assumption that the $M_{ij}h(P^{\text{eq},N}_{i j})$'s are uniformly bounded from below, then the equilibrium can only be uniform despite the potential asymmetry of the migration network.
In other words, if we restrict ourselves to the class of equilibria with a condition of minimal effective migration rates between any pair of populations, then the equilibria must be uniform.

Heuristically, this surprising result is due to the fixed point property (\ref{eq:FxPtPb}), and to the fact that random walks on a large, well-connected graph reach their invariant distribution very quickly. More precisely, assume that for all $N$ and all $i,j$  $M_{ij} h(P_{ij}^{\text{eq},N})>b$ for some constant $b$. Such well-connected graphs actually form a simple example of
a family of expander graphs (see \cite{berestycki2016mixing}, p.38ff).  Random walks on expander graphs attain their invariant distribution much faster than the time it takes two random walks to coalesce (this statement can be made rigorous by letting $N\to \infty$, see \cite{aldous2002reversible} or \cite{cooper2013coalescing}, p.4 for coalescing times, and \cite{berestycki2016mixing}, p.40). Since the invariant distribution is independent of the starting position, this suggests that by the time the two random walks coalesce, they have forgotten their initial position. Thus, the fixed point property (\ref{eq:FxPtPb}) would yield that $P^{\text{eq},N}_{i j}$ is the same for any $i ,j$, and therefore uniform. 
Furthermore, as we have seen in (\ref{eq:average}), the effect of homogenization is twofold in random networks. Not only are species complexes homogeneous at equilibrium, but an extra averaging effect on the $M_{ij}$'s allows us to deduce the genetic distances from the fixed point equation
\begin{align}
    h(p^\T{eq}) (1 - p^\T{eq}) = \frac{\mu}{\bar{m}} p^\T{eq} \, , \label{eq:sym_eq_avg}
\end{align}
where $\bar m$ is the migration rate averaged over all $M_{ij}$. 

\section{Fluctuating migration}
\label{sec:fluct}
 
In the previous section, we demonstrated that large species complexes form increasingly (with $L$) coherent and transitive entities, even with an inhomogeneous migration structure. A natural question with multistable dynamical systems concerns the crossing of basins of attraction and translates in the context of speciation studies, into the following question: Which environmental conditions are required to escape the coherent, quasi-uniform  equilibrium, and initiate a speciation event?

To investigate this question, we consider a version of our model in which migration rates change over time by resampling them at rate $\theta > 0$.

Intuitively, the homogenization effect that we uncovered for large static networks suggests that large species tend to form homogeneous structures. Thus, if resampling only impedes the migration rate between two populations, the loss in direct gene exchange is compensated by indirect migration paths (i.e., gene exchanges through intermediary populations). Thus, we expect speciation to predominantly occur when $(A)$ a single population $i$ gets isolated by chance from the rest of the complex, that is, when all the migration rates $M_{ij}$ and $M_{ji}$ are small, and $(B)$ this transitory isolation is maintained for a sufficiently long duration for the speciation process to start. For large populations, this requires the coordination of many independent events so that speciation time should sharply increase with $N$. Furthermore, speciation events should typically involve a single population detaching from the species complex and forming its own species, since larger groups of detaching populations require more migration rates satisfying $(A)$ and $(B)$ above. This indicates that upon speciation, we can identify a mother species (the large component) and a daughter species (the small component), resulting in peripatric speciation, where the large and small complexes will continue to exchange some genes during divergence. 

Is this intuition reflected in simulations? As in the previous section, we first sampled the migration rates from a (rescaled) $\beta(0.5,0.5)$ distribution. Then, we estimated the distribution of the first time to speciation $\tau$, that is the random time at which the species complex breaks into more than one connected component. The results displayed in Fig. \ref{fig:time_spec} indicate that the time to speciation increases sharply with the number of populations.

Additionally, our initial intuition about the number of detaching populations is also confirmed from simulations, where we observed that speciation events typically involved a small cluster involving one or only a few populations detaching from the species complex and forming its own species (see Fig. \ref{fig:fluct_mig_dyn} in SI for a typical realization of the dynamics of genetic proximities with fluctuating migration rates, eventually resulting in such peripatric speciation).
The expected size of this speciating cluster as a function of $N$ ranges between $1$ and $1.25$, and stabilizes at $1$ for larger $N$ ($N \geq 10$, see Fig. \ref{fig:spec_prob_nbrs} in SI).

\medskip

We investigate further the behavior of the speciation time in terms of the resampling rate $\theta$, and different migration update distributions. The simulations displayed in Fig. \ref{fig:spec_time_theta} suggest that there exists a value $\theta_{\max}(N, m) > 0$, such that the speciation probability is at its maximum. When the rate of change $\theta$ of the environment is too low, speciation times are long, since the coordination of the events leading to an isolated population is slow. In other words, the time it takes to realize migration rates satisfying condition $(A)$ above is long. Surprisingly at first glance, the speciation probability also decreases sharply when the rate of change of the environment is high, i.e., when migration rates are updated very frequently. This can be explained heuristically by noting that in order to trigger a speciation event, geographic restrictions must be upheld for some time, allowing the positive feedback loop between genetic distance and effective migration rate to kick in. If migration rates are updated too quickly, the geographical constraints required for speciation will disappear too quickly for substantial divergence to occur, thus violating condition $(B)$ above. Starting from random initial conditions, the system quickly stabilizes around the uniform quasi-equilibrium given by the solution to \eqref{eq:sym_eq_avg}, with $\bar m$ equal to the expectation of the migration update distribution.

\begin{figure}
    \centering
    \includegraphics[width=0.47\textwidth]{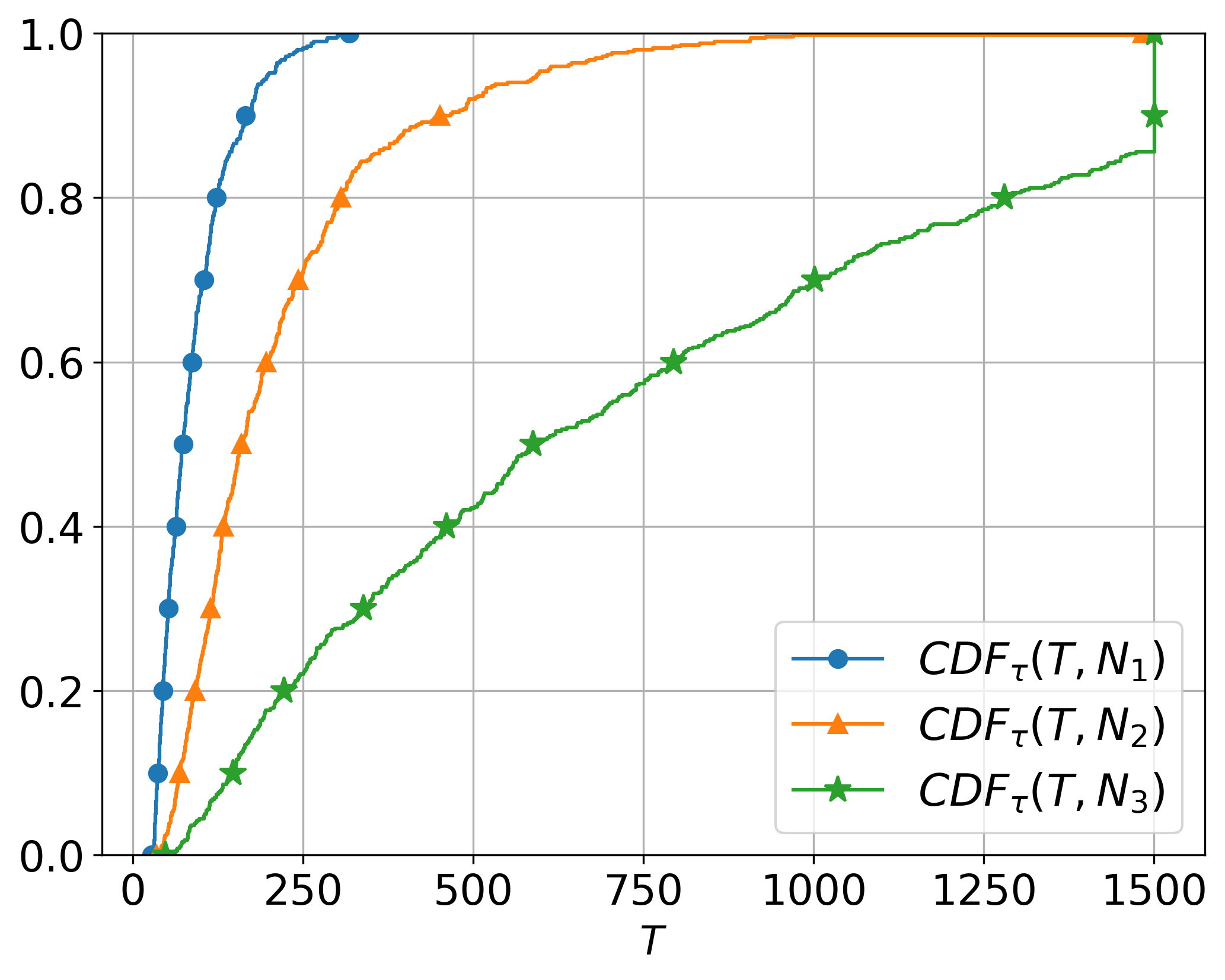}
    \captionsetup{font=small}
    \caption{Empirical cumulative distribution functions of the speciation time for different metapopulation sizes over 500 runs. We considered dynamically changing migration rates updated according to exponential clocks with rate $\theta$ and resampled independently from a (rescaled) Beta distribution $m_{\text{rate}} \cdot \beta(0.5, 0.5)$ with $m_\T{rate} = 1.675$. We plotted the empirical cumulative distribution functions of the speciation time for different metapopulation sizes, given by $N_1 = 4, N_2 = 6, N_3 = 8$. For $N_3 = 8$, in about 15\% of simulations, no speciation event occurred. Further, we chose $\mu = 0.1, \theta = 1$. Additional simulations revealed that for $N=15$, speciation occurred in only one run out of 100.}
    \label{fig:time_spec}
\end{figure}

\begin{figure}
    \centering
    \includegraphics[width=0.47\textwidth]{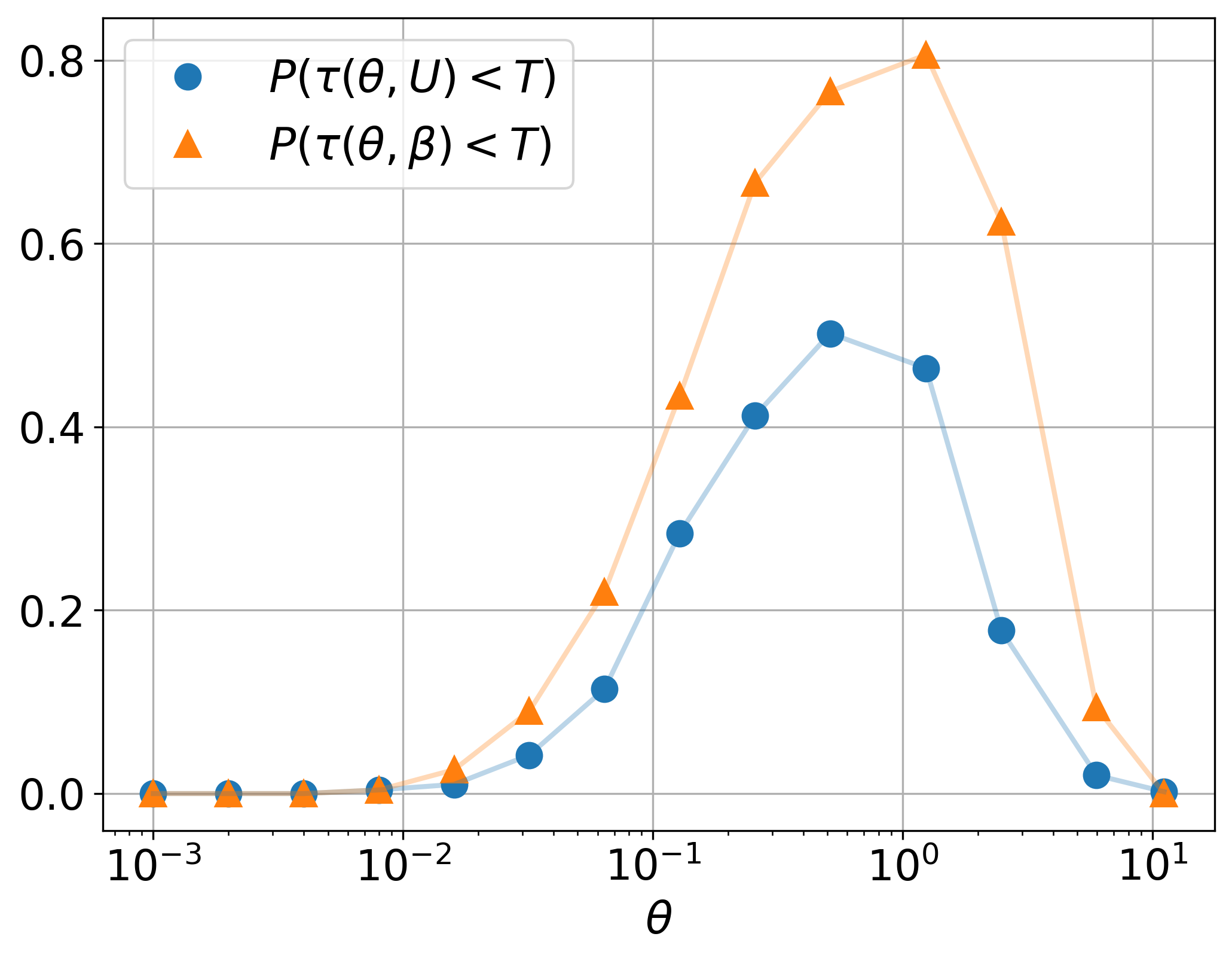}
    \captionsetup{font=small}
    \caption{Speciation probability for different migration update distributions as a function of the update rate of migration rate $\theta$, averaged over 500 runs. The values $\tau(\theta, U)$, resp. $\tau(\theta, \beta)$, refer to the time of speciation with a migration update distribution given by $m_{\text{rate}} \cdot \mathcal{U}([0,1])$ resp. $m_{\text{rate}} \cdot \beta(0.5, 0.5)$, with $m_{\text{rate}} = 1.675$. Here, we chose $h(x) = x^3, N = 5, \mu = 0.1$ and $T = 200$.}
    \label{fig:spec_time_theta}
\end{figure}

Further, Fig. \ref{fig:spec_time_theta} shows that for small $N$, the speciation time depends heavily on the update distribution. Choosing a $\beta(0.5, 0.5)$ distribution as the update law, results in higher speciation probabilities than choosing a uniform distribution. This can be explained from the fact that a $\beta(0.5, 0.5)$ is a $U$-shaped distribution that produces values close to $0$ with higher probability, and thus favors the occurrence of small migration rates which are more likely to trigger speciation events.

In summary, our results suggest that large species complexes experience lower speciation rates. This phenomenon owes to the large number of indirect paths that connect any two populations, even in the presence of a direct geographic barrier between them (see Section \ref{chap:large_pops}), which results in a robust ability to exchange alleles. For the occurrence of a speciation event, two conditions are necessary: $(A)$ a small cluster gets isolated by chance from the rest of the complex, and $(B)$ this transitory isolation is maintained for a sufficiently long duration for divergence feedback to kick in and reproductive isolation to be completed. This reasoning implies that in nature, besides the rare peripatric speciations we described, speciation
events occurring in large metapopulations must involve selective mechanisms (e.g., assortative mating, divergent selection) or massive geographic restrictions on
gene flow (e.g., vicariance).

\section{Discussion}

What are the causes of varying speciation rates across the tree of life? Numerous empirical studies have explored the biotic and abiotic determinants of speciation rates, including differences in geographic regions, life-history traits, intraspecific genetic diversity or species range (see \cite{rabosky2016reproductive, simpson1984tempo, coyne2004speciation, mayr1963animal, gavrilets2000patterns}).  This paper addresses the question by examining the effect of \hyperlink{neutral}{neutral} mutation/migration processes with \hyperlink{feedback}{divergence feedback} on the dynamics of species complexes.

\paragraph{Mathematical analysis of the model.}

\noindent We presented a genetic distance model previously introduced by \cite{couvert2024opening} to study the evolution of genetic differences between $N$ monomorphic populations at $L$ loci. This stochastic model is parameterized by three parameters: the mutation rate $\mu$, the migration matrix $(M_{i j})_{i,j \neq N}$, and a continuous function $h:[0,1]\rightarrow [0,1]$. The rate of effective migration events between two populations depends on the ecology and topology of the metapopulation (through the migration matrix $M_{i j}$) and on the genetic distance between them (through the function $h$). Referred to as the \hyperlink{feedback_func}{feedback function}, $h(p)$ encodes the extent to which a given degree of divergence (represented by the genetic proximity $p$) between two populations reduces the effective migration rates between them.
Modeling the decrease of gene flow with divergence by a general function $h$ has several important benefits. It contrasts with the standard one-off modeling choice where $h(p)=0$ when $p$ is below some threshold $c$ and $h(p)=1$ when $p>c$, which does not offer the possibility of studying the \hyperlink{feedback}{feedback effect} of divergence on gene flow. The species complex naturally evolves to some emerging equilibrium balancing homogenization by migration and differentiation by mutation, one of the possible equilibria being total reproductive isolation.

The genetic proximity $p$ is the fraction of shared alleles at a fixed set of $L$ loci, where $L$ needs to be large for our analysis to work. This set can be thought of as a set of speciation loci  potentially responsible for Dobzhansky-Muller  incompatibilities and/or underlying traits involved in reproductive isolation (mating trait, foraging trait...). It can also be viewed as the vast collection of sites that vary neutrally. In the latter case, $1-p$ can be interpreted as the fraction of synonymous positions that differ, which serves as a natural proxy for interspecific divergence. In all cases, for reproductive isolation to be completed, it is not necessary that $p=0$; if $h$ is zero below some threshold $c$, it is sufficient that $p<c$. We expect this to be the most realistic situation, with $-\log (1-c)$ in the range of $(1.5,2.5)$ \cite{roux2016shedding} in the case of synonymous divergence.

When the number of loci is large, this stochastic model can be well approximated by the solution to a nonlinear ordinary differential equation (ODE) involving only pairwise genetic proximities. This represents a considerable reduction in dimension, going from a stochastic process with values in vectors of $L$ (allelic) partitions of $\{1,\dots, N\}$ to a system of ${N\choose 2}$ ODE's. The system is more complex when the number $L$ of loci is finite, because a focal locus $K$ ``interacts'' with all other loci, in the sense that the allele it carries in population $i$ can migrate to (and fix in) population $j$ depending on the fraction of shared alleles between $i$ and $j$ at the $L-1$ other loci. As $L$ gets large, this fraction converges to a deterministic value, which can be viewed both as the portion of a big genome that is shared between $i$ and $j$ and as the probability that any given locus, e.g., locus $K$ itself, carries the same allele in $i$ and in $j$. Then the stochastic process (noted ${\cal A}(t)$ in equation \eqref{eqn:def_cal_A}) that tracks the alleles carried by locus $K$ (or any given locus) is said to interact with its own law. In the mathematical literature, this process is called a McKean-Vlasov system \cite{dawson1995stochastic} and its one-dimensional marginal follows a Fokker-Planck equation which is consequently nonlinear (see Theorem \ref{thm:conv-L} in SI).

This enabled us to analytically study the stability of reproductive structures at equilibrium within an ODE framework. Our first observation was that \hyperlink{reversibility}{irreversibility} of speciation requires the condition $h'(0) = 0$.

In fact, if this condition is not verified, a genetic proximity of 0 would be an excitable state provided migration rates are high enough, that is, two reproductively isolated populations could resume gene flow after artificial introgression of a small number of alleles from one population into the other (see \cite{orr1995population, coyne2004speciation}).

Thus, speciation in our model arises when genetic distance and effective migration rates become trapped in a positive feedback loop, causing diverging populations to snowball into complete reproductive isolation. In fact, under the right conditions on $\mu$ and the migration matrix $M$, there may exist a stable nontrivial migration-mutation equilibrium. Provided $h'(0)=0$, there is inevitably a threshold which is an unstable  equilibrium, such that, if genetic proximities fall below this threshold, they converge to the stable equilibrium of complete reproductive isolation.

\paragraph{The issue of transitivity in species complexes.}
In the context of our model, a species, or more precisely a species complex, is a set of populations connected through direct or indirect gene exchange, i.e., in mathematical terms a connected component of the effective migration graph. Then, what can we say about the structure of a species complex? In particular, under which conditions can we expect the occurrence of intransitive species complexes like ring species?
To this end, we consider two classes of feedback functions: with and without a threshold, i.e., a value $c \in [0,1)$ such that two populations are completely reproductively isolated ($h = 0$) if their genetic proximity is smaller than $c$. In the absence of a threshold ($c = 0$), we showed that any two populations in the same species complex are able to exchange genes directly. However, when $c \in (0,1)$, intransitive equilibria like ring species can occur, i.e., equilibria such that populations connected through indirect gene flow can be reproductively isolated. Strikingly, intransitive equilibria even exist in complete and uniform migration, i.e. $M_{i j} = m > 0$ for all $i\neq j$, and we gave an example introduced as the friendship graph. We can thus interpret the issue of transitivity in species complexes in terms of the presence or absence of a threshold for the feedback function.

\paragraph{Subspecies clustering in species complexes.}

\noindent In the presence of hybridization between closely related species (e.g., grizzly-polar bear \cite{pongracz2017recent}), a nontrivial question relates to distinguishing whether occasional interbreeding between populations represents a transitory state on the way to complete reproductive isolation, or a stable state of reduced, but evolutionary persistent, gene exchange.

For small values of $N$, we showed that there exist species complexes with clusters of genetically similar populations, that experience moderate gene flow between clusters -- even in the uniform migration case. Furthermore, we observed that clustering within a species complex can emerge from a coherent unit by temporary isolation, due to multi-stability.

We emphasize that we showed existence of both intransitive and clustering 
equilibria when the number of populations $N$ is small. As it turns out, the transitivity and clustering properties of species complexes change completely when we consider large values of $N$.

\paragraph{Large metapopulations.}

\noindent Our results suggest that inhomogeneous equilibria in well-connected metapopulations disappear when the number of populations becomes large.
This seemingly counterintuitive behavior can be explained by the observation that the extent to which any specific migration rate influences the shape of the entire species complex decreases when the number of populations increases, provided there are many paths connecting any two populations.
Presumably, the property of increased homogeneity and transitivity in species complexes with migration rates that are uniform (or just bounded from below) will continue to hold even if the degree of each vertex in the graph grows only faster than logarithmically with $N$, instead of linearly with $N$ in the uniform case. We formulate this hypothesis in reference to the Erd\"os-Renyi random graph, obtained by randomly setting the migration rate of each directed edge to 0 independently with the same probability $1-a_N$ (and leaving it unchanged with probability $a_N$), and which is known (see \cite{hoory2006expander}) to satisfy the expander graph property 
provided $a_N\gg \log(N)/N$. It will be interesting to test this hypothesis and to check that on the contrary, the homogeneity property of large interbreeding structures does not persist when the average degree remains bounded.

\paragraph{Fluctuating migration networks.}

\noindent The connectivity of natural metapopulations varies through time due to numerous intrinsic or extrinsic, biotic or abiotic processes: local population size fluctuations, evolution of dispersal strategies, presence/absence of predators, competitors or pathogens in specific areas, resource depletion in specific areas, presence/absence of animal or material vectors, appearance/disappearance of geographic barriers...  Changes of connectivity in species complexes can destabilise them by homogenizing gene pools when connectivity increases (e.g., incipient speciation failing at secondary contact) or by locally increasing genetic differentiation when connectivity decreases (e.g., allopatric speciation).

Such destabilising dynamics should allow the species complex to explore the genetic landscape and visit several equilibria, in particular the complete cessation of gene flow between two clusters of populations (parapatric or peripatric speciation).
We investigated an extended version of our model, where migration rates are independently updated at rate $\theta > 0$. This framework is reminiscent of a ``species pump'' mechanism, which refers to a situation of repeated temporary spatial isolation and secondary contacts generating and propagating new species by series of local adaptations (after a fission event) and character displacements (after a fusion event), (see, e.g., \cite{terborgh1992diversity,aguilee2013adaptive,aguilee2011ecological}). With our tools, the system is much simpler to analyze than what was previously done thanks to a representation by a piecewise deterministic Markov process (PDMP), whereby migration rate resamplings are the stochastic jump events between which the dynamics are deterministic.

Our results can be summarized in three observations. First, the rate of speciation is higher in smaller metapopulations. Secondly, upon a speciation event, there is typically a single population detaching from the mother species.

Finally, we examined the relationship between the rate of environmental change and the rate of speciation, and found a non-monotonic relationship between the two: at first glance, one could think that speciation rates decrease for lower values of $\theta$ (because the environment becomes increasingly stable), and that speciation is more frequent when the rate of environmental change is large. However, our observations suggest that this no longer holds when the rate of change becomes too large. Heuristically, this is due to the fact that in order to initiate a speciation event, geographic restrictions must be maintained for some time, allowing the positive feedback loop between genetic distance and effective migration rate to kick in. If migration rates are updated too quickly, the geographic restrictions necessary for speciation disappear before significant divergence can occur. 
Studies of the variations of diversification rates (speciation minus extinction) in paleontological time show that periods of increased tectonic activity correspond to periods of less diversification (see, e.g., \cite{stadler2011mammalian}). Our results suggest that this decrease may be due to less frequent speciation events that result from geographic restrictions not being maintained for a sufficiently long time in order to initiate speciation events.

\paragraph{Open questions and future work.}

\noindent The numerical simulations of the stochastic model (see Fig. \ref{fig:conv_to_ODE}) revealed an intriguing behaviour of genetic proximities when the number of loci is small. In fact, speciation seems to result from stochastic fluctuations around the quasi-equilibrium of the genetic proximities. Thus, it would be interesting to study the deviations from the stochastic model, which could shed light on questions related to the average time to first speciation as a function of the number of loci considered.

An important question relates to the interplay of subspecies clustering (see section \ref{sec:subsp_clust}) and fluctuating migration: Starting from a subspecies clustering configuration, does fluctuating migration still result in new species that are singletons or may it result in species that are the initial subspecies clusters?

\noindent Taking the dynamics of the metapopulation through colonization and extinction events into account could be an interesting addition to our model. In the split-and-drift random graph model of speciation \cite{bienvenu2019split}, the mere presence of an edge indicates possible gene flow between vertices, and edges disappear spontaneously after some random time to model the accumulation of divergence. In a refined version of this model, only vertices can disappear, due to local extinctions, and due to recolonizations, new vertices can appear along with their edges: when a vertex is replicated, its daughter copies its neighbors and the genetic proximities with them. Between recolonizations, genetic proximities evolve explicitly through mutations and migrations as described in the present paper. The large $L$ version of this model should yield a piecewise deterministic Markov model with McKean-Vlasov dynamics that is mathematically interesting in its own right.

In the framework of the model, a question of interest concerns the expected time to speciation of large metapopulations. Specifically, in Fig. \ref{fig:time_spec}, the simulations suggest that the time to speciation increases very rapidly with the number of populations. It would be interesting to find an expression of the rate of speciation as a function of the number of populations, that yields coherent results with empirical speciation rate differences as a function of metapopulation size (see \cite{makarieva2004dependence}). Further, this raises the question of whether we can constrain the set of feedback functions that can be considered to model isolation regimes within a taxon by fitting the simulated increase in time to speciation associated to a feedback function to empirical data. 

More generally, considering the observed significance of the feedback function, it is crucial to gain further insight into how to propose a biologically meaningful function $h$ from first population genetic principles, and how to infer it from experimental data. In particular, studies of diverging populations that continue to exchange genes could provide insight into this issue (see for example \cite{roux2016shedding}). We believe that this would significantly advance our understanding of speciation rate variation.

\bibliography{References}
\bibliographystyle{ieeetr}

\onecolumn

\appendix
\addcontentsline{toc}{section}{Supplementary Information}
\section*{Supplementary Information}
The SI will be devoted to the rigorous derivation of the mathematical results exposed in the main text. In the first section, we will derive the master equation (\ref{eq:ODE}). In the second section, we study the equilibria of (\ref{eq:ODE}) and their stability. Throughout, we will use the notation $[N]:=\{1,\dots, N\}$.
\section{Deriving the master equation}
\label{chap:master_eq}
The derivation of the master equation (\ref{eq:ODE}) can be achieved in two steps. First, we show that the state and transitions of our model can be represented with partitions of the set of populations $[N]$ in a Markovian way. To this end, we will divide the populations at each locus into blocks -- depending on which other populations they share the same allele with. This interpretation allows us to show via a law of large numbers that the fraction of loci with some allelic partition converges to the solution to an ODE. Second, we show that this limiting ODE can be identified with the transition probabilities of a nonlinear Moran model. By leveraging the duality of the latter, we derive the master equation (\ref{eq:ODE}), summarized in the following theorem.

\begin{theorem}\label{thm:mast_eq}
    The process of stochastic genetic proximities $(P^L_{i j}(t))_{i,j\in[N]}; t\geq 0)$ converges in distribution as $L \ra \i$ to $(P_{i j}(t))_{i,j\in[N]}; t\geq 0)$, solution to the system of ordinary differential equations given by
    \begin{align*}
   \frac{dP_{i j}}{dt} = \sum_{k = 1}^N (M_{k i} h (P_{k i}) P_{k j} + M_{kj} h (P_{k j}) P_{k i}) - P_{i j} \left( \sum_{k = 1}^N (M_{k i} h (P_{k i}) + M_{k j} h
  (P_{k j}) ) + 2 \mu \right) \, .\label{eq:EqSystemPij} 
\end{align*}
\end{theorem}

\subsection{Convergence of the allelic partition process}
\label{chap:master_eq_part}

We define $[n] \assign \{ 1, \ldots, n \}$, and denote the set of partitions of $[n]$ with $\mathcal{P}_n$. We denote by $B_n$ the cardinal of $\mathcal{P}_n$ (Bell's number). To rigorously define our process, we need to introduce some notation. Let $K \in [L]$ be a given locus, and $i,j \in [N]$ two populations, with $L,N\in\mathbb{N}$. 

For any $t\ge 0$, we let $\Pi^K_t
\in \mathcal{P}_N$ be the \tmstrong{allelic partition} at locus $K$ at time $t$. More specifically, $\Pi^K_t$ is the partition induced by the equivalence relation $\sim_{\Pi^K_t}$ defined as
\[ i \sim_{\Pi^K_t} j \quad  \Leftrightarrow \quad \text{at time $t$, $i$ and $j$ carry the same allele at locus $K$}. \]
More generally, for any partition $\pi \in \mathcal{P}_N$, we will write $i \sim_{\pi} j$ if $i$ and $j$ belong to the same block of $\pi$.

This is a simple way to study genetic differences between
populations, because we actually do not have to keep record of any allele, or,
speaking in terms of Fig. \ref{fig:Model_trans}, we do not have to keep
using different colors to distinguish differences in genetic material.

We now define the process of {\tmstrong{allelic partitions}} 
\[  (\vec{\Pi}^{(L)}_t)_{t \geq 0} := (\vec{\Pi}_t)_{t \geq 0} = (\Pi^1_t, \ldots, \Pi^L_t)_{t \geq  0} \]
which is valued in ${\cal P}_N^{\otimes L}$. Finally, to compute the genetic proximity between two populations at time $t$ from the process $(\vec{\Pi}_t)_{t \geq 0}$, we define two functions. We set, for all $\pi \in \mathcal{P}_N, \vec{\sigma} \in \mathcal{P}^{\otimes L}_N$, and all populations $i, j \in [N]$,
  \begin{equation}
f_{\pi} (\vec\sigma) := \frac{1}{L} \sum_{K = 1}^L \mathbf{1}_{\{ \sigma_K =
    \pi \}}, \label{eq:MastEqGen}
  \end{equation}
  and
  \begin{align}
    P_{i j}^L (\vec\sigma) := \frac{1}{L}\sum_{K=1}^L \mathbf{1}_{\{i \sim_{\sigma_K}j \}}. \label{eq:DefProxGen} 
  \end{align}
Intuitively, $f_\pi(\vec{\Pi}_t)$ will correspond to the fraction of loci with allelic partition given by $\pi$, while $P_{ij}^L(t) \assign P_{ij}^L(\vec{\Pi}_t)$ will correspond to the \tmstrong{genetic proximity} between populations $i$ and $j$ at time $t$. Note that this definition is just the mathematical translation of the above idea of counting the number of different alleles.

The process $(\vec{\Pi}_t)_{t>0}$, and thus the process of genetic proximities $\{(P_{i j}^L(t))_{t\geq 0} :  i,j\in[N]\}$, will be governed by two antagonistic forces:
\begin{enumerate}
  \item {\tmstrong{Mutation events}}: mutations occur within each population $i$ and at each
  locus $K$ at a constant rate $\mu$. Given such a mutation event, the allelic partition $\Pi^K_t$ changes to $s_i (\Pi^K_t)$, the partition
  created from $\Pi^K_t$ by isolating the singleton $i$ into a block of its own.
  
  \item {\tmstrong{Migration events}}: between each pair of populations $i$ and
  $j$, at each locus $K$, migration events occur at an effective rate
  \begin{align} \label{eq:effMigRatesAp}
      M_{i j}^e = M_{i j} \cdot h (P_{i j}^L (t) ),
  \end{align}
  We refer to the model description (see Section \ref{chap:model}) for the definitions of $M_{i j}$ and $h$. Given a migration event from $i$ to $j$ at locus $K$, the allelic partition $\Pi^K_t$ changes to $\sigma_{j \rightarrow i} (\Pi^K_t)$, the
  partition created from $\Pi^K$ by putting the element $j$ in the block
  containing $i$. Heuristically, when $i$ migrates to $j$, the element $j$
  will take the type of $i$, which corresponds to placing $j$ into the block
  containing $i$.
 \end{enumerate}

To expose the main result of the section, we start with some notation. Set
\[ \mathcal{M}_1({\cal P}_N):= \left\{ \vec \rho = (\rho_{\pi})_{\pi\in{\cal P}_N}, \ \  \sum_{\pi \in {\cal P}_N}
   \rho_{\pi} = 1 \right\} \]
the set of probability measures on $\mathcal{P}_N$.
For every $\vec{\rho}\in{\cal M}({\cal P}_N)$, we set $\vec{\rho} (i\sim j) := \sum_{\pi:i\sim j} \rho_\pi$. Finally,
$A(\vec \rho)$ is the transition rate matrix such that 
for $\pi \neq \pi'$
\begin{itemizedot}
  \item $A_{\pi, \pi'}(\vec{\rho}) = \mu$, if $\pi' = s_i (\pi)$ for some $i \in [N]$  
  \item $A_{\pi, \pi'}(\vec{\rho}) = M_{i j} h \bigg( \vec{\rho} (i\sim j) \bigg)$, if $\pi' = \sigma_{j \rightarrow i} (\pi)$ for some
  $i, j \in [N]$.
\end{itemizedot}
Define 
\begin{equation}
\forall \pi\in{\cal P}_N, \ \   X_{\pi}^L (t) : = f_{\pi} (\vec{\Pi}_t) = \frac{1}{L} \sum_{K = 1}^L \mathbf{1}_{\{ \Pi_t^K = \pi \}}  \label{DefXL}
\end{equation}
and  $\vec{X}^L = (X^L_\pi)_{\pi\in B_N}$ the process in 
 $\mathbb{D}$, the space of c{\`a}dl{\`a}g functions valued in ${\cal M}_1({\cal P}_N)$ endowed with
the Skorohod (J1)-topology \cite{billingsley2013convergence,skorokhod1956limit}.

\begin{theorem}\label{thm:conv-L}
The sequence $(\vec{X}^L)_L$ converges in law 
to $(\vec{X}(t))_{t\geq 0} = ((X_{\pi}(t))_{\pi \in
  \mathcal{P}_N})_{t\geq 0}$, the
 unique solution of the deterministic ODE
  \begin{eqnarray}
  \label{ast}
    \frac{d \vec X(t)}{d t} =    \vec{X}(t) A(\vec{X}(t)) =: \vec{G}(\vec{X}(t))
  \end{eqnarray}
\end{theorem}

% \Section{Deriving the master equation}

We decompose the proof into several elementary lemmas. The first lemma is obtained by straightforward computations, and thus we omit its proof.
\begin{lemma}
Let $Q^L$ be the generator of the partition process $(\vec{\Pi}^{(L)}_t)_{t\geq 0}$. Then 
  \label{PropGen}
  \begin{eqnarray}
      Q^{L} {f} (\nu)  = {f}(\nu) A ({f}(\nu)) \, ,
  \end{eqnarray}
  for all $f:\mathcal{P}^{\otimes L}_N \rightarrow \mathbb{R}, \nu \in \mathcal{P}^{\otimes L}_N$.
\end{lemma}

\begin{lemma}   \label{ConvQuadVar}
Define
\begin{eqnarray*}
  M_\pi^L (t) & := & f_{\pi} (\vec{\Pi}_t) - f_{\pi} (\vec{\Pi}_0) - \int_0^t Q^L f_{\pi}
  (\vec{\Pi}_u) d u. \nonumber \\
  & = & X^L_\pi(t) - X^L_\pi(0) - \int_0^t \left( \vec{G}(\vec{X}^L(t)) \right)_{\pi}
 d u. \label{GenMtg}
\end{eqnarray*}
Then, the quadratic variation of the martingale $M_\pi$ verifies
  \[ \langle M^{L}_\pi \rangle_t = O \left( \frac{1}{L} \right) \ \  \ \ \mbox{as $L \rightarrow
     \infty$.} \]
\end{lemma}

\begin{proof}
For any $\rho \in \mathcal{P}^{\otimes L}_N$, we denote by $\rho_{K,
  \pi'}$ the partition vector obtained from $\rho$ by changing the $K$-th
  coordinate of $\rho$ to the partition $\pi'$. Additionally, we denote by
  $\tau (\rho, \rho')$, for any $\rho, \rho' \in \mathcal{P}^{\otimes L}_N$, the
  rate of change from $\rho$ to $\rho'$.
 The quadratic variation of $M^L$ is given by
  \[ \langle M^L_\pi \rangle_t =  \int_0^t \sum_{K = 1}^L \sum_{\pi' \neq (\Pi _u)_K} (f_{\pi} ((\Pi_u)_{K,\pi'}) - f_{\pi} (\Pi_u))^2 \cdot \tau ((\Pi_u), (\Pi_u)_{K,
     \pi'}) d u . \]
    On the one hand,
     $(f_{\pi} ((\Pi_u)_{K,\pi'}) - f_{\pi} (\Pi_u))^2 \leq \frac{1}{L^2}$. 
On the other hand, the rates can be
uniformly bounded in $L$ by $0 < \tau_{\max} < \infty$. This yields
  \begin{equation}
    \langle M^L_\pi \rangle_t \leqslant \frac{t B_N \tau_{\max}}{L} ,
    \label{MajGen}
  \end{equation}
  which ends the proof.
  \end{proof}
%
%\subsection{Convergence to the master equation}
%
%Proposition \ref{PropGen} makes apparent that the proportions of loci in a
%given partition $\pi$ depend on the proportions of neighbouring $\pi' \in
%\mathcal{P}_N$. Hence, we will want to obtain a $B_N$-dimensional system of
%differential equations at the limit. Let $\pi \in \mathcal{P}_N$. 

\begin{lemma}
  \label{CorConvSub} The sequence $\{ (\vec{X}^L(t))_{t\geq0}
  \}_{L \in \mathbb{N}}$ is tight in $\mathbb{D}$.
\end{lemma}

\begin{proof}
  We will use the Aldous-Rebolledo criterion for tightness, see
  {\cite{aldous1978stopping,rebolledo1980central}}. To prove
  tightness  of $\vec{X}^L$, it suffices
  to prove tightness of each coordinate.
  
  Denote $\mathbb{F}_{L}$ the natural filtration of the $\mathbb{D}$ valued process $\vec{X}^L$. Let $S, S'$ two stopping times w.r.t. $\mathbb{F}_{L}$ such that a.s. $0 \leqslant S \leqslant S'
  \leqslant S + \delta \leqslant T$ for $T \in \mathbb{R}_+$ and $\delta > 0$.
  Let $\pi \in \mathcal{P}_N$. Remark that
  \[ X^L_{\pi} (S') - X^L_{\pi} (S) = M^L_{\pi} (S') - M^L_{\pi} (S) +
     \int_S^{S'} Q^Lf_{\pi}
  (\vec{\Pi}_u) d u . \]
  We have to prove that the laws of the martingale part and of the finite
  variation part are tight. Using that $M^L_\pi$ is a martingale and the monotonicity of the quadratic variation, we get  
  \begin{eqnarray*}
    \mathbb{E} \big( | M^L_{\pi} (S') - M^L_{\pi}(S) |^2  \big) & \leqslant &
    \mathbb{E} \big( M^L_{\pi}(S')^2 - M^L_{\pi}(S)^2 \big) \\
    & \leqslant & \mathbb{E} (\langle M^L_{\pi}\rangle_{S + \delta} -
    \langle M^L_{\pi} \rangle_S)\\
     & \leqslant & \frac{B_N \tau_{\max} }{L}
    \mathbb{E} \left( \int_S^{S + \delta} d u \right) = \frac{B_N
    \tau_{\max}\delta}L ,
  \end{eqnarray*}
  where in the last inequality we have used a similar reasoning as the one yielding \eqref{MajGen}, and which allows us to deduce tightness of the martingale part.
  It remains to prove tightness of the finite variation part. This can be seen
  directly by the same argument and the uniform boundedness in $L$ of the
  generator. 
\end{proof}

\begin{proof}[Proof of Theorem \ref{thm:conv-L}]
  Because $h$ and $\vec{G}$ are $\mathcal{C}^1$ functions, there exists a unique
  solution to (\ref{ast}) by standard
  Cauchy-Lipschitz arguments. By Lemma \ref{CorConvSub} and an application of Prohorov's Theorem, there exists a subsequence of the sequence $(\vec X^L)$, that we will still denote $(\vec X^L)$ for simplicity, which converges weakly, and even a.s. by Skorohod's Theorem to some $\vec X^\infty\in \mathbb{D}$. Let us show that $\vec{X}^\infty$ is solution to (\ref{ast}).
  
Let $t>0$, and recall that 
  \[ M^L_\pi(t)  = X^L_{\pi} (t) - X^L_{\pi} (0) - \int_0^t  \left( \vec{G}(\vec{X}^L(u))\right)_{\pi} d u. \]
  On the one hand,
  Lemma \ref{ConvQuadVar} yields
  \[ \  \mathbb{E} [(M^L_\pi(t))^2] \rightarrow 0, \ \  \ \  \mbox{as $L \rightarrow \infty$} , \]
On the other hand, by continuity of $G$ and dominated convergence, $M^L_\pi(t)$
converges a.s. to 
\begin{align}
M^\infty_\pi(t):=X^\infty_{\pi} (t) - X^\infty_{\pi} (0) - \int_0^t \left( \vec{G}(\vec{X}^\infty(u))\right)_{\pi} d u .\label{masterStoInf}
\end{align}
But now $M^L_\pi(t)$ converges to 0 in $L^2$ so by uniqueness of the limit $M^\infty_\pi(t)=0$, which ends the proof of Theorem \ref{thm:conv-L}.

\end{proof}

\subsection{Duality}

Since the dimension of the ODE (\ref{ast}) is the number of
partitions of the set $[N] = \{ 1, \ldots, N \}$, the
ODE system quickly becomes intractable. Thus, we will prove a duality relation
allowing us to reduce the dimension of the system of interest to 
$N (N - 1) / 2$, the number of pairs of $[N]$.

The main idea relies on a stochastic interpretation 
of (\ref{ast}). To gain some intuition, we first recall the definition of the Moran model with mutation on a directed weighted graph.
Consider a population of individuals $1,...,N$
and a dynamical matrix $(M(t))_{t\geq 0} = (M_{ij}(t))_{t\geq0, i\neq j\in[N]}$ with non-negative entries.

The system evolves according to the following dynamics.
\begin{itemizedot}
  \item Each individual takes on a new type at rate $\mu$ (infinite-allele assumption).
  \item For $i\neq j$ and at time $t$, individual $j$ takes on the type of
  individual $i$ at rate $M_{i j}(t)$.
\end{itemizedot}
As before, we can conveniently encode the dynamics by recording the allelic partition along time.
This defines a time-inhomogeneous Markov process valued in ${\cal P}_N$.

\bigskip

Let us now introduce the nonlinear Markov process version of the latter Moran model. Informally, this amounts to assuming 
that the dynamical migration matrix $M(t)$ depends on the law of the process itself; namely, we consider the partition process $(\sigma_t; t\geq0)$ on ${\cal P}_N$ induced by a time-inhomogeneous Moran model  with dynamical matrix
\begin{equation}
 \label{DefGenProx}
\forall i\neq j\in[N], \ \ M_{ij}(t) \ = \ M_{ij} h( \hat{P}_{i j}(t) )  \ \ \mbox{where $\hat{P}_{i j}(t) \ := \ \mathbb{P}( i \sim_{\sigma_t} j ) . $}
\end{equation}
Following the terminology of \cite{neumann2023nonlinear}, $(\sigma_t)_{t\geq0}$ defines a finite-state time-inhomogeneous Markov chain 
whose semi-group is determined by the solution to a non-linear differential forward Kolmogorov equation. 

More formally, let $s > 0$. It is clear by the definition of the dynamical migration matrix $M(s)$ that at time $s$, the transition rate matrix of the partition-valued process $(\sigma_t; t\geq0)$ is given by $A(P_s)$, where 
$$
P_s = ( \mathbb{P}( \sigma_s = \pi))_{\pi\in\mathcal{P}_N} \, ,
$$
and $A$ is given in section \ref{chap:master_eq_part}.
We note that the application $P \mapsto A (P)$ is a Lipschitz continuous and
bounded function.
By Theorem 2.1 in {\cite{neumann2023nonlinear}}, there exists a unique (time-inhomogeneous) Markov process $(\sigma_t)_{t\geq0}$ valued in ${\cal P}_N$,
whose semi-group $(S(t))_{t\geq 0}$
is characterized  by the non-linear forward Kolmogorov equation
\begin{equation}
 \label{KolmForw}
\frac{d S(t)}{dt} \ = \  S(t) A( S(t) ).
\end{equation}
In particular, we recover the limiting ODE (\ref{ast}) for each coordinate of the matrix equation, i.e., for the functions 
$$
t \mapsto \mathbb{P}_{\pi,\pi'}(t) \ = \ \mathbb{P}_{\pi}( \sigma_t = \pi' ) .
$$
This justifies the interpretation of $\hat{P}_{ij}(t)$ (see (\ref{DefGenProx})) as the genetic proximity introduced in Section \ref{chap:model} between populations $i$ and $j$, in the large $L$ regime.

As in the standard Moran model \cite{etheridge2011some}, we consider
the following graphical representation on $[N]\times \mathbb{R}_+$:
\begin{itemizedot}
  \item For a reproductive event $i\rightarrow j$ at time $t$, draw an arrow with tail at $(i,t)$ and tip at $(j,t)$ 
  \item For a mutation event at site $k$ at time $t$,
 draw a $\star$ at $(k,t)$. 
 \end{itemizedot}

Via the graphical representation we discussed in section \ref{chap:ODE_Dual}, (see Fig. \ref{fig:MoranDual}), we can associate to every individual an
ancestral lineage using the arrow-star configuration. For every point $(i,t)$ with $i \in [N]$, we define $S_{(i,t)}$ to be the  ancestral lineage starting from $(i,t)$. The system of ancestral lineages $(S^{(1,t)}(s), \ldots, S^{(N,t)}(s); s\leq t)$ starting from time horizon $t>0$ 
evolves according to the following dynamics.
\begin{itemize}
\item Lineages are running backward in time and evolve independently until they coalesce.
\item A lineage jumps from $j$ to $i$
at time $s$
at rate $M_{i j} h (P_{i j} (t - s)) $. 
\item A lineage is killed (or stopped) at rate $\mu$.
\end{itemize}

We can recover the allelic partition from these ancestral lineages by
remarking that two individuals $i,j$ are in the same block at time $t$ iff the ancestral
lineages $S_{(i,t)}$ and $S_{(j,t)}$ trace back to the same type. In turn, 
the lineages
trace back to the same type if one of two events happen:
(1) The two lineages $S_{(i,t)},S_{(j,t)}$ coalesce before time $t$; or (2) the two lineages survive up to time $t$,
they do not coalesce, but
they hit two sites in the same partition, i.e., if there are $i_0, j_0$ such that $S_{(i,t)}(t)=(i_0,0)$ and $S_{(j,t)}(t)=(j_0,0)$ for some $i_0\neq j_0\in[N]$, such that $i_0 \sim_{\sigma_0} j_0$. This leads to the following proposition.

\begin{proposition} \label{prop:DynDual}  
Consider the unkilled ancestral lineages $\bar S_{(i,t)}$ and $\bar S_{(j,t)}$, i.e., the ancestral lineages 
starting from $(i,t)$ and $(j,t)$ and ignoring the killing event $\star$. (Equivalently, this amounts to setting 
$\mu=0$). Define the coalescing time
\[ T_{(i,j),t} := \sup\{ u > 0 : \bar S_{(i,t)}(u) = \bar S_{(j,t)} (u) \} . \]
If 
$$
  P_{i j}(t) \ = \ \mathbb{P}( i \sim_{\sigma_t} j ) ,
  $$
  then 
  \begin{equation}
    P_{i j} (t) =\mathbb{E} \left (e^{- 2 \mu T_{(i,j),t}};  T_{(i,j),t} \leqslant t \right) + 
    \mathbb{E} \left(e^{- 2 \mu T_{(i,j),t}} ; \bar S_{(i,t)}( (t) \sim_{\sigma_0} \bar S_{(j,t)} (t), T_{(i,j),t} > t \right) \, .  \label{DualRel}
  \end{equation}
\end{proposition}

With the help of Proposition \ref{prop:DynDual}, we can establish an ODE system for the genetic proximities.
\begin{corollary}
  \label{cor:ODE}The
  genetic proximities $P_{i j}$ solve the following system of ordinary
  differential equations,
  \begin{eqnarray}
    \frac{d P_{i j}}{d t} (t) & = & \sum_{k = 1}^N (M_{k i} h (P_{k i} (t))
    P_{k j} (t) + M_{k j} h (P_{k j} (t)) P_{k i} (t)) \label{eq:LinSysGenProx}
    \\
    & - & P_{i j} (t) \left( \sum_{k = 1}^N (M_{k i} h (P_{k i} (t)) + M_{kj} h (P_{k j} (t)) ) + 2 \mu \right), \nonumber
  \end{eqnarray}
  where we recall that $M_{k k} = 0$ and  $P_{k k} (t) = 1$ for all $k$. 
\end{corollary}

\begin{proof} 
 We will only show the result where the initial condition is given by the
  singleton partition. The general case can be proved along the same lines.
  
  We will use Proposition \ref{prop:DynDual}, and decompose the expectation according to the possible jumps of the unkilled random walks $S_i \assign \bar S_{(i,t)}$ and $S_j \assign \bar S_{(j,t)}$ in a small
  interval of time of length $d t > 0$, for $i\neq j$. Denote the number of jumps
  of the process $(S_i, S_j)$ in the time interval $[0, s)$, where $ 0 < s
\leq t  $, as $\mathcal{N} ([0, s)) = (\mathcal{N}_i ([0, s)),
  \mathcal{N}_j ([0, s)))$. 
  Then define 
  \begin{eqnarray*}
    A_{0}(s) &:=&  \{\mathcal{N} ([0, s))= (0, 0)\}\\
 A_{1,i}(s) &:=&  \{\mathcal{N} ([0, s))= (1, 0)\}\\
 A_{1,j}(s) &:=&  \{\mathcal{N} ([0, s))= (0, 1)\}\\
 A_{2}(s) &:=&  \{\mathcal{N}_i ([0, s)\ge 1, \mathcal{N}_j ([0, s)\ge 1\}.
      \end{eqnarray*}
  Denoting $Y_{ij}(t):=e^{- 2 \mu T_{(i,j),t}} \tmstrong{1}_{\{
    T_{(i,j),t} \leqslant t \}}$, we get
$$
    P_{i j} (t) =\mathbb{E} (Y_{ij}(t))=\Delta_0 (d t) + \Delta_{1, i} (d t) + \Delta_{1, j} (d t) +
    \Delta_2 (d t),
$$
where
\begin{eqnarray*}
    \Delta_{0}(dt) &:=& \mathbb{E} (Y_{ij}(t)\tmstrong{1}_{A_0(dt)})\\
 \Delta_{1,i}(dt) &:=& \mathbb{E} (Y_{ij}(t)\tmstrong{1}_{A_{1,i}(dt)})\\
  \Delta_{1,j}(dt) &:=& \mathbb{E} (Y_{ij}(t)\tmstrong{1}_{A_{1,j}(dt)})\\
  \Delta_{2}(dt) &:=& \mathbb{E} (Y_{ij}(t)\tmstrong{1}_{A_{2}(dt)}).
      \end{eqnarray*}
  
  The last quantity will be of order $d t^2$, and hence vanish in
  the limit when we divide by $d t$.

  {\tmstrong{{\tmstrong{Case}} 1}}: $S_i$ jumps once.
  
  {\tmstrong{Case 1.1}}: $S_i$ jumps to $k \neq j$. Then, define
  \[ A_{i \rightarrow k, d t} = \{ d t < T_{(i,j),t}\leqslant t \} \cap \{
     \mathcal{N} ([0, d t)) = (1, 0), S_i (d t) = k \}, \]
  On this event, we have
  \[ T_{(i,j),t} = T_{(k,j),t - d t} + d t . \]
  The probability that the random walk starting from $i$ jumps exactly once on
  the interval $[0, d t)$ to location $k$ is given by
  \[ M_{ki} h (P_{ki} (t)) \cdot d t + o (d t), \]
  which follows from the continuity of the function $h \circ P_{ki}$.
  
  {\tmstrong{Case 1.2}}: $S_i$ jumps to $j$. We consider the event where
  coalescence happens on the time interval $[0, d t]$. The corresponding probability is given by
  \[ M_{j i} h (P_{j i} (t)) \cdot d t + o (d t), \]
  and the coalescence time $T_{(i,j),t}$ equals the jump time.
  
  Putting cases 1.1 and 1.2 together, we obtain
  \begin{eqnarray*}
    \Delta_{1, i} (dt)& = & d t \sum_{k \neq j} M_{k i} h (P_{k i} (t)) \mathbb{E}
    (e^{- 2 \mu (T_{(k,j),t - d t} + d t)} \tmstrong{1}_{\{ T_{(k,j),t - d t}
    \leqslant t - d t \}})\\
    & + & d t \cdot M_{j i} h (P_{j i} (t)) + o (d t) .
  \end{eqnarray*}
  Since the probability to see an event in a time interval of length $dt$ converges to zero, we get
  \[ \mathbb{E}
    (e^{- 2 \mu (T_{(k,j),t - d t} + d t)} \tmstrong{1}_{\{ T_{(k,j),t - d t}
    \leqslant t - d t \}}) = \mathbb{E} (e^{- 2 \mu T_{(k,j),t}}
     \tmstrong{1}_{\{ T_{(k,j),t} \leqslant t \}}) + o (1) = P_{k j} (t) + o (1), \]
  and thus
  \begin{eqnarray*}
    \frac{\Delta_{1, i} (d t)}{d t} & \xrightarrow{d t \rightarrow 0} & 
    \sum_{k \neq j} M_{k i} h (P_{k i} (t)) P_{k j} (t) + M_{j i} h (P_{j i}
    (t)) .
  \end{eqnarray*}
  
  {\tmstrong{Case 2}}: $S_j$ jumps once. The same arguments as in case 1 can
  be applied.
  
  {\tmstrong{Case 3}}: Neither $S_i$, nor $S_j$ jumps. We remark that
  conditionally on the event $A_0(dt) = \{ \mathcal{N} ([0, d t)) = (0, 0) \}$,
  the coalescence time is given by
  \[ T_{(i,j),t} = T_{(i,j),t - d t} + d t . \]
  Hence,
  \begin{align*}
    \Delta_0(dt) = \left( 1 - d t \cdot \left( \sum_{k = 1}^N M_{k i} h (P_{k i} (t)) + M_{k j} h (P_{k j} (t)) \right) + o (d t) \right) \cdot P_{i j}
    (t - d t) e^{- 2 \mu d t} .
  \end{align*}
  Finally, we obtain
  \begin{eqnarray*}
    \lim_{d t \downarrow 0} \frac{P_{i j} (t + d t) - P_{i j} (t)}{d t} & = &
    \sum_{k = 1}^N (M_{k i} h (P_{k i} (t)) P_{k j} (t) + M_{k j} h (P_{k j}
    (t)) P_{k i} (t))\\
    & - & \left( \sum_{k = 1}^N (M_{k i} h (P_{k i} (t)) + M_{k j} h (P_{k j} (t))) + 2 \mu \right) \cdot P_{i j} (t),
  \end{eqnarray*}
  which yields the desired result.
\end{proof}

This concludes the proof of Theorem \ref{thm:mast_eq}.

\section{Equilibria and stability} \label{sec:eq_stab}

This section is devoted to the study of the equilibria of the master equation and their stability. Consider a solution $P = (P_{i j})_{i \neq j}$ to the master equation such that $\vec{F}(P) = 0$. To determine stability, we study the Jacobian of $\vec F$ given by
\begin{align*}
  \frac{\partial \vec{F} (P)_{i j}}{\partial P_{i j}} = (M_{i j} + M_{j i} )
  h' (P_{i j}) (1 - P_{i j}) - \left( \sum_{k = 1}^N (M_{k i} h (P_{k i}) + M_{k j} h (P_{k j})
  ) + 2 \mu \right) 
\end{align*}
on the diagonal, and for $k \neq i, j$,
\begin{align*}
  \frac{\partial \vec{F} (P)_{i j}}{\partial P_{k i}} = M_{k i} h' (P_{k i}) P_{k j} + M_{k j} h (P_{k j}) - P_{i j} M_{k i} h' (P_{k i}) \, .
\end{align*}

In the previous section, we derived the master equation via a duality approach which relied on analysis of coalescing random walks with inhomogeneous jump rates depending on $P_{ij}(t-s)$. Since we are now studying the system at equilibrium, we can interpret the jump rates as the edge weights of a static graph, which we will call the effective migration graph. 
\begin{definition}[Dual effective migration graph]\label{def:results}
The dual effective migration graph $M^\T{eq}$ associated to an equilibrium $P^\T{eq}$ is the graph with vertices $[N]$ and directed edge weights  given from $i$ to $j$ by $M^{eq}_{ij} = M_{ji}h(P^\T{eq}_{ij})$. 
\end{definition}
With the help of the dual effective migration graph, the genetic proximities at equilibrium can be expressed by a fixed point problem.
\begin{theorem}[Fixed point problem]
  \label{thm:FixPtPb} 
 Let $P^\text{eq} = (P^\text{eq}_{i j})_{i \neq
  j}$ be an equilibrium for the system of genetic proximities (\ref{eq:LinSysGenProx}). 
  Consider the unkilled ancestral lineages $\bar S_i$ resp. $\bar S_j$ starting from $i$ resp. $j$ on the dual effective migration graph $M^\T{eq}$, i.e., with jump rates given by its weighted edges. Define the coalescing time
\[ T_{i  j} := \inf\{ u > 0 : \bar S_i(u) = \bar S_j (u) \} . \] Then, $P^\text{eq}$ satisfies the fixed point problem
    \begin{equation}
P^\text{eq}_{ij} =\mathbb{E} \left( e^{- 2 \mu T_{i j}(P^\text{eq})} \right)
      \label{app:FxPtPb}
    \end{equation}
\end{theorem}
\begin{proof}
  The proof easily follows from (\ref{DualRel}) by letting $t\to\infty$.
\end{proof}
\begin{remark}\label{rem:transitivity}
Each pair of populations  belonging to the same species has nonzero proximity. Indeed, for any populations $i$ and $j$ in the same species, there is at least one path of intermediary populations in the dual effective migration graph connecting $i$ and $j$, so that $T_{ij}(P^\text{eq})$ is  finite with positive probability. Then, equation \eqref{app:FxPtPb} ensures that $P^\text{eq}_{ij}\not=0$.\end{remark}

\begin{remark}\label{rem:glob_loc}
    The concept of the dual effective migration graph can significantly simplify stability considerations. Indeed, under the assumption $h'(0) = 0$ , the stability of an equilibrium is equivalent to the stability of each connected component in the associated dual effective migration graph (see Proposition \ref{prop:glob_loc}). This allows us to rule out the fusion of well separated species upon secondary contact in the stability analysis. In other words, this assumption entails that speciation is irreversible in any ensemble of species complexes.
\end{remark} 

\begin{proposition}\label{prop:glob_loc}
   Assume that $h$ verifies $h'(0) = 0$. Let $P^\text{eq} = (P^\text{eq}_{i j})_{i \neq j}$ an equilibrium for the system of genetic proximities (\ref{eq:LinSysGenProx}). Then, the stability of $P^\text{eq}$ is equivalent to the stability of $P^\text{eq}$ restricted to any  connected component of the dual effective migration graph. More precisely, $P^\text{eq}$ is (locally) stable iff for every connected component $S$ of $M^\T{eq}$, the modified equilibrium $P^{\text{eq},S}$ given by
    $$
    P^{\text{eq},S}_{i j} = \tmstrong{1}_{\{i\in S, j\in S\}} \cdot P^\text{eq}_{i j} \, ,
    $$ is such that for every eigenvalue $\lambda$ of the Jacobian $J (P^{\text{eq},S})$, we have $\mbox{Re}(\lambda) < 0$.
\end{proposition}
\begin{proof} Let $S_1, ..., S_n$ the connected components of $M^\T{eq}$. By abuse of notation we will define for any connected component $S$ the relation $\sim_S$ by $i\sim_S j$ whenever $P^{\text{eq},S}_{ij}\not=0$. Recall from Remark \ref{rem:transitivity} that $\sim_S$ is transitive.
We define the vector subspaces forming a direct sum
$$
E_\sim \assign \{  \vec{y} = (y_{i j})_{ i \neq j}: y_{k l} = 0 \text{ if }\exists p\in[n] \text{ such that } k \sim_{S_p} l\}
$$
and
$$
E_{\not\sim,p} \assign \{  \vec{y} = (y_{i j})_{ i \neq j}: y_{k l} = 0 \text{ if } k \not\sim_{S_p} l \} .
$$
for all $p\in[n]$.

We want to show that $J \assign J_{\vec{F}} (P^\text{eq})$ is stable iff $J$ restricted to $E_{\not\sim, p}$ is stable for all $p\in[n]$. Let us first show that $J$ verifies $J(E_{\sim}) \subset E_{\sim}$ and $J(E_{\not\sim,S}) \subset E_{\not\sim,S}$, for every connected component $S$, which yields the decomposition of the eigenvalues of $J$ in terms of the eigenvalues restricted to $E_{\sim}$ and the $E_{\not\sim,p}$.
Let $\vec{y}\in E_{\sim}$, and $i,j$ such that $i\sim_S j$. We have 
$$
(J \cdot y)_{i j} = \sum_{(k,l)} J_{(i j),(k l)} y_{k l} = \sum_{(k,l): k\not\sim_S l} J_{(i j),(k l)} y_{k l} ,
$$
where we used the definition of $E_{\sim}$. 
In all cases when  $\{i,j\}\cap \{k,l\}=\varnothing$,  $\frac{\partial \vec{F} (P^\text{eq})_{i j}}{\partial P_{k l}}=0$.
Hence, let us compute $\frac{\partial \vec{F} (P^\text{eq})_{i j}}{\partial P_{k i}}$, for $k$ and $i$ such that $i \not\sim_S k$. Since $\sim_S$ is transitive, we have $j \not\sim_S k$.
Thus,
$$
\frac{\partial \vec{F} (P^\text{eq})_{i j}}{\partial P_{i k}} = M_{k i} h'(P^\text{eq}_{i k}) (P^\text{eq}_{k j} - P^\text{eq}_{i j}) + M_{k j} h(P^\text{eq}_{j k}) = 0 ,
$$
since $h$ verifies $h'(0) = 0$. 
Thus $J(E_{\sim})\subset E_{\sim}$.

Let now $p\in[n],\vec{y}\in E_{\not\sim,p}$, and $i,j$ such that $i \not\sim_{S_p} j$. We have 
$$
(J \cdot y)_{i j} = \sum_{(k,l)} J_{(i j),(k l)} y_{k l} = \sum_{(k,l): k\sim_{S_p} l} J_{(i j),(k l)} y_{k l} .
$$
Let $k$ such that $i\sim_{S_p} k$. Transitivity of $\sim_{S_p}$ yields $j \not\sim_{S_p} k$, and thus $P^\text{eq}_{i j} = P^\text{eq}_{j k} = 0$. Thus $J_{(ij),(ik)} = 0$, and therefore $J(E_{\not\sim,p})\subset E_{\not\sim,p}$. This implies
$$
\text{Sp}(J) = \text{Sp}(J|_{E_\sim}) \cup \bigcup_{p=1}^n \text{Sp}(J|_{E_{\not\sim,p}}) \, ,
$$
because for every tuple $(kl)$, the definition of $v \in E_\sim$ and $u_1 \in E_{\not\sim,1}, \dots, u_n \in E_{\not\sim,n}$, as well as the above assure that there is exactly one vector in $Jv,Ju_1,\dots,Ju_n$ that has a non-zero entry at $(kl)$.
It remains to show that for all $\lambda\in \text{Sp}(J|_{E_{\sim}})$, $\text{Re}(\lambda)<0$. 

The natural basis of $E_{\sim}$ is indexed by the set $K$ of unordered pairs $(ij)$ such that $i\not\sim_Sj$ (for all connected components $S$) and the  representative matrix of $J|_{E_{\sim}}$ in this basis is given for all $(ij)\in K$ by
$$
(J|_{E_{\sim}})_{(i j),(i k)} = %\textbf{1}_{\{i\not\sim_S j, i \not\sim_S k, j \sim_S k\}} 
\textbf{1}_{\{\exists p\in[n]:j \sim_{S_p} k\}}
M_{k j} h(P^\text{eq}_{k j}) 
$$
for $k\not= j$, and for $k=j$ (diagonal terms)
$$
(J|_{E_{\sim}})_{(i j),(i j)}=- 2 \mu - \left( \sum_{\exists p:l\sim_{S_p} i} M_{l i} h(P^\text{eq}_{l i})\right) - \left( \sum_{\exists p:l\sim_{S_p} j} M_{l j} h(P^\text{eq}_{l j})\right).
$$
 From here, it is easy to see that we may write 
$$
J|_{E_{\sim}} = (-2\mu)\cdot \textbf{I} + U ,
$$
where \textbf{I} is the identity matrix indexed by $K$ and $U$ is the transition rate matrix of the Markov chain on $K$ that jumps from $(ij)\in K$ to $(ik)\in K$ at rate $\textbf{1}_{\{\exists p\in[n]:j \sim_{S_p} k\}} 
M_{k j} h(P^\text{eq}_{k j})$. It is known (see, for instance, \cite{stroock2013introduction}), that the eigenvalue of $U$ with largest real part is given by 0, thus the stability of $J|_{E_\sim}$. This allows us to conclude.
% \al{AL: it seems to me that we have only proved that if $P^\text{eq}$ is stable then for every connected component $S$, $P^{\text{eq},S}$ is stable. Then I wonder if we don't have to consider simultaneously all connected components $S_1,\dots, S_m$, 
% and define the vector subspaces
% $$
% E_0 \assign \{  \vec{y} = (y_{i j})_{ i \neq j}: y_{k l} = 0 \text{ if } \exists \ell, k \sim_{S_\ell} l \} ,
% $$
% and
% $$
% E_\ell \assign \{  \vec{y} = (y_{i j})_{ i \neq j}: y_{k l} = 0 \text{ if } k \not\sim_{S_\ell} l\}
% $$
% which again form a direct sum and again are all stable by $J$... Again $J|_{E_0}$ is always `stable' (all eigenvalues with non-positive real part), so that $J$ is `stable' iff ALL $J|_{E_\ell}$ are `stable'.} \done
\end{proof}

\begin{proposition}[Stability of symmetric equilibria]\label{prop:stabcplt}
  Let $M = ([N], (M_{i j})_{i \neq j})$ be a migration graph such that $M_{i
  j} = m > 0$ for all $i \neq j$, and $P^\text{eq} = (P^\text{eq}_{i j})_{i \neq j}$ a symmetric equilibrium for the system of genetic proximities (\ref{eq:LinSysGenProx}), i.e., verifying $P^\text{eq}_{i j} = p^\text{eq}>0$ for all $i \neq j$.
  Then, 
  
\begin{enumerate}
    \item $P^\text{eq}$ is solution to the equation 
    \begin{align}
        \varphi (p^\text{eq}) = h(p^\text{eq}) (1 - p^\text{eq}) - \frac{\mu}{m} p^\text{eq}  = 0 \label{eq:EqCplt}
  \end{align}
    \item $P^\text{eq}$ is (locally) stable iff
  \begin{align}
        \varphi'(p^\text{eq}) = h' (p^\text{eq}) (1 - p^\text{eq}) - h(p^\text{eq}) - \frac{\mu}{m} <0 \label{eq:stabCplt}
  \end{align}
\end{enumerate}
  
\end{proposition}

\begin{proof}
 From (\ref{eq:LinSysGenProx}), we obtain that any symmetric equilibrium verifies
 \begin{align*}
     0 = 2mh(p^\text{eq}) (1 - p^\text{eq}) + 2(N-2)mh(p^\text{eq})p^\text{eq} - p^\text{eq} (2(N-2)mh(p^\text{eq}) + 2\mu)
 \end{align*}
  Thus the first statement.
  
  The Jacobian $J \assign J_{\vec{F}} (p^\text{eq})$ of $\vec{F}$ can be computed to
  \[ \frac{\partial \vec{F} (p^\text{eq})_{i j}}{\partial P_{i j}} = 2 m h'
     (p^\text{eq}) (1 - p^\text{eq}) - 2 (N - 1) m h (p^\text{eq})
     - 2 \mu =2m\varphi'(p^\text{eq})-2 (N - 2) m h (p^\text{eq}), \]
  for the diagonal terms, and
  \[ \frac{\partial \vec{F} (p^\text{eq})_{i j}}{\partial P_{k i}} = m h
     (p^\text{eq}), \]
  if $k \neq i, j$. Finally, $\frac{\partial \vec{F} (p^\text{eq})_{i j}}{\partial
  P_{k l}} = 0$ otherwise. In particular,
  we remark that we can write
  \begin{equation*}
  J =  2m\varphi'(p^\text{eq}) \cdot \textbf{I} + A, 
  \end{equation*}
  where $\textbf{I}$ is the identity matrix whose rows and columns are indexed by the $N(N-1)/2$ unordered pairs $(ij)$ for $i\not=j$ and $A$ is the transition rate matrix of the Markov chain that jumps from $(ij))$ to $(ik)$ and to $(kj)$ for any $k\not=i,j$ (which make $2(N-2)$ possible transitions) at the same rate $mh(p^\text{eq})$. Again, it is known (see, for instance, \cite{stroock2013introduction}), that the eigenvalue of $A$ with largest real part is given by $0$. The stability condition follows.
\end{proof}

\begin{figure*}[htbp]
    \centering
    \includegraphics[width=\textwidth]{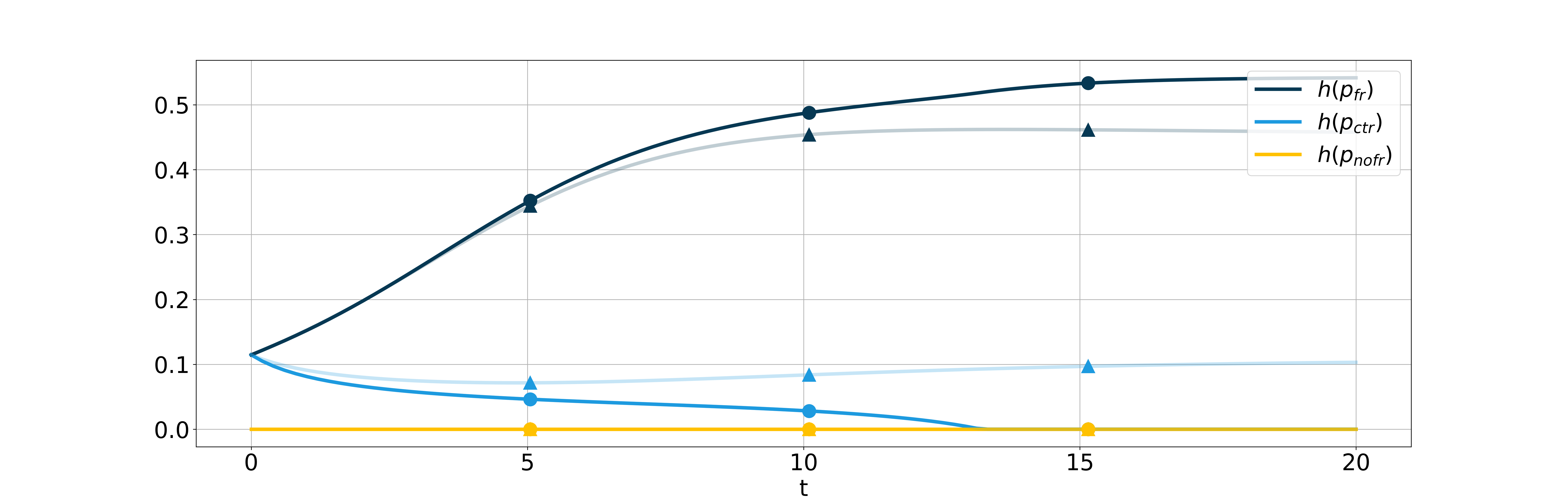}
    \captionsetup{font=small}
    \caption{Collapse of intransitive friendship equilibria for large $N$ (see Fig. \ref{fig:MigGraphs}, (a2)). The solid lines correspond to $h$ evaluated at the three different genetic proximities in our system, namely the proximity between two populations at the outer points of the same triangle ($p_\text{fr}$), between the center population and a triangle population ($p_\text{ctr}$), and between two different-triangle populations ($p_\text{nofr}$). We used a step-feedback function similar to $h_2$ in Fig. \ref{fig:asym_N}. The effective migration structure at equilibrium is given by the friendship graph (see Fig. \ref{fig:MigGraphs}, (a2)), and we used $m = 0.5,\mu = 0.2125$. The different lines of a given color correspond to simulations for different values of $N$: $N=11$ (triangle) and $N=13$ (circle). For $N=11$, the friendship graph is stable. For $N=13$, a speciation event occurs.}
    \label{fig:friendship_sim}
\end{figure*}

In the remaining part of the section, we will focus on the phenomenon of symmetry breaking in the complete and uniform migration graph. We start by considering a case where $[N]$ is split into two sets of vertices $V_1$ and $V_2$. We then consider equilibria $P^\text{eq}$ with three degrees of freedom, namely, the genetic proximity within $V_1$ (denoted by $p_1$), the genetic proximity within $V_2$ (denoted by $p_2$), and the genetic proximity between $V_1$ and $V_2$ (denoted by $p_{inter}$).

\begin{proposition}[Symmetry breaking \RNum{1}]\label{prop:sym_break_1}
    Let $M = ([N], (M_{i j})_{i \neq j})$ be a migration graph such that $M_{i j} = m > 0$ for all $i \neq j$. Consider an equilibrium $P^\text{eq}$ with three degrees of freedom $P^\text{eq} = (p_1, p_2, p_{inter})$. Then, $P^\text{eq}$ is solution to the 3-dimensional system of equations
\begin{eqnarray*}
  |V_2| h(p_{inter}) (p_{inter} - p_1) + h(p_1)(1 - p_1) - \frac{\mu}{m} p_1 & = & 0 \, ,\\
  |V_1| h(p_{inter}) (p_{inter} - p_2) + h(p_2)(1 - p_2) - \frac{\mu}{m} p_2 & = & 0 \, ,\\
  \frac{1}{2} \sum_{i = 1}^2 (|V_i| - 1) h(p_{inter})(p_i - p_{inter}) + h(p_{inter})(1 - p_{inter}) - \frac{\mu}{m}p_{inter} & = & 0 \, .
\end{eqnarray*}
\end{proposition}
\begin{proof}
    Follows by construction of the equilibrium, namely, the partition of $P^\text{eq}$ into the symmetry classes $\{P^\text{eq}_{i j} = p_k$ for $(i,j)\in V_k\}$ $(k=1,2)$, and $\{P^\text{eq}_{i j} = p_{inter}$ for $(i,j)\in V_1 \times V_2\}$, and (\ref{eq:LinSysGenProx}).
\end{proof}

Assume now that $h$ has a threshold. We want to show that there exists a stable, intransitive equilibrium in symmetric migration. Consider the friendship equilibrium $P^\text{eq} = (p_{ctr}, p_{fr}, p_{nofr})$ defined in Section \ref{chap:with_threshold}, and $c$ the threshold of the function $h$.

\begin{proposition}[Symmetry breaking \RNum{2}]\label{prop:sym_break_2}
    Let $M = ([N], (M_{i j})_{i \neq j})$ be a migration graph such that $M_{i j} = m > 0$ for all $i \neq j$. Consider a friendship equilibrium $P^\text{eq} = (p_1, p_2, p_{inter})$. Then, $P^\text{eq}$ is solution to the 3-dimensional system of equations
\begin{eqnarray*}
  \frac{(N-3)}{2}h(p_{ctr})(p_{nofr}-p_{ctr})+ \frac{h(p_{ctr})}{2}(2-3p_{ctr} + p_{fr}) - \frac{\mu}{m}p_{ctr} & = & 0 \, ,\\
  (N-3) h(p_{nofr})(p_{nofr} - p_{fr}) + h(p_{ctr})(p_{ctr} - p_{fr}) + h(p_{fr})(1-p_{fr}) - \frac{\mu}{m} p_{fr} & = & 0 \, ,\\
  h(p_{ctr}) (p_{ctr} - p_{nofr}) + h(p_{nofr})(1 - p_{fr} + 2p_{fr}) - \frac{\mu}{m} p_{nofr} & = & 0 \, .
\end{eqnarray*}
\end{proposition}
\begin{proof}
    Same argument as in the proof of Proposition \ref{prop:sym_break_1}.
\end{proof}

\begin{remark}\label{rem:sym_break}
We note that the previous two equilibria cease to exist for large $N$. In fact, consider the asymmetric equilibrium of Proposition \ref{prop:sym_break_1}. We deduce from equations 1 and 2 that for large $N$, we need to have $h(p_{inter})(p_{inter} - p_k) \propto N^{-1}$, for $k=1,2$. Therefore, the two population groups $V_1$ and $V_2$ either become reproductively isolated from each other $(h(p_{inter}) \rightarrow 0)$, or the equilibrium becomes uniform $(p_{inter} - p_k \rightarrow 0)$. The same argument allows us to deduce that there can only be a finite number of asymmetric equilibria for the equilibrium in Proposition \ref{prop:sym_break_2}. Fig. \ref{fig:friendship_sim} reveals that this collapse of asymmetry can occur with $N$ as little as 13.
\end{remark}

\section{Additional simulations for fluctuating migration networks}

\begin{figure*}[htbp]
    \centering
    \includegraphics[width=\textwidth]{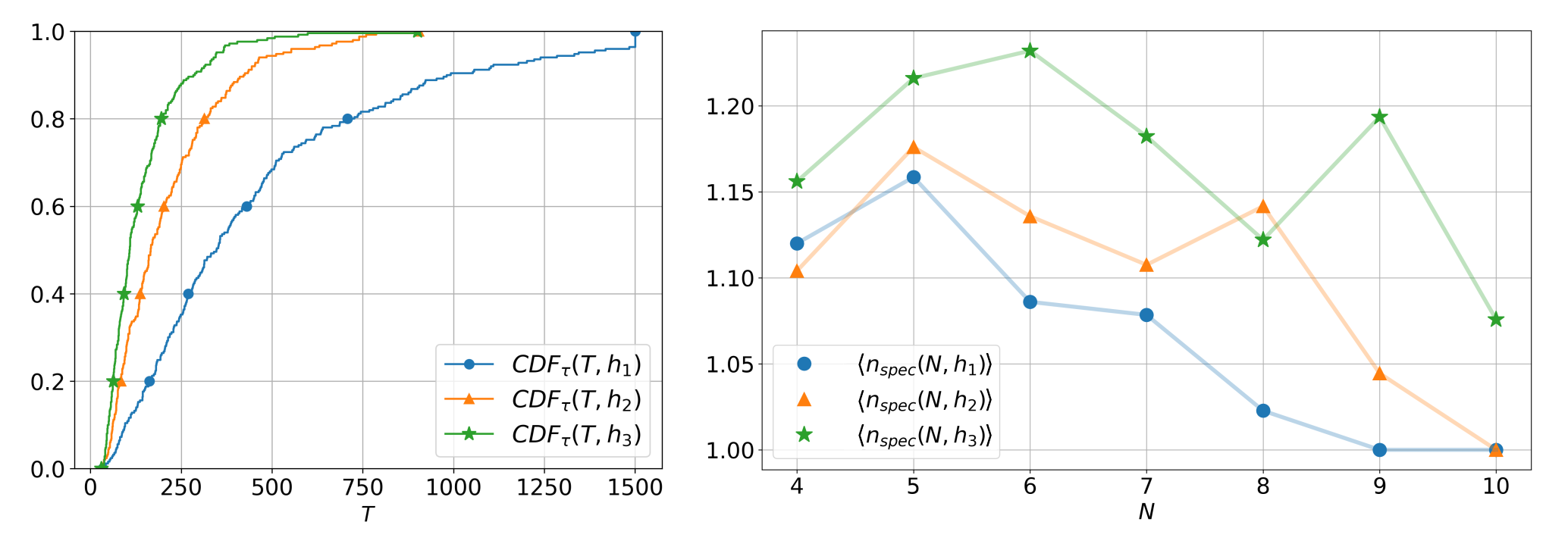}
    \captionsetup{font=small}
    \caption{Dependence of speciation time distribution on feedback regime, and mean number of detaching populations upon speciation (over 250 runs). We considered dynamically changing migration rates updated according to exponential clocks and resampled independently from a (rescaled) Beta distribution $m_\T{rate} \cdot \beta(0.5,0.5)$, with $m_\T{rate} = 1.675$. On the left, we plotted the empirical cumulative distribution function of the time to speciation for $N=5$ and different feedback functions given by $h_1(x) = x^{2.5}, h_2(x) = x^{2.75}, h_3(x) = x^{3}$. On the right, we plotted the mean number of detaching populations upon speciation. Further, we chose $\mu = 0.1, m = \mathbb{E}[M_{ij}] = 0.8375, \theta = 1$. 
    }
    \label{fig:spec_prob_nbrs}
\end{figure*}

\begin{remark}
    Note that the pronounced difference in the distribution of the speciation time w.r.t. the feedback regime (see left panel in Fig. \ref{fig:spec_prob_nbrs}) is in stark contrast with the small distances between the feedback functions ($\leq 0.01$ in the $L^\i$-distance). This indicates a high sensitivity of the speciation time w.r.t. the feedback regime. Further, we observe that the number of detaching populations upon speciation is typically one (see right panel in Fig. \ref{fig:spec_prob_nbrs}). As $N$ increases, the mean number of detaching populations decreases to one, as conjectured in section \ref{sec:fluct}.
\end{remark}

\begin{figure*}[htbp]
    \centering
    \includegraphics[width=.8\textwidth]{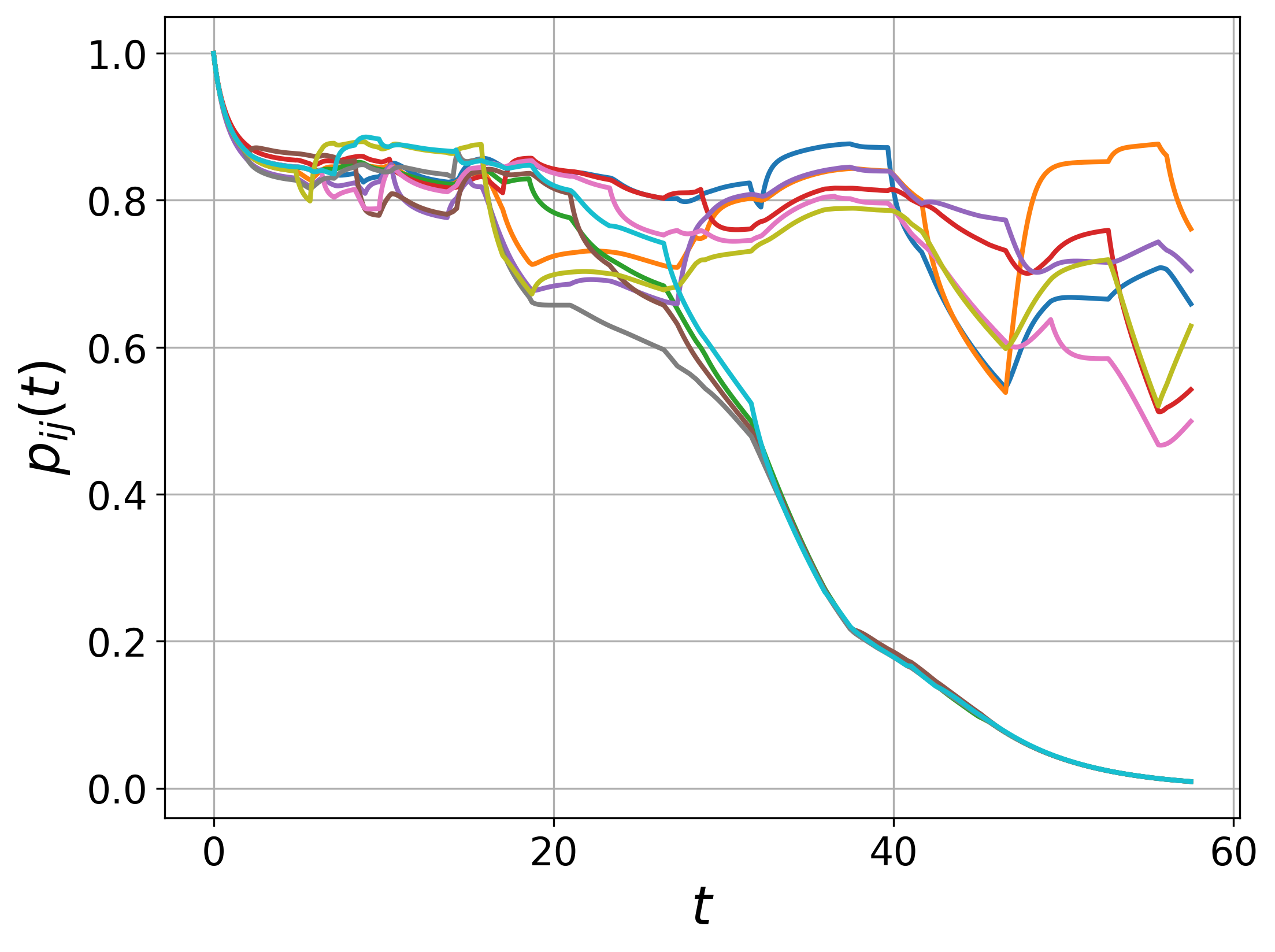}
    \captionsetup{font=small}
    \caption{Realization of ODE (\ref{eq:ODE}) in fluctuating migration networks. We considered dynamically changing migration rates updated according to exponential clocks and resampled $iid$ according to a rescaled Beta distribution, given by $m_\T{rate} \cdot \beta(0.5,0.5)$, with $m_\T{rate} = 1.675$. We plotted the genetic proximities $p_{ij}(t)$ over time. In this example, the speciation event involves one population $i_0$ detaching from the species complex, and thus $p_{i_0 j} (t)\ra 0$ for all $j\neq i_0$. Here, we chose $N = 5, h(x) = x^{2.75}, \mu = 0.1, m = \mathbb{E}[M_{ij}] = 0.8375, \theta = 1$.}
    \label{fig:fluct_mig_dyn}
\end{figure*}

\end{document}